\documentclass[12pt,letterpaper]{article}

\usepackage[T1]{fontenc}
\usepackage{lmodern}
\usepackage[utf8]{inputenc}
\usepackage{geometry}
\usepackage{ragged2e}
\usepackage{nicefrac}
\usepackage{booktabs}
\usepackage{amsfonts}
\usepackage{dsfont}
\usepackage{cancel}
\usepackage{mathtools}
\usepackage{datetime}
\usepackage{bm}
\usepackage{paralist}
\usepackage{changepage}
\usepackage{amsthm}
\usepackage{verbatim}
\usepackage{amsmath,amssymb}
\usepackage{graphicx}
\usepackage{emptypage}
\usepackage{mathabx}
\usepackage{newlfont}
\usepackage{tikz-cd}
\usepackage{tikz}
\usepackage[all,cmtip]{xy}
\usepackage[bottom]{footmisc}
\usepackage[titletoc,title]{appendix}
\usepackage{chngcntr}
\usepackage{apptools}
\usepackage{listings}
\usepackage{color}
\usepackage{sgame, tikz} 

\usepackage{kpfonts}  
\usepackage{natbib}
\AtAppendix{\counterwithin{lem}{section}}
\usepackage{caption}
\usepackage{subcaption}
\usepackage{float}
\usepackage{thmtools,thm-restate}
\usepackage{epigraph}
\usepackage{xr-hyper}
\usepackage[colorlinks=true,allcolors=blue]{hyperref}

\makeatletter
\renewcommand\@makefnmark{\hbox{\@textsuperscript{\normalfont\color{black}\@thefnmark}}}
\makeatother
\DeclareMathOperator*{\argmax}{\arg\!\max}
\DeclareMathOperator*{\argmin}{\arg\!\min}
\newcommand*\supp{\mathrm{supp}}
\usepackage{wasysym}
\hypersetup{colorlinks,linkcolor=blue,allcolors=blue}
\usepackage{sgamevar}
\hypersetup{urlcolor=blue}

\theoremstyle{plain} 
\newtheorem{cor}{Corollary} 
\newtheorem{prop}{Proposition}
 
\newtheorem{theorem}{Theorem}
\newtheorem{lemma}{Lemma}
\theoremstyle{definition}

\allowdisplaybreaks
\makeatletter
\renewenvironment{proof}[1][\proofname]{
  \par\pushQED{\qed}\normalfont
  \topsep6\p@\@plus6\p@\relax
  \trivlist\item[\hskip\labelsep\bfseries#1\@addpunct{.}]
  \ignorespaces
}{
  \popQED\endtrivlist\@endpefalse
}
\makeatother
\theoremstyle{remark} 
\newtheorem{rmk}{Remark}

\DeclareMathOperator{\E}{\mathds{E}}
\renewcommand{\P}{\mathds{P}}
\newcommand{\Pc}{\mathcal{P}}

\newcommand{\C}{\mathcal{C}}

\newcommand{\R}{\mathds{R}}
\newcommand{\Q}{\mathds{Q}}
\newcommand{\cQ}{\mathcal{Q}}
\newcommand{\N}{\mathds{N}}

\newcommand{\K}{\mathcal{K}}

\newcommand{\Exp}{\mathrm{exp}}

\newcommand{\ub}{\overline}

\newcommand{\ind}{\mathds{1}}

\newcommand{\Ext}{\mathrm{ext}}

\renewcommand{\epsilon}{\varepsilon}
\renewcommand*\d{\mathop{}\!\mathrm{d}}

\usepackage{setspace}

\theoremstyle{definition}

\reversemarginpar

\newdate{draftdate}{02}{10}{2023}
\usdate

\usepackage[nameinlink]{cleveref}
\crefname{manualasm}{assumption}{assumptions}
\crefname{taggedtheoremx}{theorem}{theorems}
\crefname{taggedcorollaryx}{corollary}{corollaries}
\crefname{cor}{corollary}{corollaries}
\crefalias{prop}{proposition}
\crefname{claim}{claim}{claims}
\crefname{ex}{example}{examples}
\crefname{defn}{definition}{definitions}
\crefname{rmk}{remark}{remarks}
\crefname{alg}{algorithm}{algorithms}
\crefname{lem}{lemma}{lemma}
\usepackage{xparse}
\makeatletter
\newlength{\eqtaghshift}
\newlength{\eqtagvshift}
\NewDocumentCommand{\eqtag}{ O{\eqtaghshift} O{\eqtagvshift} m m }{
  \hypertarget{#4}{}
  \vadjust{
    \nobreak
    \smash{
      \raisebox{#2}{
        \hbox to 0pt{
          \hspace*{#1}
          \hbox to \textwidth{\hfill \textnormal{(#3)}}
          \hss
        }
      }
    }
  }
}

\makeatother

\DeclareRobustCommand{\weqref}[2]{
  \leavevmode\nobreak\hbox{\hyperlink{#2}{\textup{(#1)}}}
}

\begin{document}
\title{\textbf{Stochastic Optimization and Coupling}\thanks{We thank Ryota Iijima, Ravi Jagadeesan, Jesse Shapiro, Ludvig Sinander, Andrzej Skrzypacz, and Tomasz Strzalecki for helpful comments and suggestions.}
}
\author{Frank Yang\thanks{Department of Economics, Harvard University. Email: fyang@fas.harvard.edu.} \and Kai Hao Yang\thanks{School of Management, Yale University. Email: kaihao.yang@yale.edu.}}
\date{\today
}
\maketitle
\begin{abstract}
We study optimization problems in which a linear functional is maximized over probability measures that are dominated by a given measure according to an integral stochastic order in an arbitrary dimension. We show that the following four properties are equivalent for any such order: (i) the test function cone is closed under pointwise minimum, (ii) the value function is affine, (iii) the solution correspondence has a convex graph with decomposable extreme points, and (iv) every ordered pair of measures admits an order-preserving coupling. As corollaries, we derive the extreme and exposed point properties involving integral stochastic orders such as multidimensional mean-preserving spreads and stochastic dominance. Applying these results, we generalize Blackwell's theorem by completely characterizing the comparisons of experiments that admit two equivalent descriptions---through instrumental values and through information technologies. We also show that these results immediately yield new insights into information design, mechanism design, and decision theory.
\\


\end{abstract}
\setcounter{page}{1}
\newpage

\addtocontents{toc}{\protect\setcounter{tocdepth}{2}} 
\newpage

\section{Introduction}

Many economic problems involve comparing or choosing probability measures. A celebrated example is the Blackwell order, which compares the distributions of posterior beliefs to understand which posterior distributions result from a more informative experiment that uniformly improves the payoff of a decision maker for all decision problems. The Blackwell order is an example of \textit{\textbf{integral stochastic orders}} where the comparisons are defined with respect to a set of test functions, which in Blackwell's case is the set of convex functions representing the indirect utility functions from decision problems.

We study abstract optimization problems involving probability measures that are dominated by a given measure according to an integral stochastic order. Formally, we aim to understand the properties of 
\[\max_{\nu \preceq_\C \mu} \int f(x) \nu(\d x)\,,\]
where $\mu$ is a given measure on some abstract set $X$, $\C$ is a cone of test functions that defines an integral stochastic order, and $\nu$ is the endogenous measure being chosen. We obtain new structural properties of this optimization problem. As corollaries, we derive the extreme and exposed point properties involving integral stochastic orders such as multidimensional mean-preserving spreads and stochastic dominance. Applying these results, we generalize Blackwell's theorem by completely characterizing the comparisons of experiments that admit two equivalent descriptions---through instrumental values and through information technologies. We also show that these results immediately yield new insights into information design, mechanism design, and decision theory.

\paragraph{Abstract Results.}\hspace{-2mm}Before discussing the economic applications, we start by describing our abstract results. Let $V^\star_f(\mu)$ denote the \textit{\textbf{value function}} of the optimization problem and $X^\star_f(\mu)$ denote its \textit{\textbf{solution correspondence}}. It is easy to see that $V^\star_f(\mu)$ is a concave functional and $X^\star_f(\mu)$ is convex-valued, but in general both objects can be very complex. Our results pertain to necessary and sufficient conditions on the stochastic order under which the optimization problem drastically simplifies and the stochastic comparison itself yields a tractable structure. 

We say that the cone of test functions $\C$ is \textit{\textbf{min-closed}} if it is closed under pointwise minimum, i.e., $\min\{g_1, g_2\} \in \C$ for $g_1, g_2 \in \C$. We say that the stochastic order $\preceq_\C$ admits \textit{\textbf{order-preserving couplings}} if for any $\nu \preceq_\C \mu$ there exists a Markov kernel $P$ such that $\nu = P * \mu$ and $P * \delta_x \preceq_\C \delta_x$ for all $x$, i.e., there exists a Strassen-type coupling.\footnote{As we explicitly define later, we use the notation $P*\mu$ to denote the measure obtained by applying the transition kernel $P$ to $\mu$. That is, $P*\mu(\d y):=\int P(\d y \mid x)\mu(\d x)$.}

Our key structural result is an equivalence theorem (\Cref{thm:main}) that establishes the following four-way equivalence: 
\begin{itemize}
    \item[(a)] $\C$ is min-closed.
    \item[(b)] $V^\star_f(\mu)$ is affine for all $f$. 
    \item[(c)] $\preceq_\C$ admits order-preserving couplings. 
    \item[(d)] $X^\star_f(\mu)$ has a convex graph whose extreme points are ``decomposable'' for all $f$.
\end{itemize}
By ``decomposable'' extreme points, we mean that every extreme point $(\mu, \nu)$ of the graph of $X^\star_f$ is one where $\mu$ is an extreme point of its convex domain, and $\nu$ is an extreme point of $X^\star_f(\mu)$ given $\mu$---we refer to such a graph as a \textit{\textbf{trapezoid graph}}, as illustrated by \Cref{fig:trapezoid}. 

\begin{figure}
\centering
\tikzset{
solid node/.style={circle,draw,inner sep=1.25,fill=black},
hollow node/.style={circle,draw,inner sep=1.25}
}
\begin{tikzpicture}[scale=8]
\fill [blue!30!white] (0,0.2)--(0.5,0.1)--(0.5,0.4)--(0,0.3); 
\draw [<->, ultra thick] (0,0.6) node (yaxis) [above] {\footnotesize$X^\star_f(\mu)$}
        |- (0.6,0) node (xaxis) [right] {\footnotesize$\mu$};
\draw [very thick] (0,0.2)--(0.5,0.1);
\draw [very thick] (0,0.3)--(0.5,0.4);
\draw (0,0) node [below=2pt] {\footnotesize $0$}; 
\draw (0.5,0) node [below=2pt] {\footnotesize $1$}; 
\draw (0,0.3) node [solid node] {};
\draw (0,0.2) node [solid node] {};
\draw (0.5,0.1) node [solid node] {};
\draw (0.5,0.4) node [solid node] {};
\draw [very thick] (0.5,0.1)--(0.5,0.4);
\end{tikzpicture}
\caption{Trapezoid Graph. $X=\{0,1\}$.} 
\medskip
    \small
    \justifying
    \noindent 
    The colored area represents the graph of a solution correspondence $X^\star_f$. The extreme points of this set are pairs $(\mu,\nu)$, where $\mu$ is an extreme point of $\Delta(X)$ and $\nu$ is an extreme point of $X^\star_f(\mu)$.
\label{fig:trapezoid}
\end{figure}
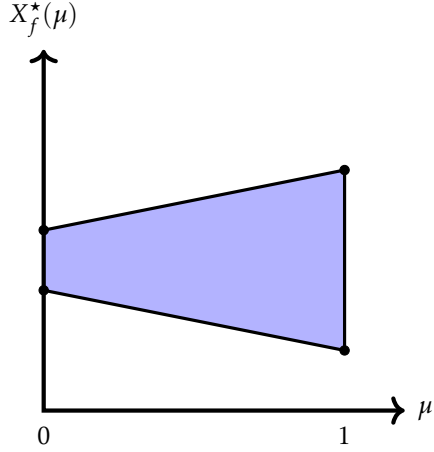

The equivalence theorem immediately yields two sets of applications. The first set of applications derives properties of stochastic orders defined by min-closed test functions, and then derives consequences about the value function and the solution correspondence---in particular, the result shows that (i) we can solve such problems in a pointwise fashion using an appropriate $\C$-envelope of $f$ (\Cref{cor:pointwise}), (ii) the extreme points of a chain of ordered measures decompose (\Cref{cor:chain}), (iii) every nested optimization problem where the second mover chooses a dominated measure according to the order admits a decomposable extreme-point characterization, and when the first mover's objective is linear, it reduces to a linear optimization problem (\Cref{cor:nested-optimization}), and (iv) every extreme and exposed point of the set of measures dominated by a given measure (``orbits'') can be characterized in a pointwise fashion by a unique order-preserving coupling (\Cref{prop:exposed}). 

To illustrate these results, as examples, we consider stochastic orders defined by multidimensional mean-preserving spreads and stochastic dominance. In particular, since concave functions are closed under pointwise minimum, we immediately derive the exposed point structure of the $\mathrm{MPS}(\mu)$ orbit given a fixed measure $\mu$ (\Cref{prop:mps}), which nests the one-dimensional case studied in \citet*{kleiner2021extreme} as a special case. Moreover, since convex functions are \textit{not} closed under pointwise minimum, by \Cref{thm:main}, we also show that the $\mathrm{MPC}(\mu)$ orbit (mean-preserving contraction) cannot admit a structurally similar characterization, which provides a precise explanation about the difference between these two orbits observed by \citet*{kleiner2021extreme}. In contrast, the stochastic orbit defined by stochastic dominance upper bound $\mathrm{LSD}(\mu)$ and that defined by stochastic dominance lower bound $\mathrm{HSD}(\mu)$ behave completely symmetrically because both nondecreasing functions and nonincreasing functions are cones that satisfy the min-closure property---in particular, just like $\mathrm{MPS}(\mu)$, our structural results immediately yield a characterization of the exposed points for these two orbits (\Cref{prop:fosd}), which nests the one-dimensional case studied in \citet{yang2024monotone}. 

Since \Cref{thm:main} is an equivalence theorem, besides applying it for optimization problems, powerful applications of the result are axiomatic in nature---our second set of applications takes a structural property of stochastic orders as given and characterizes the set of stochastic orders satisfying the structural property. In particular, the $(a) \implies (c)$ direction of \Cref{thm:main} is an immediate consequence of the well-known Strassen's theorem (\citealt{Strassen1965}).  \Cref{thm:main} shows that Strassen's theorem has a converse---the only way to admit order-preserving couplings in the Strassen sense is to have the min-closure property for the test functions, which is also equivalent to the value function being affine in the stochastic optimization problem. Given the wide application of Strassen's theorem in various settings, these equivalence results are extremely helpful for delineating the limits of possible stochastic comparisons that satisfy desirable properties such as the Blackwell order. To further illustrate, we now discuss in detail how we apply the structural result to fully characterize comparisons of experiments that generalize Blackwell's theorem by admitting two equivalent descriptions in the Blackwell sense.

\paragraph{Consistent Comparisons of Experiments.}\hspace{-2mm}We study the comparisons of posterior belief distributions. Let $X=\Delta(\Omega)$ denote the simplex of beliefs where $\Omega$ is some state space. The Blackwell order $\preceq_B$ is an order over the distributions of posterior beliefs $\mu \in \Delta(X)$. The celebrated Blackwell theorem (\citealt{Blackwell1951Comparison}) shows that the Blackwell order has two equivalent descriptions: 
\begin{itemize}
    \item[(i)] A posterior distribution higher in the Blackwell order leads to a higher expected value for any convex indirect utility function. \qquad \qquad  \qquad   \quad     \textbf{(Value Description)}
    \item[(ii)] A posterior distribution higher in the Blackwell order can be obtained from the other distribution by a martingale transition kernel. \quad        \textbf{(Information Description)}
\end{itemize}
Thus, the Blackwell order $\preceq_B$ on $\Delta(X)$ can be thought of as either comparing belief distributions according to the instrumental value they generate for solving a class of problems, or comparing belief distributions according to whether one can be obtained from the other via additional information. 

We characterize the set of all orders $\preceq$ on $\Delta(X)$ that admit such a dual representation in the Blackwell sense. Consider any cone of indirect utility functions $\C$, which may or may not be a subset of convex functions, representing the payoffs obtained given posterior beliefs.\footnote{For example, a subset of convex functions can be interpreted as a subset of decision problems, and a superset of convex functions can be interpreted as including additional \textit{persuasion} problems.} Any such $\C$ via the (\textbf{Value Description}) defines a comparison $\preceq_\C$ on the belief distributions. Now, consider any set of belief transition kernels $\mathcal{P}$, which may or may not be a subset of martingales, representing the feasible belief transitions upon observing new information. Any such $\Pc$ via the (\textbf{Information Description}) defines a comparison $\preceq^{\Pc}$ on the belief distributions as well. 

We say that an order $\preceq$ comparing distributions of beliefs is \textit{\textbf{Blackwell-consistent}} if there exist some $\C$ and some $\Pc$ such that $\preceq = \preceq_\C = \preceq^{\Pc}$, in which case we call the dual representation $(\C, \Pc)$ a \textit{\textbf{consistent pair}}. A Blackwell-consistent order has the desirable property that any description of information (through $\Pc$) has a dual representation in terms of its instrumental values (through $\C$). Applying \Cref{thm:main}, our second main result (\Cref{thm:main2} and \Cref{thm:main-compare}) characterizes all Blackwell-consistent orders---it turns out that the following are equivalent: 
\begin{itemize}
    \item[(a)] $\preceq$ is Blackwell-consistent.
    \item[(b)] $\preceq = \preceq_\C$ for some $\C$ that is closed under pointwise maximum.  
    \item[(c)] $\preceq = \preceq^\Pc$ for some $\Pc$ that is closed under kernel composition.
\end{itemize}
Moreover, we also characterize the set of consistent pairs $(\C, \Pc)$. We say that a pair $(\C, \Pc)$ is \textit{\textbf{Blackwell-invariant}} if (i) $\C$ is the set of indirect utility functions for which the payoff weakly increases under any transition kernel from $\Pc$ starting from any prior, and (ii) $\Pc$ is the set of transition kernels under which the payoff weakly increases for any indirect utility function from $\C$ under any prior. \Cref{thm:main-compare} shows that $(\C, \Pc)$ is a consistent pair if and only if $(\C, \Pc)$ is Blackwell-invariant, i.e., they form a \textit{\textbf{fixed point}} when iterating on the \textbf{(Value Description)} and \textbf{(Information Description)}. As a consequence, we show that every Blackwell-consistent order admits a \textit{\textbf{unique}} pair of descriptions $(\C, \Pc)$, exactly pinned down by Blackwell invariance.

If the agent updates according to Bayes' rule, then the test functions must at least include all the affine functions (action-independent payoffs). Then, for a consistent comparison in our sense, the closure under pointwise maximum implies that the test functions must include all convex functions, and hence the order can only be a \textit{\textbf{strengthening}} of the Blackwell order---the Blackwell order is the weakest Bayes-plausible Blackwell-consistent order. Any strict weakening of the Blackwell order such as the Lehmann order (\citealt{Lehmann}) cannot be Blackwell-consistent while maintaining Bayes plausibility. This shows that the Blackwell order is, in fact,  necessary for consistent comparisons under Bayes' rule. A natural question is whether Bayes' rule itself is necessary to admit a consistent order if one wants to compare all experiments that are comparable in the Blackwell-garbling sense. We formulate this question among all prior-dependent systematic distortion updating rules and show that an updating rule can yield a consistent comparison in the Blackwell sense if and only if it is a \textit{\textbf{divisible updating rule}} (\citealt{cripps2018divisible}) which is exactly Bayes' rule under a homeomorphic transformation. Together, these results show that both Bayes' rule and the Blackwell order itself are almost necessary to admit two equivalent descriptions if we want to compare all Blackwell-comparable experiments. 

Whenever an order is consistent, by \Cref{thm:main}, it also implies the stochastic optimization problem yields a tractable structure. Exploiting this implication, we also show how our characterizations also deliver insights into constrained and non-Bayesian information design problems. Indeed, under Bayes' rule, a \textit{\textbf{constrained information design}} problem defines a subset of martingales $\mathcal{P}$ achievable with the constrained technology---if the subset $\mathcal{P}$ is composition-closed, then by \Cref{thm:main} and \Cref{thm:main2}, we immediately obtain an envelope characterization. The constrained envelope is defined by the set of test functions $\C$ that is consistent with $\Pc$ and is weakly below the concave envelope in the unconstrained case. Using the envelope characterization, we show that if constrained information design is beneficial whenever unconstrained information design is beneficial (i.e., the constraints are not too limiting), then the optimal constrained signal \textit{cannot} be less informative (in the Blackwell sense) than the optimal unconstrained signal, and is always \textit{more} informative (in the Blackwell sense) with two states. An example of a composition-closed technology is privacy-preserving signals (\citealt{strack2024privacy}), for which we immediately obtain a new optimality characterization.  In the case of \textit{\textbf{non-Bayesian information design}}, we show that for any prior-dependent systematic distortion updating rule (\citealt{de2022non}), dynamic information design in multiple stages, unlike the Bayesian case, always yields a strictly higher payoff under \textit{some} objective of the designer, unless the updating rule is a divisible updating rule in which case a one-shot signal is sufficient. 

\paragraph{Nested Optimization and Stackelberg Principals.}\hspace{-2mm}For any order $\preceq_\C$ where $\C$ satisfies our min-closure property, as a consequence of \Cref{thm:main}, even \textit{\textbf{nested optimization}} problems admit a simple structure due to the trapezoid graph property. In our final set of applications (\Cref{sec:stackelberg}), we consider a class of problems that we call \textit{\textbf{Stackelberg-principal}} problems where (i) the leading principal first selects a measure $\mu$ from a compact convex set $M$, and (ii) the second principal then selects a measure $\nu$ dominated by $\mu$ according to some order $\preceq_\C$. Whenever $\preceq_\C$ satisfies min-closure, every leader-optimal equilibrium of this game can be succinctly characterized, and there always exists one where $\mu^\star$ is an extreme point of $M$ and $\nu^\star$ is an extreme point of the orbit given $\mu^\star$. 

We then illustrate the consequences of this general characterization via a series of examples arising from information design, decision theory, and mechanism design. First, we consider the \textit{\textbf{sequential persuasion}} problem (\citealt{LiNorman2021SequentialPersuasion}) and show that there always exists a sequential extreme equilibrium where every sender chooses a belief distribution that is an extreme point of the mean-preserving spread of the belief distribution chosen by the previous sender. In such an equilibrium, unlike the silent equilibria in \citet{LiNorman2021SequentialPersuasion} and \citet{wu2023sequential}, every sender sends a signal, and along any history, the signals sent have the same structure \textit{as if} there were no following senders. Second, we consider the \textit{\textbf{robust persuasion}} problem (\citealt{DworczakPavan2022RobustPersuasion})---there the first principal is the sender and the second principal is the adversarial nature. We illustrate how to apply our results there and show connections to their separation theorem which characterizes the robust solutions that are best-case optimal among worst-case optimal ones. Third, we consider the model of \textit{\textbf{objective ambiguity aversion}} (\citealt{olszewski2007preferences,Ahn2008AmbiguityWithoutStateSpace}) where ambiguity aversion is modeled via preferences over menus/ambiguity sets of lotteries. We ask when an outside observer by observing choices over lotteries $\mu$ can distinguish between (i) an objectively ambiguity-averse DM who has some ambiguity set $A(\mu)$ and (ii) an expected utility DM who faces no ambiguity. We show that in the case where $A(\mu)$ is described by a stochastic orbit (i.e., $\{\nu: \nu \preceq_\C \mu\}$), the outside observer can distinguish between these two models essentially if and only if $\C$ is not min-closed. In particular, for ambiguity sets of the form where the actual lottery could be a mean-preserving contraction of the reference lottery $\mu$, we can distinguish these two models; however, for ambiguity sets taking the form of  mean-preserving spreads or stochastic dominance shifts, we cannot distinguish between these two models. Lastly, we consider the model of \textit{\textbf{optimal property right}} design (\citealt{DworczakMuir2024}). Here, the leading principal is a designer who posts a menu of lotteries over a good and associated prices, which become a buyer's outside options. The second principal is a seller who contracts with the buyer as in a standard mechanism design framework but must satisfy the type-dependent IR constraints resulting from the designer's initial menu. The designer and the seller have different preferences over the final allocations. However, the main result of \citet{DworczakMuir2024} shows that the optimal property design by the designer is a simple option-to-own menu that consists of one posted price, just like in standard monopoly pricing. \citet{DworczakMuir2024} show this by developing a generalized ironing procedure and observing that the ironing procedure is actually a linear operator of the outside option function, which then leads to a surprising extreme point structure. We show that this argument can be more broadly viewed as the trapezoid property of the nested optimization problem where the designer first chooses a distribution of posted prices, and then the seller chooses another distribution of posted prices subject to giving every type a weakly higher payoff, which defines an integral stochastic order.  

\paragraph{Related Literature.}\hspace{-2mm}This paper is related to several streams of literature. Our main abstract result is closely related to Strassen's theorem and convex duality. Strassen's theorem (\citealt[][]{Strassen1965}) shows that the concave order admits an order-preserving coupling. It is known that the concave order of Strassen's theorem can be further generalized to integral stochastic orders defined by min-closed test functions (\citealt{meyer1966}). Our abstract result expands this to a four-way equivalence. In particular, we show that the existence of order-preserving couplings also implies that the stochastic order must be defined by min-closed test functions, which is the converse to Strassen's theorem. Moreover, we show that these are also equivalent to having affine value functions and trapezoid solution graphs. In that regard, our results are also related to optimization problems with mean-preserving spreads constraints, where the values are characterized by concave envelopes (e.g., \citealt[][]{Kemperman1968,Myerson1981,AumannMaschler1995,Kamenica2011,BeiglbockNutz2014}). Our four-way equivalence explains precisely why some stochastic optimization problems are more tractable than others, due to the equivalence relation between the affine value property/trapezoid graph property and the min-closure property. 

The main characterization leads to an immediate characterization of extreme points of chains of measures ordered by a single integral stochastic order, which generalizes the characterization of \citet{Ciosmak2023} for the convex order. Moreover, we characterize extreme points of multidimensional mean-preserving spread and first-order stochastic dominance orbits, which extends the characterization of \citet*{kleiner2021extreme} and \citet{yang2024monotone} respectively.\footnote{See also \citet{nikzad2023constrained} and \citet{candogan2021optimal} for extreme points of mean-preserving contraction orbits with linear constraints; and \citet{augias2025economicsconvexfunctionintervals} for extreme points of mean-preserving spread intervals.} Our multidimensional MPS orbits characterization also complements \citet{kleiner2024extreme}, who characterize the exposed points of MPC orbits. A one-dimensional probability measure can be equivalently viewed as a monotone function via its CDF; however, this equivalence breaks down for multidimensional probability measures.\footnote{Indeed, recall that a bivariate CDF $F$ must also satisfy $F(x'_1,x'_2) - F(x_1,x'_2) - F(x'_1,x_2) + F(x_1,x_2) \ge 0$ for all $x_1 < x'_1$ and $x_2 < x'_2$.} \citet{yang2025multidimensional} characterize the extreme points of multidimensional monotone functions and show applications to mechanism and information design problems.

Our second main result contributes to the literature on comparisons of experiments pioneered by \citet{Blackwell1951Comparison}. We generalize Blackwell's theorem by completely characterizing all possible comparisons of experiments that admit two equivalent descriptions in the Blackwell sense. Our characterization establishes a five-way equivalence and identifies if-and-only-if conditions under which an integral stochastic order on the marginal distributions admits an equivalent coupling order. As we show, this simultaneously connects Strassen's theorem on the existence of order-preserving couplings, the Skorokhod-type embedding theorems of general Markov processes (\citealt{rost1971stopping}), and the characterizations of optimal values of stopping problems via excessive functions (\citealt{dynkin1963optimal}). Since we identify all orders of experiments that can be equivalently represented by the instrumental values of the experiments, our results are related to other orderings of experiments (e.g., \citealt[][]{Lehmann,chen2025experiments}), as well as attempts to compare experiments and persuasion problems with non-Bayesian updating rules and heterogeneous priors (e.g., \citealt[][]{alonso2016bayesian,de2022non,kobayashi2025dynamic,azrieli2025sequential}). We show that the only non-Bayesian updating rules consistent with Blackwell-style comparisons are exactly the divisible updating rules defined by \citet{cripps2018divisible}---which implies that dynamic information design yields a strict benefit under some objective of the designer if and only if the updating rule is not divisible.\footnote{Divisible updating rules are homeomorphic to Bayes' rule but they are \textit{not} Blackwell-monotone (\citealt{whitmeyerforthcomingBlackwell}). As we discuss, combining \citet{whitmeyerforthcomingBlackwell} and our result, it follows that the garbling-based comparison admits a dual representation if and only if the updating rule is divisible, and the class of dualizing functions is the class of decision problems if and only if it is Bayes' rule.} As we demonstrate, our characterization result also helps to identify constrained information design problems that are particularly tractable and admit a similar belief-based characterization as in the unconstrained problems. These problems include privacy-preserving persuasion (\citealt*{strack2024privacy,he2024private}), as well as sequential sampling and persuasion problems (e.g., \citealt*[][]{Wald,henry2019research,morris2019wald,ball2021experimental,ni2023sequential}). 

In terms of other economic applications, as we demonstrate, many mechanism and information design problems have a nested structure---which we term Stackelberg-principal problems---in which the trapezoid-graph property delivers immediate insights. This includes sequential persuasion with multiple senders (\citealt{LiNorman2021SequentialPersuasion}); robust persuasion with a worst-case Nature who supplies additional information (\citealt{DworczakPavan2022RobustPersuasion}); a particular form of objective ambiguity aversion where the ambiguity sets are stochastic orbits (\citealt{olszewski2007preferences,Ahn2008AmbiguityWithoutStateSpace}); and optimal property right design (\citealt{DworczakMuir2024}). 

The remainder of the paper proceeds as follows. \Cref{sec:model} introduces the notation and the primitives.  \Cref{sec:abstract} presents our main abstract result. \Cref{sec:compare} presents our main application to the comparisons of experiments and its implications for constrained and non-Bayesian information design. \Cref{sec:stackelberg} presents our additional economic applications with nested optimization. \Cref{sec:conclusion} concludes.

\section{Primitives}\label{sec:model}

\paragraph{Notation.}\hspace{-2mm}Let $X$ be a compact Polish space endowed with the Borel $\sigma$-algebra $\mathcal{B}(X)$. Let $C(X)$ denote the set of continuous functions on $X$, endowed with the sup-norm. Let $\Delta(X)$ be the set of probability measures on $X$, endowed with the weak-* topology. 

For any $\mu \in \Delta(X)$ and any measurable Markov kernel $P: X \rightarrow \Delta(X)$, denote by $P*\mu$ the induced measure on $X$. That is,  
\[
[P * \mu] (A):=\int_X P(A\mid y) \mu(\d y) \eqtag[-2.5em]{\textbf{Transition Kernel}}{eq:transition-kernel}
\]
for all $A \in \mathcal{B}(X)$. For any measurable function $g:X \to \R$, denote by $g*P$ the expected value of $g$ under $P$. That is, 
\[
[g*P](x):=\int_X g(y) P(\d y\mid x)\,, \eqtag[-2.5em]{\textbf{Conditional Expectation}}{eq:conditional-expectation}
\]
for all $x \in X$. For any two transition kernels $P_1,P_2$, denote by $P_1 \circ P_2$ the composition of $P_1$ and $P_2$. That is, 
\[
[P_1\circ P_2] (A\mid x):=\int_X P_1(A\mid y)P_2(\d y\mid x)\,, \eqtag[-2.5em]{\textbf{Composition}}{eq:composition}
\]
for all measurable $A \subseteq X$ and for all $x \in X$. For any $x \in X$, let $\delta_x$ be the Dirac measure on $x$. For any convex and compact subset $M$ of $\Delta(X)$, $\Ext(M)$ denotes the set of \textit{\textbf{extreme points}} of $M$, and $\Exp(M)$ denotes the set of \textit{\textbf{exposed points}} of $M$.

\paragraph{Integral Stochastic Orders.}\hspace{-2mm}Let $\C$ be a convex cone of bounded upper semicontinuous functions such that (i) $\C \cap C(X)$ is sup-norm closed, (ii) $\C$ contains all constants, (iii) $\C$ admits a subset of continuous functions that can approximate any function in $\C$ from above,\footnote{That is, for all $g \in \C$, there exists a decreasing sequence $\{g_n\} \subseteq \C \cap C(X)$ such that $\{g_n\}\downarrow g$ pointwise.} and (iv) $\C$ is closed under bounded decreasing pointwise limits.\footnote{That is, if $\{g_n\}\downarrow g$ pointwise and $\{g_n\} \subseteq \C$, then $g \in \C$.}  

We say that $\C$ is \textit{\textbf{min-closed}} if the functions in $\C$ are closed under pointwise minimum: 
\begin{equation*}\label{eq:min-closure}
\min\{g_1,g_2\} \in \C\, \, \,\, \text{ for all } g_1,g_2 \in \C\,. \eqtag[-2.5em][0.1\baselineskip]{\textbf{Min-Closure}}{eq:min-closure}
\end{equation*}
Examples of min-closed $\C$ include nondecreasing, nonincreasing, concave, nondecreasing concave, nonincreasing concave functions (when $X$ carries the relevant order/convex structure). Non-examples include convex functions (they are max-closed). 

For any bounded upper semicontinuous function $f$, define the $\C$-\textbf{\emph{envelope}} of $f$ as 
\[
\overline{f}(x):=\inf\Big\{g(x): g \in \C\,, g \geq f\Big\}\, \tag{\textbf{$\mathcal{C}$-Envelope}}
\]
for all $x \in X$.  It can be readily shown that $\overline{f} \in \C$ if $\C$ is min-closed (see \Cref{lem:envelope} in the appendix).

For any $\mu,\nu \in \Delta(X)$, we say that $\mu$ $\C$-\textbf{\emph{dominates}} $\nu$, denoted by 
\[
\nu \preceq_\C \mu\,,
\]
if for any $g \in \C$, 
\[
\int_X g(x) \nu(\d x) \leq \int_X g(x)\mu( \d x)\,.
\]
Note that any integral stochastic order $\preceq_\C$ is reflexive and transitive, but need not be antisymmetric. Thus, formally, $\preceq_\C$ is a preorder. Throughout the paper, we use the term order to mean a preorder.\footnote{One can also impose antisymmetry without affecting any of our results, in which case the term order means a partial order.} We say that $\preceq_{\C}$ admits \textit{\textbf{order-preserving couplings}} if for every $\nu \preceq_{\C} \mu$, there exists some $P$ such that $\nu = P * \mu$ and $P* \delta_x \preceq_\C \delta_x$ for all $x$. 

\paragraph{Stochastic Optimization.}\hspace{-2mm}Fix any bounded upper semicontinuous function $f: X \to \R$. Our main results study the solutions to the following stochastic optimization problem
\begin{equation*}
\sup_{\nu \in \Delta(X):\, \nu \preceq_\C \mu} \int_X f \d \nu\,. 
\eqtag[-2.5em]{\textbf{Optimization}}{eq:stochastic-optimization}
\end{equation*}
Since $\Delta(X)$ is compact and the orbit set $\{\nu: \nu \preceq_\C \mu\}$ is closed and hence compact (see \Cref{lem:compactK} in the appendix), \weqref{\textbf{Optimization}}{eq:stochastic-optimization} has a solution. 

For any $\mu \in \Delta(X)$, let 
\[
V^\star_f(\mu):=\max_{\nu \in \Delta(X):\, \nu \preceq_\C \mu} \int_X f \d \nu \eqtag[-2.5em]{\textbf{Value Function}}{eq:value}
\]
be the \textit{\textbf{value function}} of \weqref{\textbf{Optimization}}{eq:stochastic-optimization} and let 
\[
X^\star_f(\mu):=\argmax_{\nu \in \Delta(X):\, \nu \preceq_\C \mu} \int_X f \d \nu \eqtag[-2.5em]{\textbf{Solution Correspondence}}{eq:sol}
\]
be the \textit{\textbf{solution correspondence}} of \weqref{\textbf{Optimization}}{eq:stochastic-optimization}. By standard arguments and Berge’s maximum theorem, it is clear that $V^\star_f$ is concave and upper semicontinuous on $\Delta(X)$, and $X^\star_f$ is compact-valued and upper-hemicontinuous on $\Delta(X)$. 

For any closed and convex subset $M$ of $\Delta(X)$, let
\[
G^M_f:=\big\{(\mu,\nu): \mu \in M\,, \nu \in X^\star_f(\mu)\big\} \subseteq \Delta(X) \times \Delta(X) \eqtag[-2.5em][0.4\baselineskip]{\textbf{Solution Graph}}{eq:solgraph}
\]
be the \textit{\textbf{graph}} of the solution correspondence $X^\star_f$ given domain $M$. Note that since $X^\star_f$ is upper-hemicontinuous, $G_f^M$ is compact. 

We say that $X^\star_f$ has the
\textit{\textbf{trapezoid graph}} property if for any closed convex $M$, (i) its graph $G^M_f$ is convex, and (ii) the extreme points of the graph are decomposable:
\[\Ext\big(G^M_f\big)=\Big\{(\mu,\nu): \mu \in \Ext(M)\,, \nu \in \Ext(X^\star_f(\mu))\Big\} \,.\eqtag[-2.5em]{\textbf{Decomposability}}{eq:decomp}\]

\paragraph{Nested Stochastic Optimization.}\hspace{-2mm}In the above problem, we take the reference measure $\mu$ as given, but our results also pertain to a nested optimization problem where $\mu$ is chosen by a first mover. In particular, for any continuous quasi-convex functional $W: G^M_f  \to \R$, consider the problem of selecting a pair of measures from the graph $G^M_f$ to maximize $W$, 
\begin{equation*}
\text{ } \qquad \max_{(\mu,\nu) \in G^M_f} W(\mu,\nu)\,. \eqtag[-2.5em][0.9\baselineskip]{\textbf{Nested Optimization}}{eq:nested-optimization}
\end{equation*}
That is, the first mover has preferences over $(\mu, \nu)$ given by $W$; the second mover has a linear objective given by $f$, with indifferences broken in favor of the first mover.

\section{Abstract Results}\label{sec:abstract}

\begin{theorem}[Equivalence Theorem]\label{thm:main}
The following are equivalent: 
\begin{itemize}
    \item[(a)] $\C$ is min-closed.
    \item[(b)] $V^\star_f$ is affine for all $f$.
    \item[(c)] $\preceq_{\C}$ admits order-preserving couplings.
    \item[(d)] $X^\star_f$ has the trapezoid graph property for all $f$.
\end{itemize}
\end{theorem}

The proof of \Cref{thm:main} can be found in the appendix. Below, we briefly sketch the main steps of the proof.  We prove the equivalence through the order:

\[(a) \implies (b) \implies (c) \implies (d) \implies (a) \,.\]

To prove $(a) \implies (b)$, we use standard arguments in convex duality. In particular, we apply the Fenchel-Rockafellar duality theorem. However, if we were to assume that $f$ is continuous, that every $g \in \C$ is continuous, and that $\C$ is closed under the sup norm, then there is a more direct proof in the form of the minimax theorem. To see this, note that the dual cone of $\C$ is given by 
\[\C^* = \Big\{\eta \in M(X): \int g(x) \eta(\d x) \geq 0 \text{ for all $g \in \C$} \Big\}\,,\]
where $M(X)$ is the space of finite signed Borel measures (equipped with the weak-* topology). The primal problem \weqref{\textbf{Optimization}}{eq:stochastic-optimization} can then be written as: 
\[\sup_{\nu \in \Delta(X)} \int f(x) \nu (\d x) \text{ subject to }  \mu - \nu \in \C^*\,. \eqtag[-2.5em][0.8\baselineskip]{\textbf{Primal}}{eq:primal}\]
By the bipolar theorem, the dual cone of $\C^*$ is $\C$ itself, and hence we can identify the ``multiplier'' of the above constraint as $g \in \C$. Thus, the Lagrangian is 
\[\int f(x) \nu(\d x) + \int g(x) (\mu - \nu)(\d x)\,.\]
In particular, the value of the problem is equal to 
\[V^\star_f(\mu)=\sup_{\nu \in \Delta(X)} \inf_{g \in \C} \int f(x) \nu(\d x) + \int g(x) (\mu - \nu)(\d x)\,,\]
which, by Sion's minimax theorem, is equal to 
\[\inf_{g \in \C} \sup_{\nu \in \Delta(X)} \int f(x) \nu(\d x) + \int g(x) (\mu - \nu)(\d x)\,,\]
which is equivalent to
\[\inf_{g \in \C} \Big\{\int g(x) \mu(\d x) +   \max_{x \in X} [f(x) - g(x)]\Big\}\,.\]
Now, note that since $\C$ includes all constants, and the above value is invariant to shifting $g$ by a constant, the above problem can be written as 
\[\inf_{g \in \C} \int g(x) \mu(\d x) \text{ subject to } g(x)- f(x) \geq 0 \text{ for all $x$}\,. \eqtag[-2.5em][0.8\baselineskip]{\textbf{Dual}}{eq:dual}\]
This proves that there is no duality gap. Now, importantly, note that \weqref{\textbf{Min-Closure}}{eq:min-closure} of $\C$ implies that an optimal solution to the above dual problem is simply $g = \overline{f} \in \C$: Indeed, for any $g \in \C$ and $g \geq f$, by the definition of $\overline{f}$, we have that for every $x \in X$, 
\[\overline{f}(x) \leq g(x)\,,\]
and hence $\int \overline{f}(x) \mu(\d x) \leq \int g(x) \mu(\d x)$. This proves dual attainment, and hence strong duality. Importantly, the dual solution $\overline{f}$ is independent of $\mu$. It follows immediately that 
\[V^\star_f = \int \overline{f}(x) \mu(\d x)\,,\]
which is affine in $\mu$.

The proof of $(b) \implies (c)$ relies on the Krein-Milman theorem. Specifically, we first show that the orbit
\[
\mathcal{K}_\mu:=\{\nu: \nu \preceq_\C \mu\}
\]
is convex and compact. Then, we show that $(b)$ implies that every exposed point $\nu$ of $\mathcal{K}_\mu$ must be such that 
\[
\nu=P*\mu
\]
for a transition kernel $P$ with $P*\delta_x \preceq_\C \delta_x$ for all $x$. Indeed, fix an exposed point $\nu$ of $\mathcal{K}_\mu$ and a continuous function $f$ for which $\nu$ is the unique solution to \weqref{\textbf{Optimization}}{eq:stochastic-optimization}. Since $V^\star_f$ is affine, we have 
\[
\int f \d \nu=V^\star_f(\mu)=\int V^\star_f(\delta_x) \mu(\d x)=\int \left(\max_{\eta: \eta \preceq_\C \delta_x} \int f\d\eta\right)\mu(\d x)\,.
\]
That is, the value of the problem \weqref{\textbf{Optimization}}{eq:stochastic-optimization} is the same as the average value of a family of pointwise optimization problems 
\begin{equation*}
\max_{\eta: \eta \preceq_\C \delta_x} \int f \d \eta \eqtag[-2.5em]{\textbf{Pointwise Optimization}}{eq:pointwise}\,.
\end{equation*}
Using standard measurable selection arguments, we select a transition kernel $P$ where $P(\,\cdot\,\mid x)$ is a solution of the pointwise optimization problem for all $x \in X$. Then, $P*\mu$ must also solve \weqref{\textbf{Optimization}}{eq:stochastic-optimization}. But since $\nu$ is the unique solution by construction, we have $\nu=P*\mu$. With $(c)$ established for exposed points, we can then apply the Krein-Milman theorem for exposed points (\citealt{Kleejr1958}) to represent every element $\nu \in \mathcal{K}_\mu$ as the limit of a sequence $\{\nu_n\}$ of convex combinations of exposed points. For each element $\nu_n$ along the sequence, by linearity, there also exists a transition kernel $P_n$ with $P_n*\delta_x \preceq_\C \delta_x$ for all $x$ such that $\nu_n=P_n*\mu$. By appropriate approximation arguments, $(c)$ then follows after taking limits.

The proof of $(c) \implies (d)$ is based on first principles. It is clear that $(c)$ implies that $V_f^\star$ is affine, which in turn implies that $G_f^M$ has a convex graph. To see that $G_f^M$ has decomposable extreme points, note that for any extreme point $(\mu,\nu)$ of $G^M_f$, clearly $\nu$ is an extreme point of $X^\star_f(\mu)$. To see that $\mu$ must be an extreme point of $M$, suppose the contrary. Then, there exist distinct $\mu_1,\mu_2 \in M$ such that $\frac{1}{2}\mu_1+\frac{1}{2}\mu_2=\mu$. Since $\nu$ is in $X^\star_f(\mu)$, $(c)$ implies that there exists a transition kernel $P$ such that $\nu=P*\mu$, and that $P$ solves \weqref{\textbf{Pointwise Optimization}}{eq:pointwise} for $\mu$-almost every $x \in X$. Now, let $\nu_i:=P*\mu_i$ for $i\in \{1,2\}$. It follows that $\nu_i \preceq_\C \mu_i$ for all $i \in \{1,2\}$, and that---since $P(\,\cdot\,\mid x)$ is a solution of the pointwise maximization problem $\mu$-almost surely and hence $\mu_i$-almost surely as $\mu_i \ll \mu$---we have 
\[
\int f \d \nu_i=V_f^\star(\mu_i)
\]
for $i \in \{1,2\}$. This implies that $(\mu_i,\nu_i) \in G_f^M$. However,  $(\mu,\nu) = \frac{1}{2}(\mu_1,\nu_1)+\frac{1}{2}(\mu_2,\nu_2)$ is a mixture of two distinct elements in $G_f^M$, contradicting the extremality of $(\mu,\nu)$.

Lastly, the proof of $(d) \implies (a)$ uses a separating hyperplane argument. First, we show that $(d)$ immediately implies that $V^\star_f$ is affine for all upper semicontinuous $f$. Then, suppose that $\C$ is not min-closed. Then there exist $g_1,g_2 \in \C$ such that $k:=\min\{g_1,g_2\}$ is not in $\C$. Therefore, there exists a separating hyperplane that separates $k$ and $\C$. Together with the properties of $\C$, we then show that this separating hyperplane corresponds to two probability measures $\mu,\nu \in \Delta(X)$ such that $\nu \preceq_\C \mu$, and 
\[
\int k \d \nu>\int k \d \mu\,.
\]
Then, for any $x \in X$ and for any $\eta \preceq_\C \delta_x$, since $g_1,g_2 \in \C$, 
\[
\int k \d \eta \leq \min\left\{\int g_1 \d \eta, \int g_2 \d \eta\right\} \leq \min\{g_1(x),g_2(x)\}=k(x)\,.
\]
Therefore, $V^\star_k(\delta_x)=k(x)$. But then, since $\nu \preceq_\C \mu$ and $V^\star_k$ is affine, we have 
\[
\int k \d \nu \leq V^\star_k(\mu) =\int V_k^\star(\delta_x) \mu(\d x)=\int k \d \mu< \int k \d \nu\,,
\]
a contradiction. 

\begin{rmk}
It is noteworthy that throughout the proof, we do not use any existing results besides the duality theorem, the Krein-Milman theorem, and the separating hyperplane theorem. In particular, $(a) \implies (c)$ is in fact the celebrated Strassen theorem (\citealt{Strassen1965,meyer1966}). Our approach leads to an alternative proof using duality and the Krein-Milman theorem. More importantly, we also prove that Strassen's theorem has a converse---the existence of order-preserving couplings implies that the stochastic order must be defined by a min-closed test function cone. 
\end{rmk}
\begin{rmk}
 We state \Cref{thm:main} for the optimization problem where the constraint set is $\{\nu: \nu \preceq_\C \mu\}$. Of course, we can also flip the order by considering $\{\nu: \mu \preceq_\C \nu\}$, in which case all results in \Cref{thm:main} continue to hold with the corresponding directional changes: (i) min-closure replaced by max-closure, (ii) the definition of order-preserving coupling replaced by the opposite direction, and (iii) the $\C$-envelope of $f$ that characterizes the value function $V^\star_f$ replaced by the $[-\C]$-envelope of $f$. 
\end{rmk}

\subsection{Consequences under Min-Closure}

\Cref{thm:main} immediately implies a series of structural properties about optimization involving integral stochastic orders satisfying the min-closure property. 

\begin{cor}[Pointwise Optimization and Envelope]\label{cor:pointwise}
For any $\preceq_\C$ where $\C$ is min-closed, \[
V^\star_f(\mu)=\int_X V^\star_f(\delta_x)\mu(\d x) = \int_X \overline{f}(x) \mu(\d x)\,.
\]
\end{cor}

\Cref{cor:pointwise} identifies the affine value function by the $\C$-envelope of $f$ and observes that the original optimization admits a pointwise solution where the solution can be obtained by treating each $x$ separately and maximizing over distributions $\eta_x \preceq_\C \delta_x$. 

\begin{cor}[Extreme Points of Chains]\label{cor:chain}
Consider any $\preceq_\C$ where $\C$ is min-closed and any compact convex subset $M$ of $\Delta(X)$ that is closed under $\preceq_\mathcal{C}$. The following are equivalent: 
\begin{itemize}
    \item[(a)] $(\mu_i)_{i=1}^n$ is an extreme point of the set $\{(\tilde{\mu}_i)_{i=1}^n \in M^n: \tilde{\mu}_{i+1} \preceq_\C \tilde{\mu}_i\}$.
    \item[(b)] $\mu_1$ is an extreme point of $M$ and $\mu_{i}$ is an extreme point of the set $\{\tilde{\mu} \in M: \tilde{\mu} \preceq_\C \mu_{i-1}\}$ for all $i \in \{2,\ldots,n\}$.  
\end{itemize}
\end{cor}

\Cref{cor:chain} shows that the extreme points of a chain of ordered measures $(\mu_i)_{i=1,\dots,n}$ admit a simple decomposable structure whenever the order $\preceq_\C$ satisfies the min-closure property. An example of \Cref{cor:chain} is when $\C$ consists of the set of concave functions on $\R^m$, and $M = \Delta(X)$, in which case we recover the extreme point structure of a pair of measures ordered by the convex order as characterized in  \citet{Ciosmak2023}.

\begin{cor}[Reduction of Nested Optimization]\label{cor:nested-optimization}
For any $\preceq_\C$ where $\C$ is min-closed and any  \weqref{\textbf{Nested Optimization}}{eq:nested-optimization}, there exists $(\mu,\nu)$ solving \weqref{\textbf{Nested Optimization}}{eq:nested-optimization} such that $\mu \in \Ext(M)$ and $\nu \in \Ext(X^\star_f(\mu))$. 
\end{cor}

\Cref{cor:nested-optimization} shows that \weqref{\textbf{Nested Optimization}}{eq:nested-optimization} admits simple solutions where the first mover simply chooses an extreme point from the feasible space $M$ even after taking into account the response by the second mover. In fact, as \Cref{thm:main} reveals, when the first mover's objective $W(\mu, \nu)$ is a linear functional and hence can be written as $\int w_A(x) \mu(\d x) + \int w_B(x) \nu(\d x)$, there exists a modified linear objective $\widetilde{W}$ where $w_B$ is modified to be $\widetilde{w}_B:=w_B*P^\star_f$ and  $P^\star_f(\,\cdot\,\mid x)$ is the pointwise solution for the second mover, such that the first mover can simply optimize the modified linear objective \textit{as if} the second mover did not exist. We describe our results and applications of nested optimization in detail in \Cref{sec:stackelberg}. 

For any $(\mu, \nu)$ with $\nu \preceq_\C \mu$, we say that they are \textit{\textbf{uniquely rationalized}} if they admit a unique order-preserving kernel $P$, i.e., there exists a $\mu$-a.e. unique $P$ such that $\nu = P* \mu$ and $P *\delta_x \preceq_\C \delta_x$ for $\mu$-a.e. $x$.  The following result characterizes structural properties of extreme and exposed points of stochastic orbits for any integral stochastic order satisfying the min-closure property: 

\begin{prop}[Extreme and Exposed Points of Stochastic Orbits]\label{prop:exposed}
For any $\preceq_{\C}$ and any $\mu \in \Delta(X)$, let 
\[\mathcal{K}_\mu:= \Big\{\nu: \nu \preceq_\C \mu \Big\}\,.\]
If $\C$ is min-closed, then:
\begin{enumerate}
    \item [(i)] $\nu\in \Ext(\mathcal{K}_\mu)$ if and only if $(\mu, \nu)$ is uniquely rationalized by an order-preserving kernel $P$ where $P*\delta_x$ is an extreme point of $\mathcal{K}_{\delta_x}$ $\mu$-a.e.
    \item [(ii)] $\nu\in \Exp(\mathcal{K}_\mu)$ if and only if $(\mu, \nu)$ is uniquely rationalized by an order-preserving kernel $P$ where $P*\delta_x$ is an exposed point of $\mathcal{K}_{\delta_x}$ $\mu$-a.e. by the same continuous function $f$.    
\end{enumerate}
\end{prop}

\Cref{prop:exposed} shows that for any convex orbit set $\mathcal{K}_\mu$ where $\preceq_\C$ satisfies the min-closure property, the extreme and exposed points admit a simple structure: (i) $\nu$ is an extreme point if and only if it can be uniquely constructed via an order-preserving kernel from $\mu$ that is pointwise extremal; (ii) $\nu$ is an exposed point if and only if it can be uniquely constructed via an order-preserving kernel from $\mu$ that is pointwise exposed by the same function $f$. \Cref{prop:exposed} nests the extreme-point characterization of \citet{roy1976extremal} for the convex order as a special case. As we will demonstrate, these properties are extremely helpful for explicitly identifying the extreme and exposed points of the orbits defined by multidimensional mean-preserving spreads and multidimensional stochastic dominance.

\subsection{Exposed Points of Multidimensional MPS and FOSD}
In what follows, we demonstrate the application of \Cref{prop:exposed} by characterizing exposed points of multidimensional MPS and FOSD orbits.\footnote{For simplicity, we focus on the characterization of exposed points, as they are already dense in the set of extreme points by Straszewicz's theorem (\citealt{Kleejr1958}). Alternatively, we could directly apply part (i) of \Cref{prop:exposed} to characterize the extreme points. Unlike in the one-dimensional case, not all extreme points are exposed, but the characterizations are largely similar. } 

\paragraph{Multidimensional MPS.}\hspace{-2mm}Let 
\[\text{MPS}(\mu) := \big\{\nu: \nu \preceq_{\text{concave}} \mu\big\}\,,\]
where $\preceq_{\text{concave}}$ is the concave order, i.e., $\C$ is the set of concave functions on $\R^n$. Since concave functions are closed under pointwise minimum, by \Cref{thm:main}, optimization problems involving $\text{MPS}(\mu)$ admit simple structures---in particular, \Cref{cor:pointwise,cor:chain,cor:nested-optimization} and \Cref{prop:exposed} immediately apply. Now, to illustrate, we apply \Cref{prop:exposed} to further characterize the exposed points of $\text{MPS}(\mu)$.

For any $f$, the $\C$-envelope $\overline{f}$ is simply the concavification of $f$, which is a.e. differentiable. Let $\mathcal{A}$ be the set of differentiability points of $\overline{f}$. For any $x \in \mathcal{A}$, let $h^x$ be the unique supporting hyperplane of $\overline{f}$. Let $\mathcal{H}:=\{h^x\}_{x \in \mathcal{A}}$. Define 
\[S := \Big\{x: \overline{f}(x) = f(x)\Big\}\,.\]
For any $h \in \mathcal{H}$, let 
\[R_h := \Big\{x: \overline{f}(x) = h(x) \Big\}\,,\]
and let $S^c_h := R_h \cap S^c$. Then, note that $\big\{S^c_h\big\}_{h \in \mathcal{H}}$ is a partition of $S^c \cap \mathcal{A}$. We refer to $\big\{S^c_h\big\}_{h \in \mathcal{H}}$
as the \textit{\textbf{affine envelope regions}}.
For any $h \in \mathcal{H}$, let 
\[T_h := \Big\{x: f(x) = h(x)\Big\}\,.\]
We say that $f$ has a \textit{\textbf{strict contact}} if $x \not \in \text{conv}(T_{h^x} \backslash \{x\})$ for a.e. $x \in S \cap \mathcal{A}$.
We say an affine envelope region $S^c_h$ is \textit{\textbf{simplicial-touching}} if 
\[R_h \text{ is a simplex} \text{ and } \Ext(R_h) = T_h \,.\]
If $S^c_h$ is simplicial-touching, then for any $x \in S^c_h$, there exists a unique mixture of extreme points of $R_h$ that represents $x$---call it the unique \textit{\textbf{barycentric splitting}} of $x$.
\begin{figure}[t]
  \begin{subfigure}[b]{0.45\textwidth}
    \centering
    \includegraphics[scale=0.4]{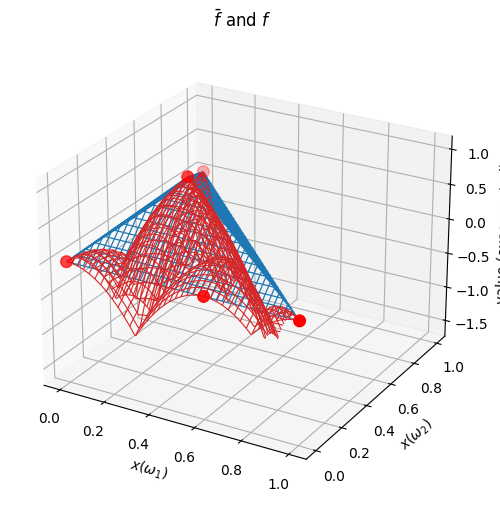}    
    \caption{$f$ and $\overline{f}$\label{fig:mps-both}}

  \end{subfigure}
  \qquad \qquad 
  \begin{subfigure}[b]{0.35\textwidth}
    \centering
    \includegraphics[scale=0.35]{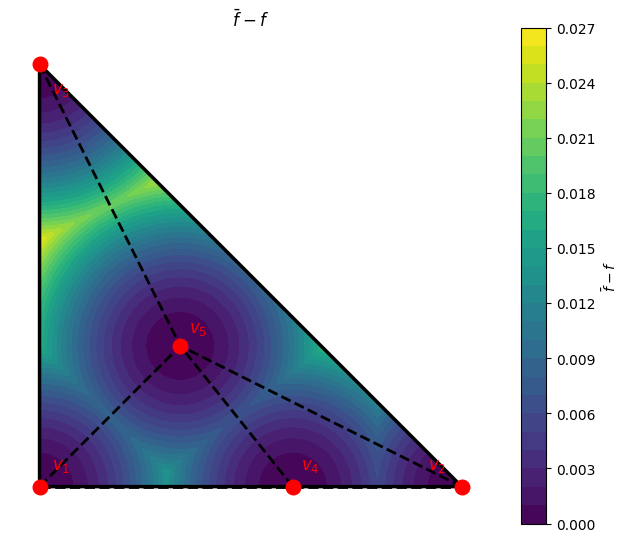}     
    \caption{$\mathrm{supp}(\nu)$\label{fig:mps-heat}}
  \end{subfigure}
  \caption{An exposed point $\nu$ of $\mathrm{MPS}(\mu)$}
  \label{fig:mps-exposed}
\end{figure}

\begin{prop}[Exposed Points of Multidimensional MPS]\label{prop:mps}
For any full-dimensional compact convex $X \subseteq \R^n$ and any $\mu \in \Delta(X)$ that is mutually absolutely continuous with respect to the Lebesgue measure, the following statements are equivalent: 
\begin{itemize}
    \item[(a)] $\nu$ is an exposed point of $\emph{MPS}(\mu)$.
    \item[(b)] There exists a continuous $f$ with a strict contact and simplicial-touching affine envelope regions $S^{c}_{h^x}$ for a full-measure set of $x \in S^c$, and $\nu = P^f * \mu$, where $P^f(\,\cdot\,\mid x) = \delta_x$ for $x \in S$ and $P^f(\,\cdot\,\mid x)$ is the unique barycentric splitting for the full-measure set of $x \in S^c$.
\end{itemize}
\end{prop}

\Cref{fig:mps-exposed} illustrates an exposed point $\nu$ of the set $\text{MPS}(\mu)$. \Cref{fig:mps-both} plots the associated exposing function $f$ (in red) and its concave envelope $\overline{f}$ (in blue), and marks the points where $\overline{f}=f$. It can be readily seen from \Cref{fig:mps-both} that $f$ has a strict contact and simplicial-touching affine envelope regions. \Cref{fig:mps-heat} plots the heat map of $\overline{f}-f$, and the dots correspond to the points where $\overline{f}=f$, which indicate the support of $\nu$. 

An immediate consequence of \Cref{prop:mps} is that any exposed point of $\text{MPS}(\mu)$ takes the form of $P*\mu$, where the kernel $P$ has the property that either (i) it does not move the mass; or (ii) on a collection of non-overlapping simplices (e.g., triangles in $\R^2$), it splits each point in a simplex into the vertices of the simplex. This immediately generalizes the exposed point characterization of \citet*{kleiner2021extreme} in the one-dimensional case, where the convex regions are simply intervals.

\paragraph{Multidimensional MPC.}\hspace{-2mm}Now, consider the orbit 
\[\text{MPC}(\mu) := \big\{\nu: \mu \preceq_{\text{concave}} \nu\big\} = \big\{\nu: \nu \preceq_{\text{convex}} \mu \big\}\,.\]
Since the set of convex functions is not closed under pointwise minimum---it is closed under pointwise maximum---it follows that optimization involving $\text{MPC}(\mu)$ generally admits complex structures. For example, by \Cref{thm:main}, the value function $V^\star_f$ must be strictly concave on $\Delta(X)$ for some $f$. See \citet{kleiner2024extreme} for the characterization of the structure of Lipschitz exposed points of $\text{MPC}(\mu)$ for finitely-supported $\mu$. 

\paragraph{Multidimensional Stochastic Dominance.}\hspace{-2mm}Let 
\[\text{LSD}(\mu) := \big\{\nu: \nu \preceq_{\text{nondecreasing}} \mu\big\}\,,\]
where $\preceq_{\text{nondecreasing}}$ is the stochastic dominance order, i.e., $\C$ is the set of nondecreasing functions on $\R^n$. Similarly, let 
\[\text{HSD}(\mu) := \big\{\nu:   \mu \preceq_{\text{nondecreasing}} \nu\big\} = \big\{\nu:   \nu \preceq_{\text{nonincreasing}} \mu\big\}\,.\]
Since both nondecreasing functions and nonincreasing functions are closed under pointwise minimum, by \Cref{thm:main}, optimization problems involving $\text{LSD}(\mu)$ and $\text{HSD}(\mu)$ admit simple structures---in particular, \Cref{cor:pointwise,cor:chain,cor:nested-optimization}  and \Cref{prop:exposed} immediately apply. Now, to illustrate, we apply \Cref{prop:exposed} to further characterize the exposed points of $\text{LSD}(\mu)$ and $\text{HSD}(\mu)$.

We start with $\text{LSD}(\mu)$. For any $f$, the $\C$-envelope $\overline{f}$ is simply the monotone envelope of $f$, i.e. $\overline{f}(x) = \max_{y \leq x} f(y)$. Define 
\[S := \Big\{x: \overline{f}(x) = f(x)\Big\}\,.\]
For any $z \in \text{Ran}(\overline{f})$, let 
\[R_z := \Big\{x: \overline{f}(x) = z\Big\}\,, \]
and 
\[T_z := S \cap R_z = \Big\{x: \overline{f}(x) = f(x) = z\Big\}\,.\]
We say that $f$ has a \textit{\textbf{strict monotone contact}} if for a.e. $x \in S$, we have $f(x) > f(y)$ for all $y \leq x$ and $y \neq x$ (in the coordinatewise order). For any $z \in \text{Ran}(\overline{f})$, let  
\[S^c_z:= R_z \cap S^c\,.\]
Clearly, $\{S^c_z\}_{z \in \text{Ran}(\overline{f})}$ is a partition of $S^c$. We write $z^x = \overline{f}(x)$. We refer to $\{S^c_z\}$ as \textit{\textbf{monotone envelope regions}}. We say a monotone envelope region $S^c_z$ is \textit{\textbf{downward staircase-like}} if there exists some upward-closed set $A_z \subset [0, 1]^n$ such that 
\[R_z \backslash A_z = \biguplus_{\underline{x} \in T_z}\Big(\big\{\,x: x \geq \underline{x} \,\big\}\, \backslash \,A_z\Big) \,.\]
If $S^c_z$ is downward staircase-like, then for any $x \in R_z \backslash A_z$, there exists a unique \textit{\textbf{downward transport}} that deterministically transports $x$ to $\underline{x}$ for some $\underline{x} \in T_z$. 

\begin{prop}[Exposed Points of Multidimensional FOSD]\label{prop:fosd}
For any $\mu \in \Delta([0, 1]^n)$ that is mutually absolutely continuous with respect to the Lebesgue measure, the following statements are equivalent: 
\begin{itemize}
    \item[(a)] $\nu$ is an exposed point of $\emph{LSD}(\mu)$.
    \item[(b)] There exists a continuous $f$ with a strict monotone contact and downward staircase-like monotone envelope regions $S^{c}_{z^x}$ for a full-measure set of $x \in S^c$ with $x \in R_{z^x}\backslash A_{z^x}$, and $\nu = P^f * \mu$, where $P^f(\,\cdot\,\mid x) = \delta_x$ for $x \in S$ and $P^f(\,\cdot\,\mid x)$ is the unique downward transport for the full-measure set of $x \in S^c$.
\end{itemize}
\end{prop}

\begin{figure}[t]
  \begin{subfigure}[b]{0.45\textwidth}
    \centering
    \includegraphics[scale=0.3]{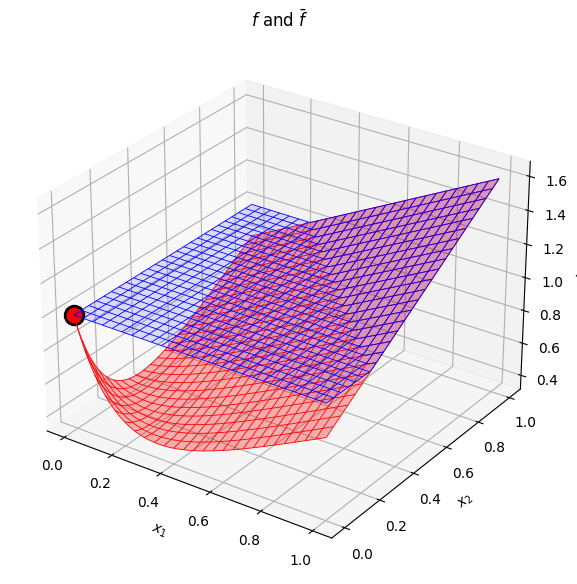}     
    \caption{$f$ and $\overline{f}$\label{fig:fosd-both}}
 
  \end{subfigure}
  \quad 
  \begin{subfigure}[b]{0.45\textwidth}
    \centering
    \includegraphics[scale=0.3]{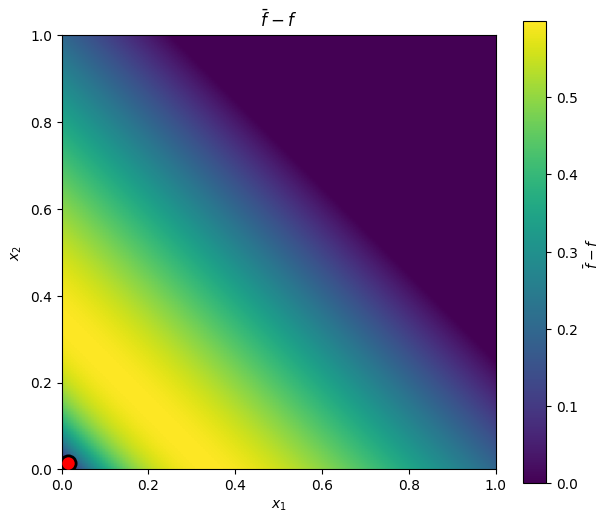}     
    \caption{$\mathrm{supp}(\nu)$\label{fig:fosd-heat}}
  \end{subfigure}
  \caption{An exposed point $\nu$ of $\mathrm{LSD}(\mu)$}
  \label{fig:fosd-exposed}
\end{figure}

\Cref{fig:fosd-exposed} illustrates an exposed point $\nu$ of $\text{LSD}(\mu)$. \Cref{fig:fosd-both} plots the associated exposing function $f$ and its monotone envelope $\overline{f}$. It can be seen from \Cref{fig:fosd-both} that $f$ has a strict monotone contact. \Cref{fig:fosd-heat} plots the heat map of $\overline{f}-f$. The region where $\overline{f}=f$ consists of the point $(0,0)$ and the top-right corner. Every point not in this region has a unique downward transport to $(0,0)$.   

An immediate consequence of \Cref{prop:fosd} is that any exposed point of $\text{LSD}(\mu)$ takes the form of $P*\mu$, where the kernel $P$ has the property that either (i) it does not move the mass; or (ii) on a collection of non-overlapping differences of two nested upsets that are downward staircase-like, it transports each point in any such region down to a corresponding unique minimal element. This immediately generalizes the exposed point characterization of \citet{yang2024monotone} in the one-dimensional case, where the differences of two nested upsets are simply intervals. 

The case of $\text{HSD}(\mu)$ is exactly symmetric to $\text{LSD}(\mu)$, in contrast to the asymmetry that arises between $\text{MPS}(\mu)$ and $\text{MPC}(\mu)$. In particular, it can be characterized symmetrically by using \textit{\textbf{decreasing envelope regions}}, the differences of nested downward-closed sets that form \textit{\textbf{upward staircase-like}} shapes, and the \textit{\textbf{upward transport}} maps. We omit the details here.

\section{Consistent Comparisons of Experiments}\label{sec:compare}
In this section, we apply our main characterization to generalize Blackwell's theorem. In a pioneering paper, \citet{Blackwell1951Comparison} identifies an order on experiments that has a dual representation. On the one hand, the Blackwell order has an instrumental-value-based representation: an experiment Blackwell dominates another if and only if one yields a higher value under \emph{all} statistical decision problems than the other. On the other hand, the Blackwell order has an information-technology-based interpretation: an experiment Blackwell dominates another if and only if one can be obtained by adding noise to (i.e., garbling) the other. In the space of distributions of posterior beliefs induced by an experiment under a given prior, this dual representation essentially identifies an order $\preceq$ over $\Delta(X)$, a convex cone of test functions (i.e., a class of problems) $\C$, and a set of information technologies $\Pc$, such that the order $\preceq$ can be simultaneously represented by comparing values of problems $g \in \C$, as well as being coupled by some information technologies $P \in \Pc$. In this regard, Blackwell's theorem identifies such an order where $\C$ is the set of convex functions and $\Pc$ is the set of martingales.  

Given Blackwell's theorem, one natural question is: in general, which orders $\preceq$ over $\Delta(X)$ admit such dual representations? When is it the case that an abstract order $\preceq$ over $\Delta(X)$ can be simultaneously interpreted as an order that reflects instrumental values, as well as an order that reflects the information content of distributions over posterior beliefs? In this section, we characterize all such orders. To begin with, we first introduce some notation and terminology. 

Throughout this section, fix a compact Polish state space $\Omega$, and let $X := \Delta(\Omega)$ be the space of beliefs over $\Omega$. Let $\C_0$ be the set of bounded lower semicontinuous functions on $X$ and $\Pc_0$ be the set of measurable Markov kernels. For any $\Pc \subseteq \Pc_0$ and for any $x \in X$, let 
\[
\Pc_x:=\Big\{P(\,\cdot\,\mid x): P \in \Pc \Big\}\,.
\]
We say that $\Pc \subseteq \Pc_0$ is \textit{\textbf{rectangular}} if $P \in \Pc$ whenever $P(\,\cdot\,\mid x) \in \Pc_x$ for all $x \in X$. We say that a pair $(\C,\Pc)$ is \textit{\textbf{regular}} if $\C \subseteq \C_0$ is a convex cone such that $\C \cap C(X)$ is sup-norm closed, contains all constant functions, admits continuous approximations from below, and is closed under bounded increasing pointwise limits; while $\Pc \subseteq \Pc_0$ has a closed graph,\footnote{That is, the graph $\{(x,\eta): x\in X\,,\, \eta \in \mathcal{P}_x\}$ is closed in $X \times \Delta(X)$.} contains the identity kernel, and is convex and rectangular. We impose regularity on $(\C,\Pc)$ throughout this section unless stated otherwise. 

For any $\C$, define the following comparison:
\[\mu \preceq_{\C} \nu \,\,\text{ if }\,\, \int g(x) \nu(\d x) \geq \int g(x) \mu( \d x)\,, \forall\, g \in \C\,. \eqtag[-2.5em][0.6\baselineskip]{\textbf{Value Description}}{eq:payoff}\]
For any $\Pc$, define the following comparison:
\[\mu \preceq^{\Pc} \nu \,\,\text{ if }\,\, \exists \, P \in \Pc \text{ s.t. } \nu = P * \mu\,. \eqtag[-2.5em][0.2\baselineskip]{\textbf{Information Description}}{eq:information} \]
For any order $\preceq$ on $\Delta(X)$, we say that $\preceq$ is \textit{\textbf{Blackwell-consistent}} if there exists some (regular) $(\C, \Pc)$ such that 
\[\preceq_{\C} \,= \,\preceq \,= \, \preceq^\mathcal{P}\,, \eqtag[-2.5em][0.2\baselineskip]{\textbf{Blackwell Consistency}}{eq:blackwell-consistency}\]
in which case we call $(\C, \Pc)$ a \textit{\textbf{Blackwell-consistent pair}}.

We say that $\C$ is \textit{\textbf{max-closed}} if for any $g_1,g_2 \in \C$, $\max\{g_1,g_2\} \in \C$. As an immediate consequence of \Cref{thm:main}, our next result characterizes all Blackwell-consistent orders: 

\begin{theorem}\label{thm:main2}
An order $\preceq$ is Blackwell-consistent if and only if $\preceq = \preceq_\C$ for a max-closed $\C$.
\end{theorem}
\begin{proof}[Proof of \Cref{thm:main2}]
\textbf{The ``if'' direction:} Suppose $\preceq = \preceq_\C$ for a max-closed $\C$. Let 
\[\Pc := \Big\{P \in \Pc_0: \delta_x \preceq_\C  P *\delta_x \text{ for all $x$} \Big\}\,.\]
Note that $\mu\preceq \nu \iff \nu \preceq_{[-\C]} \mu$ where $[-\C]$ is a convex cone that is min-closed. Thus, for any $\mu\preceq \nu$, by \Cref{thm:main},\footnote{It can be verified that the technical assumptions required by \Cref{thm:main} hold for $[-\C]$ if $\C$ is regular.}  there exists some $P$ such that $\nu = P * \mu$ where $P * \delta_x  \preceq_{[-\C]} \delta_x$, i.e., $\delta_x \preceq_{\C} P * \delta_x$ for all $x$, and hence $P \in \Pc$.  Therefore, $\preceq_\C \implies \preceq^\Pc$. Moreover, for any $\nu = P * \mu$ for some $P \in \Pc$, since $\preceq_\C$ is an integral stochastic order and $\delta_x \preceq_\C P * \delta_x$ for all $x$, it follows immediately that $\mu \preceq_\C \nu$. Thus, $\preceq^\Pc \implies \preceq_\C$ and hence these two are equivalent. Thus, $\preceq$ is Blackwell-consistent.\footnote{It can be verified that the constructed $\Pc$ is regular; see \Cref{lem:basic} in the appendix.}  

\textbf{The ``only if'' direction:} Suppose $\preceq$ is Blackwell-consistent. Then there exists some regular $(\C, \Pc)$ such that $\preceq_\C = \preceq = \preceq^\Pc$. We claim that $\preceq_{[-\C]}$ admits order-preserving couplings. Indeed, fix any $\nu \preceq_{[-\C]} \mu$. Then, $\mu \preceq_{\C} \nu$ and hence $\nu = P * \mu$ for some $P \in \Pc$. For any $x$, note that by definition $\delta_x \preceq^{\Pc} P * \delta_x$ and hence $\delta_x \preceq_\C P * \delta_x$ by Blackwell consistency. Thus, for all $x \in X$ and all $g \in \C$, we have  
\[\int g(y) P(\d y\mid x) \geq g(x)\]
and hence 
\[\int [-g](y) P(\d y\mid x) \leq [-g](x)\,.\]
Thus, $P * \delta_x \preceq_{[-\C]} \delta_x$ for all $x$. Therefore, $\preceq_{[-\C]}$ admits order-preserving couplings. By \Cref{thm:main}, $[-\C]$ is min-closed and hence $\C$ is max-closed, as desired. 
\end{proof}

\Cref{thm:main2} asserts that $\preceq$ is Blackwell-consistent---admitting simultaneous value-based and information-based representations---if and only if the class of test functions $\C$ in its value description is closed under pointwise maximum. Intuitively, it means that whenever $g_1, g_2 \in \C$ are two problems, the problem whose value at each posterior belief $x$ is $\max\{g_1(x), g_2(x)\}$---corresponding to choosing the better of the two problems at each belief---must also belong to $\C$. This ``choose the better problem'' closure property turns out to be the \textit{exact} condition that bridges the value-based and information-based descriptions of an order. 

Now, a few natural questions arise. What is the key structure of $\Pc$ in the dual space that mirrors the max-closure property of $\C$? Which pairs $(\C, \Pc)$ are consistent? Given a Blackwell-consistent order, are there multiple pairs $(\C, \Pc)$ that represent the same order? We now answer all of them in our next result.

For any $P_1,P_2 \in \Pc_0$, recall from \weqref{\textbf{Composition}}{eq:composition} that $P_1 \circ P_2$ denotes the composition of the two kernels. We say that $\Pc \subseteq \Pc_0$ is \textit{\textbf{composition-closed}} if for any $P_1,P_2 \in \Pc$, $P_1\circ P_2 \in \Pc$. As we show, it turns out that composition-closure property of $\Pc$ is the \textit{exact} condition that mirrors the max-closure property of $\C$, each of which is equivalent to Blackwell consistency. 

To characterize the consistent pairs, we build on the construction given in the proof of \Cref{thm:main2}. To start, recall the notation from \weqref{\textbf{Conditional Expectation}}{eq:conditional-expectation}, $[g*P](x) := \int g(y) P(\d y \mid x)$. Define two operators $\Phi$ and $\Psi$ as follows: For any $\C \subseteq \C_0$, let 
\[
\Psi(\C):=\Big\{P \in \Pc_0: g*P \geq g\,, \, \, \forall\, g \in \C\Big\}\, \eqtag[-2.5em]{\textbf{Improving Kernels}}{eq:improving-kernels}
\]
denote the set of belief transitions that are uniformly improving across all test functions in $\C$. For any $\Pc \subseteq \Pc_0$, let 
\[
\Phi(\Pc):=\Big\{g \in \C_0: g*P \geq g\,, \, \, \forall\, P \in \Pc\Big\}\, \eqtag[-2.5em]{\textbf{Improving Utilities}}{eq:improving-utilities}
\]
denote the set of test functions that are uniformly improving across all belief transitions in $\Pc$. We say that $(\C, \Pc)$ is \textit{\textbf{Blackwell-invariant}} if 
\[ \Psi(\C)=\Pc  \text{ and } \Phi(\Pc)=\C \eqtag[-2.5em][0.2\baselineskip]{\textbf{Blackwell Invariance}}{eq:blackwell-invariance}\,.\]
A Blackwell-invariant pair $(\C,\Pc)$ has a straightforward interpretation. It means that the available information technologies $\Pc$ are exactly the ones that are valuable for problems in $\C$ (i.e., $\Pc=\Psi(\C)$). In the meantime, the problems $\C$ are exactly the ones for which the available information technologies $\Pc$ are valuable (i.e., $\C=\Phi(\Pc)$).

Finally, we say that $\mathcal{P}$ is \textit{\textbf{one-shot sufficient}} if for all $f$ and all $x$, we have
\[\sup_{P \in \mathcal{P}} \int f(y) P(\d y \mid x) = \sup_{P \in \mathcal{P} \otimes \mathcal{P}} \int f(y) P(\d y \mid x), \eqtag[-2.5em][0.6\baselineskip]{\textbf{One Shot}}{eq:one-shot-sufficient} \]
where 
\[
\Pc \otimes \Pc:=\{P_1 \circ P_2: P_1,P_2 \in \Pc\}
\] 
denotes the set of all (twice) compositions of elements in $\Pc$. It is clear that if $\mathcal{P}$ is composition-closed, then it is one-shot sufficient. 

In the case of the Blackwell order, $\C$ is the cone of convex functions and $\Pc$ is the set of martingale transition kernels. Clearly, $\C$ is max-closed and $\Pc$ is composition-closed. Moreover, $\Psi(\mathcal{C}) = \Pc$ and  $\Phi(\Pc) = \mathcal{C}$. That is, the test functions $\C$ and transition kernels $\Pc$ associated with the Blackwell order are in fact Blackwell-invariant. This turns out not to be a coincidence. Our next result shows that (i) the composition-closure property of $\Pc$ mirrors the max-closure property of $\C$, each of which characterizes Blackwell consistency, and (ii) Blackwell consistency is equivalent to Blackwell invariance, which turns out to be also equivalent to the one-shot sufficiency property of the set of transition kernels for optimization problems. 

In general, different convex cones $\C$ and $\C'$ or different sets of kernels $\Pc$ and $\Pc'$ might induce the same comparisons, i.e., $\preceq_\C=\preceq_{\C'}$ or $\preceq^{\Pc}=\preceq^{\Pc'}$. Since Blackwell consistency is only a property about the induced order, one might naturally expect that there might be multiplicity involved in determining the consistent pairs. However, perhaps surprisingly, every Blackwell-consistent order admits a \textit{\textbf{unique}} consistent pair that characterizes it. 

\begin{theorem}[Blackwell-Consistent Pairs]\label{thm:main-compare}
The following are equivalent: 
\begin{itemize}
    \item[(a)] $(\C, \Pc)$ is Blackwell-consistent.
    \item[(b)] $(\C, \Pc)$ is Blackwell-invariant.
    \item[(c)] $\Pc$ is one-shot sufficient, and $\C = \Phi(\Pc)$. 
    \item[(d)] $\Pc$ is composition-closed, and $\C = \Phi(\Pc)$.
    \item[(e)] $\mathcal{C}$ is max-closed, and $\mathcal{P} = \Psi(\C)$. 
    \end{itemize}
Moreover, every Blackwell-consistent order $\preceq$ admits a unique consistent pair $(\C, \Pc)$.   
\end{theorem}

\subsection{Which Orders are Consistent under Bayesian Updating?}\label{sec:Bayes-plausible}
With \Cref{thm:main2} and \Cref{thm:main-compare},  we now further specialize to Blackwell's environment, and characterize which orders of Blackwell experiments are Blackwell-consistent. Clearly, if beliefs are updated via Bayes' rule, then by Blackwell's theorem, the Blackwell order is Blackwell-consistent. We characterize all orders of Blackwell experiments that are Blackwell-consistent when beliefs are updated via Bayes' rule in this section, and characterize all updating rules that admit Blackwell-consistent orders in the next section.

\subsubsection{Bayes-Plausible and Blackwell-Consistent Orders}
To begin with, we say that an order $\preceq$ on $\Delta(X)$ is \textit{\textbf{Bayes-plausible}} if 
\[\mu \preceq \nu \implies \int x \mu(\d x) = \int x \nu(\d x)\, \eqtag[-2.5em]{\textbf{Bayes Plausibility}}{eq:bayes-plausible}\,,\]
where the integral is in the sense of the Pettis integral, i.e., the two belief distributions must have the same barycenter. 

An immediate consequence of \Cref{thm:main2} is the following: 
\begin{cor}\label{cor:blackwell-consistent-orders}
For any pair $(\C, \Pc)$, the following are equivalent:
\begin{itemize}
    \item[(a)] $\preceq_\C = \preceq^\Pc$ is Bayes-plausible.
    \item[(b)] $\C$ is a max-closed superset of all 
    convex functions, and $\Pc = \Psi(\C)$.
    \item[(c)] $\Pc$ is a composition-closed subset of 
    all martingale kernels, and $\C = \Phi(\Pc)$.
\end{itemize}
In particular, any Bayes-plausible and Blackwell-consistent order strengthens the Blackwell order, and the Blackwell 
order is the weakest such order.
\end{cor}

Indeed, \Cref{cor:blackwell-consistent-orders} follows because Bayes plausibility implies that all affine functions must be in the cone of test functions, which combined with \Cref{thm:main2} and \Cref{thm:main-compare} implies that all convex functions must be in the test cone since the test cone must be max-closed and any convex function is a pointwise supremum of a collection of affine functions. Therefore, any Bayes-plausible weakening of the Blackwell order through a subset of convex functions, such as the Lehmann order (\citealt{Lehmann}) or the linear convex order (\citealt{chen2025experiments}), cannot be consistent in our sense.\footnote{It is immediate to verify that neither of these orders is defined by a set of test functions that is max-closed.} \Cref{cor:blackwell-consistent-orders} shows that any Bayes-plausible and Blackwell-consistent order $\preceq$ must imply the Blackwell order---therefore, the Blackwell order is necessary for a consistent comparison, while maintaining Bayes plausibility.  In other words, any Bayes-plausible, Blackwell-consistent order $\preceq$ must rely on a set $\C$ of problems that includes all decision problems, and on a set $\Pc$ of information technologies that is a subset of martingales. 

\subsubsection{Constrained Information Design}
An alternative interpretation of \Cref{cor:blackwell-consistent-orders} is a characterization of what kinds of constrained information technologies admit a consistent representation based on the instrumental value of information. In particular, \Cref{cor:blackwell-consistent-orders} identifies a class of \textbf{\textit{constrained information design}} problems---where the feasible information technology is constrained to a subset $\Pc$ of martingales---that are particularly tractable and preserve many of the desirable properties of the unconstrained problem.  

A rectangular set of transition kernels $\Pc$ equivalently defines a set of feasible experiments $\mathcal{P}_x$ at each $x$. Thus, $\Pc$ defines a constrained information design problem: 
\[V^{\mathcal{P}}_f(\delta_{x}) := \max_{\delta_{x} \preceq^{\mathcal{P}} \nu} \int f(y)  \nu(\d y)\,,\]
where $\nu$ is the posterior distribution given the prior $x$ and an experiment chosen from $\mathcal{P}_x$. In persuasion problems, the function $f$ is the indirect utility function of the sender. In flexible learning problems with a uniformly posterior separable cost function, the function $f(x) = v(x) - h(x)$ where the convex function $v$ is the indirect utility function from a decision problem, and the convex function $h$ is the measure of uncertainty. 

By \Cref{thm:main-compare}, if $\Pc$ is composition-closed, then $\preceq^{\mathcal{P}} = \preceq_{\mathcal{C}}$ for a max-closed $\mathcal{C}$, which by \Cref{thm:main} then implies the following envelope characterization: 
\[V^{\mathcal{P}}_f(\mu) = \int \Big( \max_{\delta_{x} \preceq_{\mathcal{C}} \nu} \int f(y)  \nu(\d y)\Big) \mu(\d x) = \int \widehat{f}(x)  \mu(\d x) \,,\]
where $\widehat{f}$ is the $[-\mathcal{C}]$-envelope of $f$. Importantly, the envelope $\widehat{f}$ here is \textit{\textbf{prior-independent}}, just like in the unconstrained case. This contrasts with the case where the constraints are moment constraints on the posterior distributions $\nu$---such as the ones studied by \citet{doval2024constrained}, who show that the value of the constrained problem can be found by concavifying the Lagrangian given the optimal multipliers which depend on the prior. Indeed, such moment constraints on the posterior distributions are generally \textit{not} composition-closed in terms of the implied set $\mathcal{P}$, which as our results show are also necessary for the existence of prior-independent envelope characterization. An immediate consequence of this prior-independent envelope is the following characterization generalizing the unconstrained case: 
\begin{cor}\label{cor:envelope}
For any composition-closed $\mathcal{P}$, a feasible $\nu$ is constrained-optimal if and only if 
\[\supp(\nu) \subseteq \Big\{y \in X: f(y) = \widehat{f}(y)\Big\}\]
and 
\[\int f(y)\nu(\d y) = \widehat{f}(x)\,, \]
where $x$ is the prior and $\widehat{f}$ is the $[-\Phi(\mathcal{P})]$-envelope of $f$. The value of a constrained optimal signal is $\widehat{f}(x)$, and the designer benefits from constrained signals if and only if $\widehat{f}(x) > f(x)$. 
\end{cor}

\Cref{cor:envelope} shows that many structural properties derived in  \citet{Kamenica2011} can be immediately generalized to the constrained case where signals need not be fully flexible as long as $\mathcal{P}$ is composition-closed. The key geometric structure can be obtained by replacing the concave envelope with the $[-\Phi(\mathcal{P})]$-envelope. 

How does the constrained optimal signal compare to the KG signal? 
We say that the constrained problem is  \textit{\textbf{value-eliminating}} if there exists some prior $x \in X$ under which the constrained designer cannot benefit from information design while the unconstrained designer strictly benefits. A consequence of \Cref{cor:envelope} is the following comparison: 
\begin{prop}[Constrained vs. KG Signals]\label{prop:comparison}
For finite $\Omega$, composition-closed $\mathcal{P}$, and any $f$,
\begin{itemize}
    \item[(i)] either the constrained problem is value-eliminating\,;     
    \item[(ii)] or at any prior $x$ where the unconstrained optimal signal is unique, any constrained optimal signal cannot be strictly Blackwell-dominated by the unconstrained optimal signal, and always weakly Blackwell-dominates the unconstrained optimum when $|\Omega| = 2$. 
\end{itemize}
\end{prop}

Indeed, to see \Cref{prop:comparison}, note that if the constrained problem is not value-eliminating, then by \Cref{cor:envelope} 
\[\widehat{f}^{\text{cav}}(x) > f(x) \implies \widehat{f}^{\C}(x) > f(x)\,,\]
where we use the superscripts to distinguish the concave envelope versus $\C$-envelope where $\C = -\Phi(\Pc)$. By  \Cref{cor:envelope} again, this implies that any constrained optimal $\nu^\star$ must satisfy
\[\nu^\star\Big(\text{conv}\big(\supp(\nu_{\text{KG}})\big) \backslash \supp(\nu_{\text{KG}})\Big) = 0\,,\]
since otherwise there exists some $y$ in the above set such that $\widehat{f}^{\text{cav}}(y) = \widehat{f}^\C(y) = f(y)$ which would then contradict that $\nu_{\text{KG}}$ is the unique unconstrained optimum. This then implies that $\nu^\star$ cannot be strictly Blackwell-dominated by $\nu_{\text{KG}}$ with any finite $\Omega$, and must weakly Blackwell-dominate $\nu_{\text{KG}}$ with two states.

Next, we show two examples of constrained signals where $\Pc$ is composition-closed. 

\paragraph{Privacy-Preserving Signals.}\hspace{-2mm}Consider the setting adopted from \citet{strack2024privacy}. Let $\Omega$ be a compact Polish space and let $\mathcal{E}$ be a sub-$\sigma$-algebra on $\Omega$. An event $E \in \mathcal{E}$ describes a \emph{\textbf{privacy event}} that needs to be protected. Given a prior $x_0$, a signal is said to be \emph{\textbf{privacy-preserving}} if under any signal realization, the posterior belief about any event $E \in \mathcal{E}$ remains the same as the prior $x_0(E)$. For any $x \in X=\Delta(\Omega)$, let 
\[
K(x):=\{y \in X: y(E)=x(E)\,, \, \, \forall E \in \mathcal{E}\}
\]
be the set of posterior beliefs that preserves privacy when the prior is $x$. The set of privacy-preserving signals can be equivalently represented by a set $\Pc$ of martingales such that for all $P \in \Pc$ and for all $x \in X$,
\[
P(K(x)\mid x)=1\,.
\]
Note that $\mathcal{K}:=\{K(x): x \in X\}$ forms a partition of $X$. As a result, $\Pc$ is composition-closed. 

The order $\preceq^{\Pc}$ has a natural interpretation. Two distributions $\mu,\nu \in \Delta(X)$ are ordered by $\preceq^{\Pc}$ if and only if $\nu$ can be attained from $\mu$ via privacy-preserving signals $P \in \Pc$ that do not further reveal any information about events in $\mathcal{E}$ than what $\mu$ had already revealed.\footnote{In the language of \citet{strack2024privacy}, this means that $\nu$ is the distribution of posterior beliefs induced by a conditionally privacy-preserving signal given the information revealed by $\mu$.} Since $\Pc$ is composition-closed, \Cref{thm:main2} ensures that $\preceq^{\Pc}$ is a Blackwell-consistent order. As a result, any distribution of posterior beliefs $\nu$ that can be obtained from $\mu$, via disclosing privacy-preserving information, gives a higher instrumental value in all problems $g \in \C$, where 
\[
\C:=\Phi(\Pc)\,.
\]
The set $\C$ of problems also has an intuitive description: every $g \in \C$ must be such that $g$ is convex when restricted to each partitional element $K \in \mathcal{K}$. In other words, information disclosed by privacy-preserving signals can be ordered by their instrumental values for decision problems, restricted to the domain of beliefs that preserve privacy. In particular, for a fixed prior $x_0$, the order over privacy-preserving signals defined by $\Pc_{x_0}$ coincides with the Blackwell order. In general, however, the order $\preceq^{\Pc}=\preceq_{\C}$ is a strengthening of the Blackwell order, and there are signals that are Blackwell-ordered but are not ordered by the privacy-preserving order. These are exactly signals that disclose more information in ways that violate privacy. 

An immediate consequence of \Cref{thm:main-compare} here is an envelope characterization of the optimal value of privacy-preserving persuasion. Given a public signal $\mu \in \Delta(X)$, a sender chooses a (conditionally) privacy-preserving signal that does not reveal more information regarding $\mathcal{E}$. The sender's problem, by \Cref{thm:main-compare}, can be written as: 
\[
\max_{\nu: \mu \preceq_{\C} \nu} \int f \d \nu\,,
\]
where $f:X \to \R$ is the indirect utility of the sender. Since $\C$ is max-closed, $-\C$ is min-closed.  \Cref{thm:main} then implies that the value of the problem is characterized by 
\begin{equation*}
V^\star_f(\mu)=\int \widehat{f}^{\C} \d \mu\,, \eqtag[-2.5em]{\textbf{Privacy-Preserving Persuasion}}{eq:privacy}
\end{equation*}
where 
\[
\widehat{f}^{\C}(x):=\inf\{g(x):g \in -\C\,, \, \, g \geq f\}
\]
is the $[-\C]$-envelope of $f$. The expression of \weqref{\textbf{Privacy-Preserving Persuasion}}{eq:privacy} further simplifies and generalizes Proposition 6 of \citet{strack2024privacy}, which characterizes the value of privacy-preserving persuasion for a fixed prior via concavification and optimal transport. In contrast, \weqref{\textbf{Privacy-Preserving Persuasion}}{eq:privacy} does not involve solving an optimal transport problem, and allows for arbitrary public information $\mu$.  Operationally, the $[-\C]$-envelope is also often easier to compute relative to the concave closure. Indeed, instead of taking the concave closure of the entire function $f$, computing the $[-\C]$-envelope only requires taking the lower-dimensional concave envelope of $f$ on each partitional element $K \in \mathcal{K}$. 

For instance, consider a simple privacy-preserving persuasion problem where the sender is a doctor and the receiver is an insurance company. The doctor can provide information about a patient to the insurance company. The state of the world could be one of three states: $\omega_1$, where the patient is genetically prone to cancer; $\omega_2$, where the patient is not genetically prone but often exposed to chemicals; and $\omega_3$, where the patient is unlikely to have cancer. The receiver can choose whether to insure the patient. The doctor prefers the patient to be insured: they get a payoff of 1 if the patient is insured, and a payoff of zero otherwise. The insurance company prefers to insure a healthy patient: it gets payoff $1$ if $\omega=\omega_3$, $0$ if $\omega=\omega_2$, and $-1$ if $\omega=\omega_1$. The privacy events $\mathcal{E}$ are $\{\{\omega_1\},\{\omega_2,\omega_3\}\}$, since HIPAA prohibits the doctor from disclosing genetic information about a patient. Suppose that there is already some public information about the patient known to the insurance company. Denote by $\mu \in \Delta(X)$ the receiver's initial distribution of posterior beliefs. The sender's privacy-preserving persuasion problem can be written as 
\[
\max_{\nu: \mu \preceq_{\C} \nu} \int f \d \nu\,,
\]
where $f(x):=\ind\{x(\omega_3) \geq x(\omega_1)\}$ is the sender's indirect utility as a function of posterior beliefs. It can be readily shown that the $[-\C]$-envelope of $f$ is given by 
\[
\widehat{f}^\C(x):=\ind\Big\{x(\omega_1) \leq \frac{1}{2}\Big\}\cdot \min\left\{1,\frac{x(\omega_3)}{x(\omega_1)}\right\}\,,
\]
and therefore the sender's value is 
\[
\int_X \ind\Big\{x(\omega_1) \leq \frac{1}{2}\Big\}\cdot \min\left\{1,\frac{x(\omega_3)}{x(\omega_1)}\right\} \mu(\d x)\,.
\]
In contrast, the concave envelope of $f$ is given by 
\[
\widehat{f}^{\text{cav}}(x):=\min\Big\{1,1-x(\omega_1)+x(\omega_3)\Big\}\,,\]
which is greater than $\widehat{f}^\C$. Thus, the sender's value, had there been no privacy constraints, would be 
\[
\int_X \min\Big\{1,1-x(\omega_1)+x(\omega_3)\Big\}  \mu(\d x) \geq \int_X \ind\Big\{x(\omega_1) \leq \frac{1}{2}\Big\}\cdot \min\left\{1,\frac{x(\omega_3)}{x(\omega_1)}\right\} \mu(\d x)\,.
\]
\Cref{fig:persuasion} plots these two envelopes, where the functions are plotted on the $(x(\omega_3),x(\omega_1))$-plane, using the fact that $x(\omega_2)=1-x(\omega_1)-x(\omega_3)$.

\begin{figure}
\centering
\tikzset{
solid node/.style={circle,draw,inner sep=1.25,fill=black},
hollow node/.style={circle,draw,inner sep=1.25}
}
\begin{subfigure}[b]{0.4\linewidth}
\centering
    \includegraphics[scale=0.5]{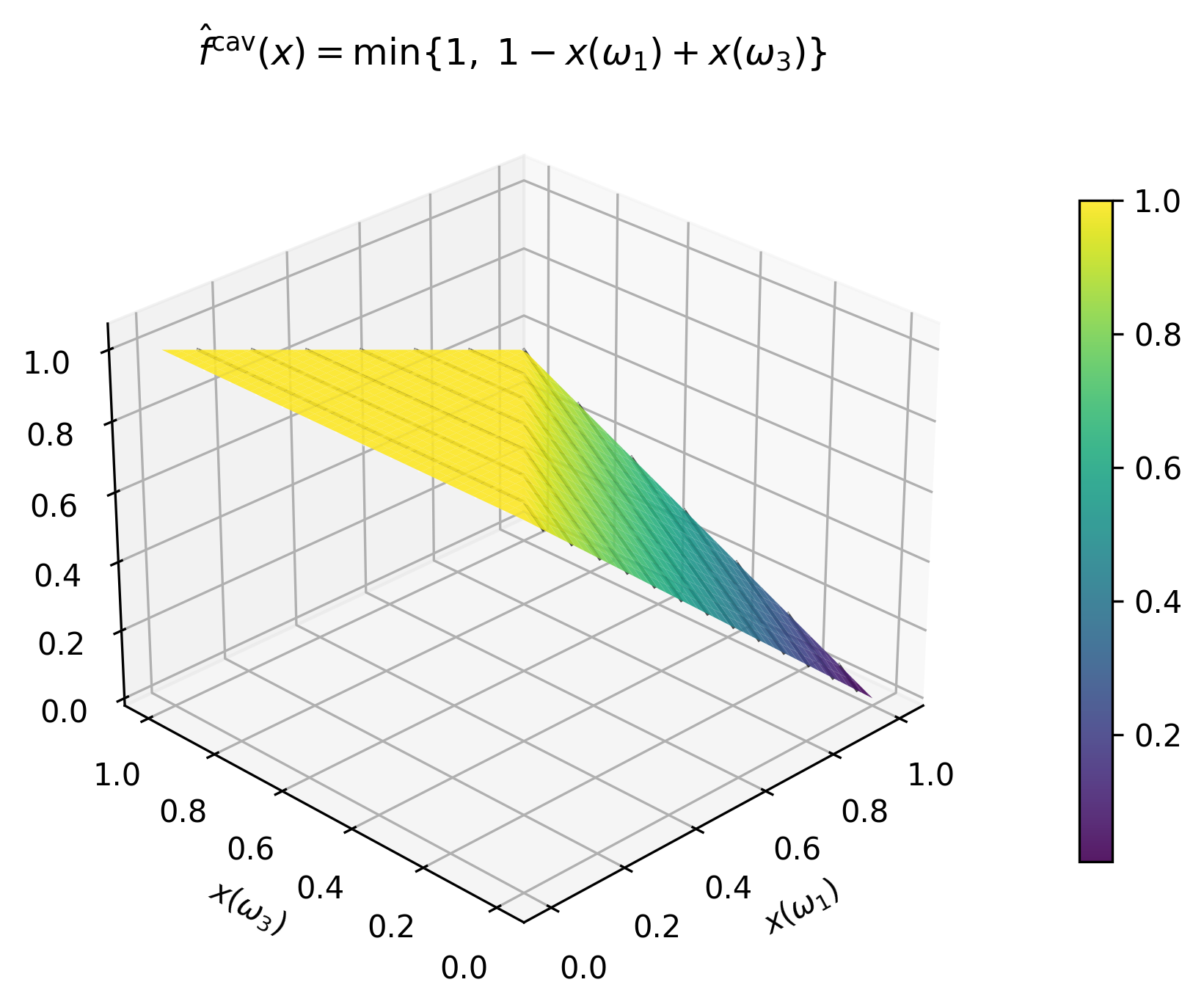}
\caption{Concave envelope $\widehat{f}^{\text{cav}}$}
\label{fig1a}
\end{subfigure}
\quad \quad
\begin{subfigure}[b]{0.4\linewidth}
\centering
 \includegraphics[scale=0.5]{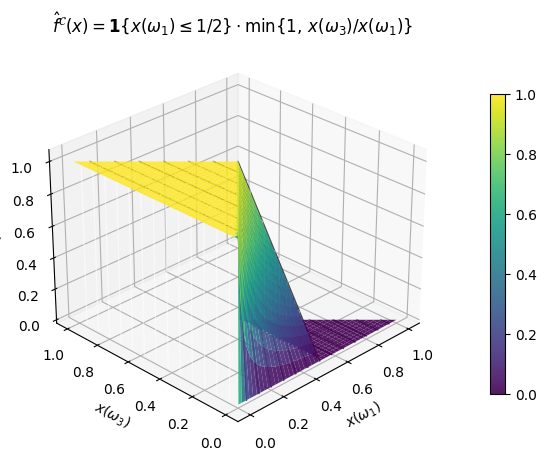}
\caption{Privacy-preserving envelope $\widehat{f}^{\C}$}
\label{fig1b}
\end{subfigure}
\caption{Comparison of envelopes}
\label{fig:persuasion}
\end{figure}

\paragraph{Sequential Sampling.}\hspace{-2mm}Another class of technologies that is automatically composition-closed is \textit{\textbf{dynamic sampling/stopping}} technologies. For example, \citet{brocas2007influence}, \citet{henry2019research}, \citet{morris2019wald}, and \citet{ball2021experimental} consider the classic Wald sampling model (\citealt{Wald}) adopted in information design settings. We describe a general setting in discrete time. Fix a finite state space $\Omega$. The belief process $(x_n)_{n \geq 0}$ taking values in $X = \Delta(\Omega)$ under any monitoring technology in the style of \citet{Wald} would evolve according to a time-homogeneous Markov martingale with some transition kernel $Q$.

A designer chooses a \textit{\textbf{randomized stopping time}} $\tau$ with
respect to the natural filtration of $(x_n)_{n \geq 0}$. The induced posterior distribution is the law of the stopped belief $x_\tau$. Every such $\tau$ generates a transition kernel $P^\tau$ on $X$ defined by 
\[P^\tau(\,\cdot\, \mid x) = \P\big[x_\tau \in \,\cdot\, \mid x_0 = x\big]\,.\]
Let 
\[\mathcal{P} := \Big\{P^\tau: \tau \text{ is a randomized stopping time for $(x_n)_{n\geq 0}$}\Big\}\,.\]
Then, it is evident that $\mathcal{P}$ is composition-closed and regular. Therefore, by \Cref{thm:main-compare}, there exists an order $\preceq_\C$ stronger than the Blackwell order on $\Delta(X)$ such that 
\[\mu \preceq_\C \nu \iff \text{there exists a randomized stopping time $\tau$ s.t. $x_\tau  \sim \nu$ with $x_0 \sim \mu$\,.}\]
Therefore, any dynamic sampling/stopping problem is equivalent to a \textit{\textbf{reduced-form constrained information design}} for which \Cref{cor:envelope} and \Cref{prop:comparison} hold immediately---in particular, they admit envelope characterizations and satisfy structural properties just like in the unconstrained case studied in \citet{Kamenica2011}. Moreover, as \Cref{prop:comparison} shows, unless the constrained problem is value-eliminating,  the induced posterior distributions for any such problem must weakly Blackwell dominate the KG signal with a binary state, and cannot be strictly Blackwell dominated by the KG signal with any finite state. 

What is the appropriate envelope here? By \Cref{cor:envelope}, the appropriate envelope $\widehat{f}$ is defined by $[-\Phi(\mathcal{P})]$-envelope of $f$. By definition, $g \in -\Phi(\mathcal{P})$ if and only if 
\[\int g(y) P(\d y\mid x) \leq g(x) \text{ for all $x$ and all $P \in \mathcal{P}$}\,,\]
which is equivalent to 
\[\int g(y) P^\tau(\d y\mid x) \leq g(x) \text{ for all $x$ and all randomized stopping times $\tau$}\,.\]
Since the Markov process is time-homogeneous, note that the above is equivalent to 
\[\int g(y) Q(\d y\mid x) \leq g(x) \text{ for all $x$}\,,\]
where $Q$ is the transition kernel of the Markov process. In the language of optimal stopping theory, such a function is called a \textit{\textbf{superharmonic/excessive function}}. This provides a point of connection to the well-known \textit{\textbf{Dynkin's theorem}} (\citealt{dynkin1963optimal}) in optimal stopping. It is further generalized in  \citet{dayanik2003optimal} and states that the value function  
\[V = \inf\Big\{g : g \text{ is excessive and } g \geq f \Big\}, \tag{\textbf{Dynkin}}\]
which our results immediately recover. 

Combining \Cref{thm:main-compare} and the above observation about time-homogeneous Markov chains, we also immediately obtain: 
\[\int g \d \mu \geq \int g \d \nu \text{ for all excessive $g$} \iff \text{$\exists$ $\tau$ s.t. $x_\tau  \sim \nu$ with $x_0 \sim \mu$\,,} \tag{\textbf{Rost}}\]
which is the well-known \textit{\textbf{Rost's theorem}} (\citealt{rost1971stopping}) in the theory of Skorokhod embeddings for Markov
processes. Rost's theorem is extremely useful for identifying the set of inducible posterior distributions from stopping a Markov chain, and has been applied in, e.g., \citet{morris2019wald}. 

More generally, our results apply to the cases where (i) the Markov chain need not be time-homogeneous and (ii) the choice need not be an optimal stopping problem but a \textit{\textbf{control problem}}. Indeed, suppose in each period, the designer can choose from a feasible set of experiments (e.g., as studied in \citealt*{ni2023sequential}), which corresponds to a base set of Markov transition kernels. By the same logic as above, it follows immediately that the problem admits an envelope characterization, and the inducible posterior belief distributions can be characterized by an integral stochastic order.

\subsection{Which Updating Rules Admit Consistent Orders?}
In \Cref{sec:Bayes-plausible}, we show that any Blackwell-consistent Bayes-plausible order on $\Delta(X)$ must be a strengthening of the Blackwell order. In particular, if the updating rule is fixed to be Bayes' rule, then the garbling order is necessary for there to be a Blackwell-consistent value-based order of information. Given this result, a natural question arises: are there other Blackwell-consistent value-based orders of information if the updating rule is allowed to be non-Bayesian? Without Bayesian updating, orders on $\Delta(X)$ need not be Bayes-plausible. We now apply \Cref{thm:main2} and \Cref{thm:main-compare} to characterize all possible updating rules that can admit a consistent order $\preceq$ on the set of posterior belief distributions (which need not be Bayes-plausible). It turns out that if we want to compare all experiments that are comparable by Blackwell's garbling criterion, then essentially the only possible updating rules are homeomorphic to Bayes' rule. To begin with, we first introduce a general framework for non-Bayesian updating. 

\subsubsection{A General Framework for Non-Bayesian Updating}
Let $\Omega$ be a finite set of states, and let $X:=\Delta(\Omega)$ be the set of beliefs over $\Omega$. \citet{cripps2018divisible} develops a general theory for \textit{\textbf{distorted updating rules}} which distort Bayesian posterior beliefs. A \textit{\textbf{(prior-dependent) systematic distortion}} of belief updates is a continuous function $d:X\times X \to X$, where $d(x, y) \in X$ denotes the distorted posterior belief when the prior is $x$ and the Bayesian posterior is $y$. For example, $d^{\mathrm{B}}(x, y) := y$ corresponds to the Bayesian updating rule without any distortion. We impose several regularity conditions on $d$: (i) $d$ is continuous; (ii) for any $x \in X^0:= \text{int } \Delta(\Omega), z \mapsto d(x,z)$ is injective on $X^0$; (iii) $d(x,x)=x$ for all $x \in X$.\footnote{These conditions, together with the convention that distorted updating rules are defined in terms of beliefs, correspond to Axioms 1-3 of \citet{cripps2018divisible}.} Let $\mathcal{D}$ be the set of such updating rules. 

For any $d \in \mathcal{D}$, it is convenient to define
\[
\hat{d}(x_B,x_D,y):=d\left(x_D,\frac{x_D \oslash x_B \odot y}{\langle x_D \oslash x_B,y\rangle}\right)\,,
\]
for all $x_B,x_D,y \in X^0$, where $\oslash$ denotes the component-wise division of two vectors, and $\odot$ denotes the component-wise product of two vectors.\footnote{Since $d(x,y)$ is continuous, $\hat{d}$ can be extended continuously on $X \times X \times X$.} The function $\hat{d}(x_B,x_D,y)$ corresponds to the distorted belief where the distorted prior is $x_D$, and the Bayesian posterior is $y$, which is updated from a (potentially different) Bayesian prior $x_B \in X$. When the Bayesian prior is different from the distorted prior, the Bayesian posterior updated from the distorted prior is distinct from, but is determined by (up to a likelihood ratio), the Bayesian posterior updated from the Bayesian prior $x_B$. The second argument of the right-hand side makes this correction. It follows immediately from the definition that 
\[
\hat{d}(x_B,x_B,y)=d(x_B,y)\,,\quad \quad \mbox{ and } \quad \quad \hat{d}(x_B,x_D,x_B)=x_D\,.
\]
That is, when the Bayesian prior $x$ is the same as the distorted prior, the distorted posterior given the Bayesian posterior is exactly given by $d(x,y)$. Moreover, when the Bayesian posterior equals the Bayesian prior $x_B$ (i.e., there is no Bayesian updating), the distorted posterior must be the same as the distorted prior $x_D$ (i.e., there is no distorted updating either). 

While an updating rule $d \in \mathcal{D}$ describes only how beliefs will be updated (relative to Bayesian updating) from a single signal, it implies a natural description of how beliefs will be updated from multiple signals that arrive sequentially. To see this, for any updating rule $d \in \mathcal{D}$, let $d_{\mathrm{II}}:X \times X \times X \to X$ be defined by 
\[
d_{\mathrm{II}}(x,z;y):=\hat{d}(y,d(x,y),z)=d\left(d(x,y),\frac{d(x,y)\oslash y \odot z}{\langle d(x,y)\oslash y,z \rangle}\right)\,,
\]
for all $x,y,z \in X^0$.\footnote{Since $d(x,y)$ is continuous, $d_{\mathrm{II}}$ can be extended continuously on $X \times X \times X$.} Under this definition, $d_{\mathrm{II}}(x,z;y)$ describes the terminal distorted posterior belief under a two-stage updating process, where one signal is drawn from a Blackwell experiment first, and then a second signal is drawn from another Blackwell experiment (which could potentially depend on the realization of the first signal). In this process, the prior is $x$, the first-stage Bayesian posterior belief is $y$ (and hence the first-stage interim distorted posterior belief is $d(x,y)$), and the second-stage Bayesian posterior belief is $z$. Clearly, if the updating rule is Bayesian, then sequentially updating must lead to the same outcome as if updating is done at once, and the terminal posterior belief would not depend on the realization of the interim belief $y$, conditional on $z$ itself. That is, 
\[
d_{\mathrm{II}}^\mathrm{B}(x,z;y)=d^\mathrm{B}(x,z)=z\,,
\]
for all $x,y,z \in X^0$. We refer to updating rules that satisfy the same property as being \textit{\textbf{divisible}}. That is, we say that $d\in \mathcal{D}$ is divisible if $d(x,\cdot)$ is bijective for all $x \in X$, and
\[
d_{\mathrm{II}}(x,z;y)=d(x,z)\,,\tag{\textbf{Divisibility}}
\]
for all $x,y,z \in X^0$. Divisible updating rules have a compact characterization, due to \citet{cripps2018divisible}. Namely, for $d \in \mathcal{D}$, $d$ is divisible if and only if there exists a homeomorphism $F:X \to X$ such that
\[
d(x,y)=F^{-1}\left(\frac{F(x)\oslash x \odot y}{\langle F(x)\oslash x,y \rangle}\right)
\]
for all $x,y \in X^0$. When $F$ is the identity function, $d(x,y)$ is the Bayesian updating rule. In a sense, a divisible updating rule $d$ is \emph{essentially} Bayesian: it distorts the prior via a homeomorphism $F$, updates via Bayes' rule, and distorts the posterior back again via $F^{-1}$. 

\subsubsection{Consistent Comparisons under Non-Bayesian Updating}
Let $\Pc_M$ be the set of all martingale kernels. For any $P \in \Pc_M$ and for any $x \in X$, we may interpret $P(\,\cdot\,\mid x)$ as the distribution over (Bayesian) posterior beliefs induced by a Blackwell experiment when the (Bayesian) prior is $x$. As noted before, the garbling order of Blackwell experiments can be represented by the order $\preceq^{\Pc_M}$ on $\Delta(X)$. Therefore, for there to be a Blackwell-consistent, value-based order on $\Delta(X)$, there has to be a class of problems under which the expected value under any $P(\,\cdot\,\mid x)$ is higher than the value at $x$, for all priors $x$. When $d \in \mathcal{D}$ is the Bayesian updating rule, by Blackwell's theorem, the class of problems are exactly decision problems (i.e., convex functions of posterior beliefs). When the updating rule $d \in \mathcal{D}$ is non-Bayesian, however, this is generally not the case. 

To evaluate the value of information under a potentially non-Bayesian updating rule $d \in \mathcal{D}$, a natural perspective is to evaluate the expected value of a problem \emph{\textbf{paternalistically}}, as in \citet{whitmeyerforthcomingBlackwell}. From the paternalistic perspective, a decision is made under the distorted belief, whereas the expected value is computed using the Bayesian belief. Therefore, a problem would be described by a function $g:X \times X \to \R$, where $g(x_B,x_D)$ denotes the indirect value of a problem when the Bayesian belief, which is used to evaluate the expected value, is $x_B$, and when the distorted belief, which is used to determine the action, is $x_D$. For example, if the problem is a decision problem, with action set $A$ and payoff function $u$, then 
\[
g(x_B,x_D)=\E_{\omega \sim x_B} [u(\omega,a^*(x_D))]
\]
is the indirect utility, where $a^*(x_D)$ is the optimal action under belief $x_D$. However, the problem need not be a decision problem; for example, it can be a non-Bayesian persuasion problem where $a^*(x_D)$ is the optimal action taken under the distorted belief $x_D$ and an objective $\tilde{u}$ different from $u$, with the payoff evaluated under the Bayesian posterior $x_B$ and the sender's objective $u$.  

For the garbling order to have a Blackwell-consistent, value-based order, the relevant class of problems must be such that 
\[
\int g(y, \hat{d}(x_B,x_D,y))P(\d y\mid x_B) \geq g(x_B,x_D)\,.
\]
That is, given any pair of Bayesian and distorted priors $(x_B,x_D)$, the expected paternalistic value of the problem with more information---given by $P(\,\cdot\,\mid x_B)$---must be higher than the value without any information. Let $\C^d \subseteq B(X \times X)$ be such functions. 

Naturally, one might conjecture that the class of problems $\C^d$ dualizes the garbling-based comparison under the updating rule $d$. However, our next result shows that this is in general not the case. In fact, we show the garbling-based comparison admits a Blackwell-consistent value-based order if and only if $d$ is divisible. As a result, there is generally no instrumental-value-based representation of the garbling order once the updating rule is non-Bayesian, with the only exception being the case when the updating rule is divisible (i.e., being ``essentially'' Bayesian).   

For any $(x_B,x_D)$, and for any $P \in \Pc_M$, let $Y$ be the random variable that follows $P(\,\cdot\,\mid x_B)$ and consider the distribution $Q$ of 
\[
(Y,\hat{d}(x_B,x_D,Y))\,.
\]
If we interpret $x_B$ as the Bayesian prior, $x_D$ as a distorted prior, and $P(\,\cdot\,\mid x_B)$ as a Blackwell experiment (identified by the distribution of \emph{Bayesian} posterior beliefs), then $Q(\,\cdot\,\mid x_B,x_D)$ is the joint distribution of the Bayesian and $d$-distorted posterior beliefs under the same experiment. Let $\mathcal{Q}_0$ be the set of transition kernels on $X\times X$. Let 
\[\mathcal{Q}^{d}_{(x_B, x_D)}:= \Big\{ \text{Law}(Y, \hat{d}(x_B,x_D,Y)): \text{ $Y$ is a random variable with $\E[Y]=x_B$ }\Big\}\,.\]
Let 
\[
\cQ^d:=\Big\{Q \in \cQ_0:  Q(\,\cdot\,\mid x_B, x_D) \in \mathcal{Q}^{d}_{(x_B, x_D)} \text{ for all $x_B, x_D$} \Big\}
\]
denote the set of all distributions of posterior beliefs under the distorted updating rule $d$. Note that $\cQ^d$ is regular. Moreover, $\cQ^d$ defines a binary relation on $\Delta(X\times X)$, where $\mu \preceq^{\cQ^d} \nu$ if and only if $\nu=Q*\mu$ for some $Q \in \cQ^d$. By the definition of $\cQ^d$, the comparison given by $\preceq^{\cQ^d}$ encodes the garbling order of Blackwell experiments: for any $(x_B, x_D)$, by construction, $\delta_{(x_B, x_D)} \preceq^{\cQ^d} \nu$ if and only if $\nu$ is the joint  distribution of the Bayesian posterior and the distorted posterior under the updating rule $d$ for \textit{some} experiment $P$.
 
Our next result characterizes updating rules that admit consistent comparisons in the Blackwell sense by applying \Cref{thm:main2} to the space $X \times X$. 

\begin{prop}[Consistent Comparisons under Non-Bayesian Updating]\label{prop:divisible}
The garbling-based comparison $\preceq^{\cQ^d}$ is a Blackwell-consistent order if and only if $d$ is divisible.  
\end{prop}

Naturally, one might conjecture that the garbling comparison always has a dual representation in terms of its instrumental value just like in Blackwell's theorem under Bayes' rule. However, \Cref{prop:divisible} shows that this is \textit{false}. According to \Cref{prop:divisible}, when the updating rule $d$ is not divisible, there cannot be a value-based order that represents the garbling-based comparison $\preceq^{\cQ^d}$. In particular, there exist distributions of posterior beliefs $\mu,\nu \in \Delta(X \times X)$ such that $\nu=Q*\mu$ and the instrumental value for any problem $g \in \C^d$ is increased: 
\[
\int g(x_B,x_D) \mu(\d x_B, \d x_D) \leq \int g(x_B,x_D) \nu (\d x_B, \d x_D)\,,
\]
and yet $Q \not\in \mathcal{Q}^d$, i.e., $Q$ cannot be represented via a garbling of an experiment under the updating rule $d$. Moreover, \Cref{prop:divisible} shows that when $d$ is divisible, $\preceq^{\cQ^d}$ is a Blackwell-consistent order---equivalent to the value-based order $\preceq_{\C^d}$---and thus Blackwell's theorem is preserved in this regard. 

Given that the garbling-based comparison only has a dual representation under divisible updating rules, which are homeomorphic to Bayes' rule, one might naturally further conjecture that the class of problems $\C^d$, which dualizes the garbling order under a divisible updating rule $d$, is the class of decision problems just like in Blackwell's theorem. However, this is again \textit{false}. Indeed, by the characterization of \citet{whitmeyerforthcomingBlackwell}, any non-Bayesian divisible updating rule is \textit{not} Blackwell-monotone, and hence there exists a decision problem under which a more informative experiment according to the garbling order leads to a lower paternalistic payoff. Nonetheless, by \Cref{prop:divisible}, a divisible updating rule is Blackwell-consistent, and therefore such a decision problem cannot appear in the class of problems that dualizes the garbling order under the divisible updating rule. In other words, combining \citet{whitmeyerforthcomingBlackwell} and \Cref{prop:divisible}, we then know that the class of problems that dualizes the garbling order under a divisible, non-Bayesian updating rule is \textit{not} the class of all decision problems. It necessarily excludes \textit{some} decision problem and may include \textit{some} persuasion problem or costly information acquisition problem---and yet, it dualizes the garbling order for a divisible updating rule according to \Cref{prop:divisible}. 

\subsubsection{Dynamic Values and Optimization under Non-Bayesian Updating}
Similar to \Cref{sec:Bayes-plausible}, by virtue of \Cref{thm:main}, our results have immediate implications for optimization problems. In particular, we now show how \Cref{prop:divisible} yields new insights into information design problems with non-Bayesian updating (see, e.g., \citealt{alonso2016bayesian,de2022non,azrieli2025sequential,kobayashi2025dynamic}). 

In information design problems with a non-Bayesian receiver, the sender's problem is equivalent to choosing a martingale $P \in \Pc_M$ to maximize 
\[
\int f(y,\hat{d}(x_B,x_D,y))P(\d y\mid x_B)
\]
for a fixed (potentially heterogeneous) prior $(x_B,x_D)$, where $f:X \times X \to \R$ is the (Bayesian) sender's indirect utility that depends on the sender's Bayesian belief and the receiver's $d$-distorted belief. The optimal value of such a persuasion problem, as shown by \citet{alonso2016bayesian} and \citet{de2022non}, can be geometrically characterized by the concave envelope of the function $y \mapsto f(y,\hat{d}(x_B,x_D,y))$ evaluated at the prior $(x_B,x_D)$. However, a crucial implication of \Cref{prop:divisible} is that, when the updating rule is not divisible, such envelopes do not correspond to the functions that dualize the garbling order. In particular, the envelope is prior-dependent. Indeed, by \Cref{thm:main-compare}, the sender may improve their payoff by sending signals \textit{\textbf{sequentially}}, and static/one-shot persuasion would not be optimal in general. 

To formalize this, for any $d \in \mathcal{D}$ and $x_B,x_D \in X$, consider the problem 
\begin{equation*}
\sup_{Q \in \cQ^{d} \otimes \cQ^d} \int f(y_B,y_D) Q(\d y_B, \d y_D \mid x_B,x_D)\,. \eqtag[-2.5em][0.8\baselineskip]{\textbf{Dynamic}}{eq:two-stage}
\end{equation*}
In particular, \weqref{\textbf{Dynamic}}{eq:two-stage} could be interpreted as a dynamic information design problem with a non-Bayesian receiver, potentially with a different prior than the Bayesian sender, in which case $f(y_B,y_D)=v(y_B,a^*(y_D))$ is the sender's payoff as a function of the sender's own Bayesian prior and the receiver's optimal action taken in the second period given the distorted posterior belief. When information is costless and the sender only cares about the receiver's action, $v(y_B, a) = \E_{\omega \sim x_B} [u(\omega, a)]$ would be affine in $y_B$; however, when information is costly to the sender as in \citet{gentzkow2014costly}, or when the sender cares about the information intrinsically for their own action, $v(y_B, a)$ would not be affine in $y_B$. More generally, our formulation captures all information design problems in which the sender's payoff depends on the Bayesian and distorted posterior beliefs.

When $d$ is the Bayesian updating rule, in our setting, it is well-known that dynamic information design does not have any additional value compared to one-shot information design. Without Bayesian updating, there could be a benefit for dynamic information design. We say that $d \in \mathcal{D}$ admits \textit{\textbf{no dynamic value}} if for all $x_B,x_D \in X$, and for any upper semicontinuous $f: X\times X \to \R$, 
\[
\sup_{Q \in \cQ^{d} \otimes \cQ^d} \int f(y_B,y_D) Q(\d y_B, \d y_D\mid x_B,x_D)=\sup_{Q \in \cQ^d} \int f(y_B,y_D) Q(\d y_B, \d y_D\mid x_B, x_D)\,.
\]

\begin{cor}\label{cor:dynamic}
An updating rule $d \in \mathcal{D}$ admits no dynamic value if and only if $d$ is divisible. 
\end{cor}

The sufficiency part of \Cref{cor:dynamic} immediately reproduces the result of \citet{kobayashi2025dynamic}, and extends it to \emph{all} information design problems, including persuasion problems where the sender's payoff is state-dependent. Moreover, \Cref{cor:dynamic} also answers the question raised in \citet{azrieli2025sequential} and \citet{kobayashi2025dynamic} by proving the converse---whenever the updating rule is not divisible, there exists \textit{some} information design problem where the sender can strictly benefit from sending signals sequentially.

\section{Nested Optimization and Stackelberg Principals}\label{sec:stackelberg}

We now describe in more detail the nested optimization problem and give a few examples that arise in information design, decision theory, and mechanism design. 

\subsection{General Framework}
There are two principals $A$ and $B$. Principal $A$ is the \textit{\textbf{Stackelberg leader}} who chooses a measure $\mu \in M \subseteq \Delta(X)$, where $M$ is a compact convex set. Principal $B$ is the \textit{\textbf{Stackelberg follower}} who chooses a measure $\nu \preceq_{\C} \mu$ after principal $A$ has moved, where $\preceq_\C$ is an integral stochastic order. Principal $B$ solves a stochastic optimization problem as in \Cref{sec:abstract}: 
\[\max_{\nu \preceq_\C \mu}\int f(x) \nu(\d x)\,.\]
Recall that $X^\star_f(\mu)$ is the solution correspondence of this problem. 

Principal $A$ solves a nested optimization problem anticipating the effect on the follower. We assume that principal $B$ breaks ties in favor of principal $A$ so that principal $A$'s problem always has a solution. Formally, principal $A$ cares about the pair of measures $(\mu, \nu)$ selected at the end of the game, and thus solves the nested optimization as in \Cref{sec:abstract}: 
\begin{equation*}
\text{ } \qquad \max_{(\mu,\nu) \in G^M_f} W(\mu,\nu)\,, 
\end{equation*}
where $W$ is some continuous quasi-convex functional describing the principal's preferences and $G^M_f$ is the graph of the solution correspondence of principal $B$, i.e., 
\[G^M_f = \big\{(\mu,\nu): \mu \in M\,, \nu \in X^\star_f(\mu)\big\} \subseteq \Delta(X) \times \Delta(X)\,.\]
Formally, our solution concept is \textit{\textbf{leader-optimal subgame perfect equilibria}} (also known as strong Stackelberg equilibria).

By \Cref{cor:nested-optimization}, which is an immediate consequence of \Cref{thm:main}, we know that whenever $\C$ satisfies the min-closure property, there must exist an equilibrium in which (i) principal $A$ selects $\mu \in \Ext(M)$ and then (ii) principal $B$ selects $\nu \in \Ext(X^{\star}_f(\mu))$. 

Moreover, in the case where $W(\mu, \nu)$ is a linear functional and hence can be written as $\int w_A(x) \mu(\d x) + \int w_B(x) \nu(\d x)$, we can explicitly construct all equilibria by invoking \Cref{thm:main}. To see this, let  
\[w_B^\star(x):= \max_{\eta \in X^\star_f(\delta_x)} \int w_B(y) \eta(\d y) \quad \text{ and }\quad \mathcal{P}^\star_B(x):= \argmax_{\eta \in X^\star_f(\delta_x)} \int w_B(y) \eta(\d y)\,.\]
Then, we immediately have:  
\begin{cor}\label{cor:modified}
Suppose $\C$ is min-closed. The leader value of any leader-optimal equilibrium is
\[\max_{\mu \in M} \int (w_A(x) + w^\star_B(x)) \mu(\d x)\,.   \eqtag[-2.5em]{\textbf{Modified Objective}}{eq:modified}\]
Moreover, $(\mu, \nu)$ is a leader-optimal equilibrium if and only if 
\begin{itemize}
    \item[(i)] $\mu$ maximizes \weqref{\textbf{Modified Objective}}{eq:modified} ;
    \item[(ii)] $\nu = P * \mu$ for some $P$ such that $P(\,\cdot\,\mid x) \in \mathcal{P}^\star_B(x)$ for $\mu$-a.e. $x$.
\end{itemize}
\end{cor}

Moreover, if the convex set $M$ for principal $A$ is also described by a stochastic orbit of the form $\{\mu: \mu \preceq_{\C_A} \gamma\}$ where $\C_A$ is min-closed and $\gamma$ is an exogenous measure, then by \Cref{thm:main} again, the leader value is simply given by 
\[\int \overline{w^\star}(x) \gamma(\d x)\]
where $w^\star = w_A + w^\star_B$ and the envelope above is taken with respect to the set $\C_A$ and any leader-optimal strategy $\mu$ can only be supported on $\{x: \overline{w^\star}(x) = w^\star(x)\}$. While the above setup is described via the leader-optimal selection, any selection rule is compatible with our approach with the only difference being that the definition of $w^\star_B$ and $\Pc^\star_B$  will be modified accordingly.

\subsection{Economic Applications}

We now describe a series of examples where our general results apply. 

\paragraph{Sequential Persuasion.}\hspace{-2mm}Consider the following setting adapted from \citet{LiNorman2021SequentialPersuasion}. There is a finite state space $\Omega$. There are two senders. Sender $A$ designs a Blackwell experiment first, and then after observing $A$'s choice, sender $B$ designs an additional Blackwell experiment that can depend on the signal realization of $A$'s experiment. The receiver observes both experiments and both realized signals and takes an action from a finite set. This model corresponds to the case where $X = \Delta(\Omega)$, $\preceq_\C$ is the concave order, and $M = \{\mu: \mu \preceq_\C \delta_{x_0}\}$ where $x_0$ is the common prior about the state. As an immediate consequence of our results (\Cref{cor:modified}), there always exists a \textit{\textbf{sequential extremal equilibrium}} in which every sender $i$ selects an extreme point $\mu_i$ of $\text{MPS}(\mu_{i-1})$ where $\mu_{i-1}$ is the belief distribution selected by the previous sender. In particular, in this equilibrium, along any history, \textit{every} sender sends a signal with support size at most $|\Omega|$ \textit{as if} the subsequent senders do not exist, unlike in the \textit{\textbf{one-shot equilibria}} characterized in \citet{LiNorman2021SequentialPersuasion} in which the first sender splits beliefs in the set of ``stable beliefs'' and the subsequent senders all stay silent (also known as \textit{\textbf{silent equilibria}} in \citealt{wu2023sequential}).

\paragraph{Robust Persuasion.}\hspace{-2mm}Consider the following setting adapted from \citet{DworczakPavan2022RobustPersuasion}. Again, there is a finite state space $\Omega$ and two senders who move sequentially just like in sequential persuasion. However, the second sender, who chooses an additional experiment after observing the realized signal of the first sender, is the \textit{\textbf{adversarial nature}} whose objective is to minimize the payoff of the first sender. Because the second mover has completely opposite preferences to the first mover, the leader-optimal selection rule here has no bite: The set of all SPEs---which correspond to the \textit{\textbf{worst-case optimal signals}} in \citet{DworczakPavan2022RobustPersuasion}---is exactly the leader-optimal SPEs, characterized by \Cref{cor:modified}. In particular, the modified objective in \Cref{cor:modified} is given by the \textit{\textbf{lower convex envelope}} of the initial sender's indirect utility function $V$ (assuming the receiver has a unique best response at any posterior belief). Since the initial sender's problem has the constraint set $M = \{\mu: \mu \preceq_\C \delta_{x_0}\}$ also given by a stochastic orbit, the leader value of this game must be given by the concavification of the lower convex envelope $V_{\text{lcx}}$, evaluated at the prior $x_0$, which is the \textit{\textbf{full information payoff}} where the initial sender just releases all information. Moreover, a posterior distribution $\mu \in \Delta(X)$ is induced by a worst-case optimal signal if and only if \[\supp(\mu) \subseteq \Big\{x:  V_{\text{lcx}}(x) = \widehat{V_{\text{lcx}}}^{\text{cav}}(x)  = V_{\text{full}}(x)\Big\} =: Y\,,\]
where $V_{\text{full}}(x)$ is the affine function describing the payoff of the initial sender given prior $x$ and release of full information. \citet{DworczakPavan2022RobustPersuasion} consider a refinement of worst-case optimal signals, the \textit{\textbf{robust solutions}}, which maximize the initial sender's indirect utility function among the worst-case optimal signals and therefore corresponds to solving a standard persuasion problem with a support restriction that $\supp(\mu) \subseteq Y$. For any non-affine $V_{\text{lcx}}$, note that any interior belief $x \in \text{int}(\Delta(\Omega))$ cannot be in $Y$ because then by construction it implies that $V_{\text{lcx}}$ is an affine function. Therefore, as observed in \citet{DworczakPavan2022RobustPersuasion}, any worst-case optimal signal must induce \textit{\textbf{separation}} of states where the support of any worst-case optimal $\mu$ must rule out some states. In fact, for any subset $B \subset \Omega$, the same argument implies that for any $x \in \text{int}(\Delta(B))$ to be in $Y$,  $V_{\text{lcx}} \mid_{\Delta(B)}$ must be an affine function coinciding with $V_{\text{full}} \mid_{\Delta(B)}$ which happens if and only if $V\mid_{\Delta(B)} \geq V_{\text{full}} \mid_{\Delta(B)}$; conversely, if $V_{\text{lcx}} \mid_{\Delta(B)}$ coincides with $V_{\text{full}} \mid_{\Delta(B)}$, then any $x \in \Delta(B)$ is in $Y$. Together, these characterize a set $\mathcal{F}$ of subsets $B \subset \Omega$ such that $x \in Y$ if and only if $\supp(x) \in \mathcal{F}$, which is the main separation theorem (Theorem 1) of \citet{DworczakPavan2022RobustPersuasion}.

\paragraph{Objective Ambiguity Aversion.}\hspace{-2mm}Consider the following setting adapted from \citet{olszewski2007preferences} and \citet{Ahn2008AmbiguityWithoutStateSpace} who study \textit{\textbf{objective ambiguity aversion}} of a decision maker by interpreting ambiguity as a menu of lotteries. The set of outcomes $X$ is a compact set in $\R^n$. A \textit{\textbf{lottery}} $\mu \in \Delta(X)$ specifies a distribution over outcomes. A \textit{\textbf{menu}} $A$ is a set of lotteries. \citet{olszewski2007preferences} and \citet{Ahn2008AmbiguityWithoutStateSpace} characterize preferences over menus of lotteries and formalize the notion of objective ambiguity aversion. Specifically, under a standard set of axioms, \citet{olszewski2007preferences}  shows that the representation of objective ambiguity-averse preference takes the following form:
\[
\mathcal{V}(A) := \alpha \cdot \max_{\mu \in A} \int u(x) \mu(\d x) + (1 - \alpha) \cdot \min_{\mu \in A} \int u(x) \mu(\d x)\,,
\]
for some function $u$, and some constant $\alpha \in [0, 1]$. In other words, the objective ambiguity-aversion representation takes the Hurwicz max-min form: the decision maker evaluates menus by computing the ``best-case'' payoff and the ``worst-case'' expected utility (evaluated by $u$) from lotteries in that menu, and takes a convex combination of the best and the worst case payoff. For any $u$ and any $\alpha \in [0,1]$, the representation $\mathcal{V}$ defines a preference $\preceq$ by 
\[
A \preceq  B \iff \mathcal{V}(A) \leq \mathcal{V}(B)\,,
\]
for all menus $A$ and $B$. 

Suppose now that the menus are parametrized by a single lottery via a stochastic order. That is, each feasible menu $A$ takes the form of $A=\{\nu: \nu \preceq_\C \mu\}=:A_\C(\mu)$, for some convex cone $\C$ of test functions that contains $-1,1$ and can be approximated by continuous functions from above. As menus are parametrized by single lotteries, both the ambiguity-aversion preference $\preceq $ and its representation $\mathcal{V}$ immediately imply a preference $\widetilde{\preceq} $ and a representation $\widetilde{V}$ over lotteries, via 
\[
\mu_1 \widetilde{\preceq} \mu_2 \iff A_\C(\mu_1) \preceq A_\C(\mu_2)\,
\]
and 
\[
\widetilde{V}(\mu):=\mathcal{V}(A_\C(\mu))\,.
\]
Two natural questions then arise: (i) when can the representation $\widetilde{V}$ be distinguished from a vNM \textit{expected utility representation} over lotteries (i.e., when is it the case that $\widetilde{V}(\mu)=\int \hat{u} \d \mu$ for some $\hat{u}:X \to \R$)? (ii) When can the preference $\widetilde{\preceq} $ be distinguished from an \textit{expected utility preference} (i.e., when does the preference $\widetilde{\preceq}$ violate the vNM expected utility axioms)? 

Applying our results, we can answer these questions. Indeed, under this parameterization of menus, note that the representation of an ambiguity-averse decision maker is exactly a Stackelberg leader who chooses a lottery $\mu$ anticipating the nature as a follower to choose $\nu \preceq_\C \mu$---with probability $\alpha$, the nature wants to maximize the DM's payoff (an optimistic conjecture) and with probability $1-\alpha$, the nature wants to minimize the DM's payoff (a pessimistic conjecture). 

To state our result, recall that $\C$ is not min-closed if there exist $g_1, g_2 \in \C$ such that $\min\{g_1,g_2\} \notin \C$. By \Cref{thm:main}, for any continuous $f:X \to \R$, and for any $\mu \in \Delta(X)$, the dual minimization problem 
\[
\min_{g \in \C,\, g \geq f} \int g \d \mu \eqtag[-2.5em]{\textbf{Dual}}{eq:dual}
\]
has a unique solution (up to a $\mu$-measure zero set) whenever $\C$ is min-closed (and the solution equals the $\C$-envelope of $f$). Motivated by this, we say that $\C$ is \textit{\textbf{strongly non-min-closed}} if there exist a continuous function $f:X \to \R$ and a measure $\mu_0 \in \Delta(X)$ that is absolutely continuous with respect to the Lebesgue measure and has a density bounded away from zero, such that the dual minimization problem \weqref{\textbf{Dual}}{eq:dual} has two distinct continuous solutions $g_1, g_2$ that are not affinely related. Note that $\C$ is not min-closed whenever it is strongly non-min-closed.\footnote{Indeed, suppose that $\C$ is min-closed. Then $\min\{g_1,g_2\} \in \C$ must also solve the dual problem, which implies that $\int g_1 \d \mu_0=\int g_2 \d \mu_0=\int \min\{g_1,g_2\} \d \mu_0$, and hence $g_1=g_2$ $\mu_0$-a.e., a contradiction.}  As an example, the set of convex functions (which contains all affine functions, and is not min-closed) is in fact strongly non-min-closed.\footnote{For example, consider any convex $X$ with affine dimension at least $2$. Let $x_0$ be the barycenter of the Lebesgue measure on $X$. Take two non-affinely-related affine functions $l_1, l_2$ that cross at $x_0$: $l_1(x_0)=l_2(x_0)$. Let $f:=\min\{l_1,l_2\}$, and let $\mu_0$ be the Lebesgue measure. By Jensen's inequality, every $g \in \C$ with $g \geq f$ yields $\int g \d \mu_0 \geq g(x_0) \geq f(x_0)$, and $\int l_1 \d \mu_0=\int l_2 \d \mu_0=l_1(x_0)=l_2(x_0)=f(x_0)$. Thus, both $l_1$ and $l_2$ are solutions of the dual problem. One can also verify that the set of convex functions is strongly non-min-closed even if $X$ has affine dimension $1$.}

\begin{prop}\label{prop:ambiguity}
Consider stochastic-orbit ambiguity sets defined by $\preceq_\C$. Then: 
\begin{enumerate}
\item [(i)] There exists an ambiguity-aversion representation $\widetilde{V}$ that is not an expected utility representation if and only if $\C$ is not min-closed. 
\item [(ii)] There exists an ambiguity-aversion preference $\widetilde{\preceq}$ that is not an expected utility preference if $\C$ is strongly non-min-closed. 
\end{enumerate}
\end{prop}
Consequently, whenever $\C$ is min-closed, any ambiguity-averse preference, when translated to the space of lotteries, can in fact be represented by an expected utility \emph{as if} it were an expected utility preference. For example, when the ambiguity menus are given by mean-preserving spread orbits, every ambiguity-aversion preference over these menus has an expected utility representation. Conversely, an outside observer who observes all binary choices over lotteries can distinguish ambiguity-aversion preferences from expected utility preferences whenever $\C$ is strongly non-min-closed. For instance, when the ambiguity menus are given by the mean-preserving contraction orbits, an outside observer who observes all binary choices over lotteries can distinguish an ambiguity-averse DM from an expected-utility DM who faces no ambiguity. 

\paragraph{Property Right Design.}\hspace{-2mm}Consider the following setting adapted from \citet{DworczakMuir2024}. There is a \textit{\textbf{seller}} and a \textit{\textbf{buyer}} who contract on the sale of a single good. Let $x \in [0, 1]$ denote the allocation probabilities and $t \in \R_+$ denote the transfers. The buyer has private value $\theta \in [\underline{\theta}, \overline{\theta}] \subset \R_+$ for the good. The seller maximizes some linear objective $\E[V_S(\theta) x(\theta) + \lambda_S t(\theta)]$, where $\lambda_S \geq 0$, by posting a menu $M_S = \{(x, t)\}$ of lotteries and prices for the buyer to choose. However, before the seller moves, a \textit{\textbf{designer}} first designs the property right by posting an initial menu $M_D = \{(x, t)\}$ that the buyer treats as \textit{\textbf{outside options}} and hence can select if declining to contract with the seller. The designer maximizes a different linear objective $\E[V_D(\theta) x(\theta) + \lambda_D t(\theta)]$, where $\lambda_D \geq 0$. 

Without loss of generality, we may assume the seller always induces the buyer to contract with them. Thus, the seller's problem is equivalent to a mechanism design problem with a \textit{\textbf{type-dependent participation constraint}} $U(\theta) \geq R(\theta)$, where $U(\theta)$ is the indirect utility function of the buyer under the seller's menu and $R(\theta)$ is the indirect utility function of the buyer under the designer's menu. Equivalently, as in \citet{DworczakMuir2024}, the seller maximizes 
\[\max_{x,\, \underline{u} \geq 0}\int_{\underline{\theta}}^{\overline{\theta}} W_S(\theta) x(\theta) \d \theta - \alpha \underline{u}\]
subject to $x$ being non-decreasing, and 
\[\underline{u} + \int_{\underline{\theta}}^{\theta} x(s) \d s \geq \underline{u}_0 + \int_{\underline{\theta}}^{\theta} x_0(s) \d s \, \text{ for all $\theta$}\,,\]
where $W_S(\theta)$ is an appropriate virtual value function for the seller, $\underline{u}_0 = R(\underline{\theta})$ and $x_0(\theta) = R'(\theta)$. Anticipating the seller solving this problem, the designer chooses $(\underline{u}_0, x_0)$ in the first stage where $x_0$ is non-decreasing and $\underline{u}_0 \geq 0$. Theorem 1 of \citet{DworczakMuir2024} shows that despite the designer not directly controlling the allocation rule and having a different objective from the seller, the designer's optimal property right menu is very simple: $M_D = \{(1, p)\}$ interpreted as an \textit{\textbf{option-to-own}} property right where the buyer can always own the good at a price of $p$. \citet{DworczakMuir2024} proves this result by first developing a \textit{\textbf{generalized ironing procedure}} that solves the seller's problem subject to the type-dependent IR constraints, and then uncovering a surprising feature that the seller's optimal mechanism depends on the outside-option function $R(\theta)$ only through a \textit{\textbf{linear transformation}} via the ironing procedure. Thus, there exists a designer's solution that is an extreme point of the feasible $R(\theta)$ functions, corresponding to a menu $\{(1, p)\}$. 

We now describe how to understand this result via our framework and the trapezoid property. Since we assume the transfers are non-negative, as is standard in this environment, any IC mechanism can be equivalently represented as a \textit{\textbf{randomized posted price}} (with the price potentially lower than $\underline{\theta}$).\footnote{Indeed, any such mechanism can be extended to be an IC mechanism on $[0,\overline{\theta}]$ where type $0$ gets $0$ payoff given non-negative transfers. Thus, there exists a randomized posted price to replicate the original mechanism.} Thus, consider the lotteries of the posted prices $\mu, \nu \in \Delta([0,\overline{\theta}])$ chosen by the designer and the seller, respectively. In this language, the type-dependent IR constraint is exactly that 
\[\int \max\{\theta - p, 0\} \nu(\d p) \geq \int \max\{\theta - p, 0\} \mu(\d p) \,\text{ for all $\theta \in [0,\overline{\theta}]$ }\,,\]
which is exactly that 
\[\nu \succeq_{\text{decreasing convex}} \mu \]
or equivalently 
\[\nu \preceq_{\text{increasing concave}} \mu\,.\]
Let $\C$ be the set of \textit{\textbf{increasing concave}} functions on $[0, \overline{\theta}]$. Both the seller and the designer's objectives are linear functionals of $\nu$. Therefore, the seller solves the problem 
\[\max_{\nu: \nu \preceq_\C \mu} \int f(p) \nu(\d p)\,,\]
for some function $f$. The feasible set of $\mu$ here is unrestricted, i.e., $M = \Delta([0, \overline{\theta}])$. The designer solves the problem
\[\max_{(\mu, \nu) \in G^M_f} \int w(p) \nu(\d p)\,,\]
for some function $w$. Since the set of increasing concave functions is min-closed, by
\Cref{cor:nested-optimization}, we immediately obtain that there exists $(\mu^\star, \nu^\star)$ solving the nested optimization such that $\mu^\star \in \Ext(M)$, i.e., $\mu^\star = \delta_p$ for some $p$, and $\nu^\star \in \Ext(\{\nu: \nu \preceq_\C \mu^\star\})$ which takes the form of $\nu^\star = \alpha \delta_{p_1} + (1 - \alpha) \delta_{p_2}$ where $p_1 \leq p_2$. Therefore, the designer's optimal menu takes the form $\{(1, p)\}$, i.e., an option-to-own design, and then subsequently the seller's optimal menu takes the form $\{(\alpha, \alpha p_1), (1, \alpha p_1 + (1 - \alpha) p_2)\}$. 

\section{Conclusion}\label{sec:conclusion}

We study optimization problems in which a linear functional is maximized over probability measures that are dominated by a given measure according to an integral stochastic order in an arbitrary dimension. We show that the following four properties are equivalent for any such order: (i) the test function cone is closed under pointwise minimum, (ii) the value function is affine, (iii) the solution correspondence has a convex graph with decomposable extreme points, and (iv) every ordered pair of measures admits an order-preserving coupling. As corollaries, we derive the extreme and exposed point properties involving integral stochastic orders such as multidimensional mean-preserving spreads and stochastic dominance. Applying these results, we generalize Blackwell's theorem by completely characterizing the comparisons of experiments that admit two equivalent descriptions---through instrumental values and through information technologies. We also show that these results immediately yield new insights into information design, mechanism design, and decision theory.

\appendix 
\section{Omitted Proofs}
\subsection{Technical Lemmas}
\subsubsection{Technical Lemmas for \texorpdfstring{\Cref{sec:abstract}}{Section 3}}
\begin{lemma}\label{lem:conti-approx}
There exists a countable set of continuous functions $\widehat{\C} \subseteq \C \cap C(X)$ such that for all $g \in \C$, there exists a sequence $\{g_n\} \subseteq \widehat{\C}$ such that $g_n \geq g$ for all $n$, and $g_n \rightarrow g$ pointwise.  
\end{lemma}

\begin{proof}
Since $X$ is a compact metric space, $C(X)\cap \C$ is separable under the sup-norm. Take a countable dense subset $\widetilde{\C}=\{h_m\}_{m=1}^\infty$ of $\C \cap C(X)$. Let 
\[
\widehat{\C}:=\{h_m+q: m \in \N\,, q \in \Q_+ \}\,.
\]
Since $\C$ is a convex cone that contains all constants, and since the sum of a continuous function and a constant is continuous, it follows that $\widehat{\C} \subseteq \C \cap C(X)$. Clearly, $\widehat{\C}$ is also countable. 

We first claim that for any $h \in \C \cap C(X)$ and for any $\varepsilon>0$, there exists $\hat{h} \in \widehat{\C}$ such that $h < \hat{h} < h+\varepsilon$. Indeed, for any such $h \in \C \cap C(X)$ and $\varepsilon$, take $m \in \N$ so that $\|h_m-h\|<\nicefrac{\varepsilon}{3}$ and $q \in \Q_+ \cap (\nicefrac{\varepsilon}{3},\nicefrac{2\varepsilon}{3})$. Then, for all $x \in X$, 
\[
h_m(x)+q > h(x)-\frac{\varepsilon}{3}+\frac{\varepsilon}{3}=h(x)\,,
\]
and 
\[
h_m(x)+q < h(x)+\frac{\varepsilon}{3}+\frac{2\varepsilon}{3}=h(x)+\varepsilon\,,
\]
proving the claim. 

Now, consider any $g \in \C$. Since $g$ is approximated by continuous functions in $\C$ from above, there exists $\{g_n\} \subseteq \C \cap C(X)$ such that $\{g_n\} \downarrow g$. For each $n \in \N$, take $\hat{h}_n \in \widehat{\C}$ such that $g_n<\hat{h}_n<g_n+2^{-n}$. Together, we have 
\[
g \leq g_n<\hat{h}_n<g_n+2^{-n}\,,
\]
for all $n \in \N$. Therefore, for all $x \in X$, 
\[
\limsup_{n \to \infty} |\hat{h}_n(x)-g(x)| \leq \limsup_{n \to \infty}\Big(|g_n(x)-g(x)|+2^{-n}\Big)=0\,.
\]
Therefore, the sequence $\{\hat{h}_n\} \subseteq \widehat{\C}$ satisfies that $\hat{h}_n \geq g$ for all $n$ and $\hat{h}_n \rightarrow g$ pointwise, as desired. 
\end{proof}

\begin{lemma}\label{lem:envelope}
Suppose $\C$ is min-closed. Fix any upper semicontinuous function $f:X \to \R$. Let 
\[
\overline{f}(x):=\inf \Big\{g(x): g \in \C\,, g \geq f\Big\}\,,
\]
for all $x \in X$. Then $\overline{f} \in \C$. 
\end{lemma}

\begin{proof}
Fix a countable approximating class $\widehat{\C} \subseteq \C \cap C(X)$ from \Cref{lem:conti-approx}. Let 
\[\widehat{\C}_f := \Big\{h \in \widehat{\C}: h \geq f\Big\}\]
which is non-empty by \Cref{lem:conti-approx} since $f$ is bounded from above and $\C$ contains all constants and thus there exists some $h \in \widehat{\C}$ with $h \geq f$. Enumerate $\widehat{\C}_f$ as $\{h_n\}_{n=1}^\infty$ and write 
\[k_n := \min\{h_1, \dots, h_n\}\,.\]
Since $\C$ is min-closed, note that $k_n \in \C \cap C(X)$. Moreover, the sequence $\{k_n\}$ is decreasing with $k_n \geq f$. Fix any $x \in X$. We claim that  
\[\lim_{n \rightarrow \infty} k_n(x) = \inf_{h \in \widehat{\C}_f} h(x) = \inf_{g \in \C\,, g \geq f} g(x) = \overline{f}(x)\,.\]
Indeed, the first equality holds by definition, and $\inf_{h \in \widehat{\C}_f} h(x) \geq \overline{f}(x)$ follows by that $\widehat{\C}_f \subseteq \{g \in \C: g \geq f\}$. Now, to see the reverse inequality, fix any $g \in \C$ with $g \geq f$, and any $\varepsilon > 0$. By \Cref{lem:conti-approx}, there exists $h \in \widehat{\C}$ such that $g(x) \leq h(x) \leq g(x) + \varepsilon$. Since $g \geq f$, we have $h \geq f$ and hence $h \in \widehat{\C}_f$. Therefore, 
\[\inf_{h \in \widehat{\C}_f} h(x) \leq g(x) + \varepsilon\,.\]
Taking $\varepsilon \downarrow 0$ gives  
\[\inf_{h \in \widehat{\C}_f} h(x) \leq g(x)\,,\]
for all $g \in \C$ with $g \geq f$. Therefore, we have 
\[\inf_{h \in \widehat{\C}_f} h(x) \leq \inf_{g \in \C,\, g \geq f} g(x) = \overline{f}(x)\,,\]
as claimed. Thus, $\{k_n\} \subseteq \C \cap C(X)$ is a bounded decreasing sequence that converges pointwise to $\overline{f}$. Since $\C$ is closed under pointwise convergence from above, $\overline{f} \in \C$. 
\end{proof}

\begin{lemma}\label{lem:compactK}
For any $\mu \in \Delta(X)$, the set 
\[\K_\mu:= \Big\{\nu \in \Delta(X):  \nu \preceq_\C \mu\Big\}\]
is closed and hence compact. 
\end{lemma}
\begin{proof}
For any $\{\nu_n\} \subseteq \mathcal{K}_\mu$ such that $\{\nu_n\} \to \nu$ for some $\nu \in \Delta(X)$, since $\nu_n \preceq_\C \mu$, 
\[
\int g \d \nu_n \leq \int g \d \mu
\]
for all $g \in \C$. Therefore, for all $g \in \C$ that is continuous, 
\[
\int g \d \nu =\lim_{n \to \infty} \int g \d \nu_n \leq \int g \d \mu\,.
\]
Since for any $g \in \C$, there exists a sequence of continuous functions $\{g_m\} \subseteq \C$ such that $\{g_m\} \downarrow g$. As a result, by the monotone convergence theorem, 
\[
\int g \d \nu=\lim_{m \to \infty} \int g_m \d \nu \leq \lim_{m \to \infty} \int g_m \d \mu=\int g \d \mu\,.
\]
Therefore, $\nu \in \mathcal{K}_\mu$. That is, $\mathcal{K}_\mu\subseteq \Delta(X)$ is closed, and hence, compact.     
\end{proof}

\begin{lemma}\label{lem:ri-contact-ae}
Let $X\subset \mathbb{R}^n$ be a full-dimensional  compact convex set. Consider any upper semicontinuous concave function $f:X \to \R$. For each differentiability point $x$ of $f$, let
\[
h^x(y):=f(x)+\langle \nabla f(x),  y-x \rangle
\]
be the supporting hyperplane of $f$ at $x$, and let
\[
R_{h^x}:=\Big\{y\in X:f(y)=h^x(y)\Big\}\,.
\]
Then, for Lebesgue-almost every $x \in X$, we have 
\[
x\in \mathrm{ri}(R_{h^x})\,.
\]
\end{lemma}

\begin{proof}
Define $g:\R^n\to(-\infty,+\infty]$ by
\[
g(y):=
\begin{cases}
-f(y), & \text{ if $y\in X$}\,,\\
+\infty, & \text{ otherwise}\,.
\end{cases}
\]
Since $f$ is upper semicontinuous and concave on the compact convex set $X$, the function $g$ is proper, lower semicontinuous, and convex. Its conjugate is therefore
\[
g^*(p)=\sup_{y\in \R^n}[\langle p,y\rangle-g(y)]
      =\sup_{y\in X}[\langle p,y\rangle+ f(y)],
\]
which is convex, as it is the pointwise supremum of a family of affine functions. It is also finite for every $p \in \R^n$ as $X$ is compact and $f$ is upper semicontinuous. 

Let $D$ denote the set of differentiability points of $f$ in the interior of the full-dimensional $X$. Since $f$ is concave and $\partial X$ has Lebesgue measure zero, $D$ is a full-Lebesgue-measure subset of $X$. Thus, it suffices to show that
\[
x\in \mathrm{ri}(R_{h^x})
\]
for a full-Lebesgue-measure subset of $D$. 

To this end, fix any $x\in D$ and let
\[
p_x:=-\nabla f(x)\,.
\]
Since $\nabla f(x)$ is a supergradient of $f$ at $x$,
\[
f(y)\leq f(x)+\langle \nabla f(x), y-x\rangle
\]
for all $y \in X$. Equivalently, we have
\[
\langle p_x, y\rangle-g(y)\leq \langle p_x,x\rangle-g(x)
\]
for all $y \in \R^n$. Note that the equality holds if and only if $y\in R_{h^x}$---indeed, if the equality holds, then we must have $y \in X$, and 
\[f(y)= f(x)+\langle \nabla f(x), y-x\rangle = h^x(y)\,,
\]
and thus $y \in R_{h^x}$; conversely, if $y \in R_{h^x}$, then $y \in X$ and the above also hold. Therefore, we have 
\[
R_{h^x}=\argmax_{y\in \R^n}[\langle p_x,y\rangle-g(y)]\,.
\]

We claim that, for every $p\in\R^n$, we have 
\[
\argmax_{y\in \R^n}[\langle p,y\rangle-g(y)]
=\partial g^*(p)\,,
\]
where $\partial g^*(p)$ is the subdifferential of $g^*$ at $p$. Indeed, for any 
\[y \in \argmax_{y' \in \R^n}[\langle p, y' \rangle-g(y')]\,,\]
we have
\[
g^*(p)=\langle p,y\rangle-g(y)\,,
\]
and therefore for every $q\in\R^n$, by definition of $g^*$, we have 
\[
g^*(q)\geq \langle q,y\rangle-g(y) = g^*(p)+\langle y, q-p\rangle\,,
\]
and hence $y\in \partial g^*(p)$. Conversely, if $y\in \partial g^*(p)$, then we have 
\[
g^*(q)\geq g^*(p)+\langle y,q-p\rangle\,,
\]
for all $q \in \R^n$. Rearranging, we have
\[
\langle q,y\rangle-g^*(q)\leq \langle p,y\rangle-g^*(p)
\]
for all $q \in \R^n$. Since $g$ is proper, lower semicontinuous, and convex, we have $g = g^{**}$. Now, taking the supremum of the above over $q$ and using $g = g^{**}$, we obtain
\[
g(y) = g^{**}(y) \leq \langle p,y\rangle-g^*(p)\,.
\]
By Fenchel's inequality, we also have 
\[
g(y) \geq \langle p,y \rangle - g^*(p)\,.
\]
Therefore, we conclude that 
\[
g^*(p)=\langle p,y\rangle-g(y)\,,
\]
and thus $y$ attains the maximum of $\langle p, y' \rangle-g(y')$ over all $y' \in \R^n$, proving the claim.

Applying the previous claim with $p=p_x$ yields
\[
 \partial g^*(p_x) = \argmax_{y \in \R^n}[\langle p_x, y\rangle - g(y)] = R_{h^x} \,.
\]
Therefore, we have 
\[
B := \Big\{x \in D : x \notin \mathrm{ri}\big(R_{h^x}\big) \Big\} = \Big\{x \in D : x \in \mathrm{rb}\big(R_{h^x}\big) \Big\} = \Big\{x \in D : x \in \mathrm{rb}\big(\partial g^*(p_x)\big) \Big\}\,,
\]
where $\mathrm{rb}(A)$ denotes the relative boundary of a convex set $A$. Hence, we conclude that 
\[
B \subseteq \bigcup_{p\in \R^n}\mathrm{rb}\big(\partial g^*(p)\big)\,.
\]
Since $g^*: \R^n \to \R$ is convex, by Theorem 3.3 of \citet{drusvyatskiy2011generic}, we conclude that 
\[
\bigcup_{p\in \R^n}\mathrm{rb}\big(\partial g^*(p)\big)
\]
has Lebesgue measure zero. Thus, $B$ has Lebesgue measure zero, as desired. 
\end{proof}

\subsubsection{Technical Lemmas for \texorpdfstring{\Cref{sec:compare}}{Section 4}}
\begin{lemma}\label{lem:rectangle}
For any rectangular $\Pc,\Pc'$, $\Pc=\Pc'$ if and only if $\preceq^{\Pc}=\preceq^{\Pc'}$.
\end{lemma}
\begin{proof}
The ``only if'' direction is clear. To prove the ``if'' direction, fix any $x \in X$. Note that for any $P \in \Pc_0$, $\nu=P*\delta_x$ if and only if $\nu=P(\,\cdot\,\mid x)$. Since $\preceq^{\Pc}=\preceq^{\Pc'}$, we have 
\[
\{\nu: \nu=P*\delta_x \text{ for some $P \in \Pc$}\}=\{\nu:  \delta_x \preceq^\Pc \nu \} = \{\nu: \delta_x  \preceq^{\Pc'} \nu \} = \{\nu: \nu=P*\delta_x \text{ for some $P \in \Pc'$}\}\,.
\]
Thus, 
\[
\Pc_x = \Big\{P(\,\cdot\,\mid x): P \in \Pc \Big\}=\Big\{P(\,\cdot\,\mid x): P \in \Pc'\Big\}= \Pc'_x \,.\]
Since this holds for all $x$, and $\Pc, \Pc'$ are rectangular, it follows that $\Pc=\Pc'$, as desired. 
\end{proof}

\begin{lemma}\label{lem:basic}
The operators $\Phi$ and $\Psi$ have the following properties: 
\begin{enumerate}
\item [(a)] $\Phi$ and $\Psi$ are decreasing in the set-inclusion order. 
\item [(b)] For any $\C \subseteq \C_0$ that admits continuous approximations from below, $\Psi(\C) \subseteq \Pc_0$ is convex, rectangular, composition-closed, and has a closed graph.  
\item [(c)] For any $\Pc \subseteq \Pc_0$, $\Phi(\Pc)$ is a max-closed convex cone that contains $1,-1$ and is closed under bounded increasing pointwise limits, with $\Phi(\Pc) \cap C(X)$ being closed under the sup-norm. 
\item [(d)] For any $\C \subseteq \C_0$ , $\C \subseteq \Phi(\Psi(\C))$
\item [(e)] For any $\Pc \subseteq \Pc_0$, $\Pc \subseteq \Psi(\Phi(\Pc))$. 
\end{enumerate}
\end{lemma}

\begin{proof}
$(a)$: Consider any $\C, \C'$ such that $\C \subseteq \C' \subseteq \C_0$. For any $P \in \Psi(\C')$, we have $g*P \geq g$ for all $g \in \C'$ by definition. Since $\C \subseteq \C'$, we have $g*P \geq g$ for all $g \in \C$ as well. Thus, $P \in \Psi(\C)$. Likewise, consider any $\Pc,\Pc'$ such that $\Pc \subseteq \Pc' \subseteq \Pc_0$. For any $g \in \Phi(\Pc')$, we have $g*P \geq g$ for all $P \in \Pc'$ by definition. Since $\Pc \subseteq \Pc'$, we have $g*P \geq g$ for all $P \in \Pc$ as well. Thus, $g \in \Phi(\Pc)$. \\
\mbox{}

\noindent $(b)$: It is easy to see that $\Psi(\C)$ is convex and rectangular by definition. To see that the graph of $x \mapsto \Psi(\C)_x$ is closed, note that the graph can be written as:
\[\text{Gr} = \bigcap_{g \in \C \cap C(X)} \Big\{(x, \eta): \int g \d \eta \geq g(x) \Big\}\]
by the assumption that any $g \in \C$ can be approximated from below via $\{g_n\} \subseteq \C \cap C(X)$ and the monotone convergence theorem. Since the graph is an intersection of a collection of closed sets, it is closed. To see that $\Psi(\C)$ is composition-closed, take any $P_1,P_2 \in \Psi(\C)$ and note that for any $g \in \C$ and for any $x \in X$
\[
[g * (P_1 \circ P_2)] (x)=\int \left(\int g(z)P_1(\d z\mid y)\right)P_2(\d y\mid x) \geq \int g(y) P_2(\d y\mid x) \geq g(x)\,.
\]
Thus, $P_1 \circ P_2 \in \Psi(\C)$.\\
\mbox{}

\noindent $(c)$: It is easy to see that $\Phi(\Pc)$ is a convex cone and that $1,-1 \in \Phi(\Pc)$ since $P(\,\cdot\,\mid x)$ is a probability measure for all $x \in X$. To see that $\Phi(\Pc)$ is max-closed, consider any $g_1,g_2 \in \Phi(\Pc)$. For any $P \in \Pc$ and for any $x \in X$, 
\begin{align*}
[\max\{g_1,g_2\}*P](x) = \int \max\{g_1(y),g_2(y)\}P(\d y\mid x) \geq& \max\left\{\int g_1(y)P(\d y \mid x), \int g_2(y) P(\d y\mid x)\right\} \\
\geq& \max\{g_1(x),g_2(x)\}\,.
\end{align*}
Thus, $\max\{g_1,g_2\} \in \Phi(\Pc)$. 

To see that $\Phi(\Pc) \cap C(X)$ is closed under the sup-norm, consider any sequence $\{g_n\} \subseteq \Phi(\Pc) \cap C(X)$ and suppose that $g_n \to g$ under the sup-norm. Since $C(X)$ is closed under the sup-norm, $g \in C(X)$. Moreover, since $g_n \in \Phi(\Pc)$ for all $n$, for any $P \in \Pc$ and for any $x \in X$, we have 
\begin{align*}
    0 &\geq \limsup_{n \to \infty} \left[g_n(x)-\int g_n(y)P(\d y \mid x)\right]\\
    &= g(x)-\int g(y)P(\d y \mid x)\,,
\end{align*}
where the equality follows from $\{g_n\} \to g$ in the sup-norm, and hence pointwise, as well as the dominated convergence theorem. 

Lastly, to see that $\Phi(\Pc)$ is closed under bounded increasing pointwise limits, consider any $\{g_n\} \subseteq \Phi(\Pc)$ such that $\{g_n\} \uparrow g$ for some $g$. Then, since $g$ is the limit of an increasing sequence of lower semicontinuous functions, $g$ is lower semicontinuous. Combining the previous argument and the monotone convergence theorem, we have $g \in \Phi(\Pc)$. \\
\mbox{}

\noindent $(d)$: Consider any $\C \subseteq \C_0$ and any $g \in \C$. For any $P \in \Psi(\C)$, by definition, $g*P \geq g$. Therefore, $g \in \Phi(\Psi(\C))$.   \\
\mbox{}

\noindent $(e)$: Consider any $\Pc \subseteq \Pc_0$ and any $P \in \Pc$. For any $g \in \Phi(\Pc)$, by definition $g*P \geq g$. Therefore, $P \in \Psi(\Phi(\Pc))$.
\end{proof}

\begin{lemma}\label{lem:compactm}
For any regular $\Pc$ and for any $\mu \in \Delta(X)$, let 
\[
M(\mu):=\{P*\mu: P \in \Pc\}\,.
\]
Then, $M(\mu)$ is a compact and convex subset of $\Delta(X)$. 
\end{lemma}

\begin{proof}
Clearly, $M(\mu)$ is convex since $\Pc$ is regular, and in particular, convex. To see that $M(\mu)$ is compact, it suffices to show that $M(\mu)$ is closed, since $M(\mu) \subseteq \Delta(X)$ and $\Delta(X)$ is compact. To this end, consider any sequence $\{\nu_n\} \subseteq M(\mu)$ such that $\{\nu_n\} \to \nu$ for some $\nu \in \Delta(X)$. We show that $\nu=P*\mu$ for some $P \in \Pc$. To begin with, for each $n \in \N$, since $\nu_n \in M(\mu)$, there exists $P_n \in \Pc$ such that $\nu_n=P_n*\mu$.  

Since $C(X)$ is separable, there exists a countable set $\{h_m\} \subseteq C(X)$ with $\|h_m\| \leq 1$ that is dense in the unit ball $\{h \in C(X): \|h\| \leq 1\}$. For each $m \in \N$ and $n \in \N$, let 
\[
g_{n,m}(x):=\int_X h_m(y)P_n(\d y\mid x)\,,
\]
for all $x \in X$. Since $\|h_m\| \leq 1$ and since $P_n(\,\cdot\,\mid x) \in \Delta(X)$, we have that $|g_{n,m}(x)| \leq 1$ for all $n,m\in \N$ and for all $x \in X$. Note that, by definition, 
\[
\int_X g_{n,m}\d \mu=\int_X \left(\int_X h_m(y)P_n(\d y\mid x)\right) \mu(\d x)=\int_X h_m \d P_n*\mu=\int_X h_m \d \nu_n\,, 
\]
for all $n,m \in \N$, and hence, since $\{\nu_n\} \to \nu$, 
\begin{equation}\label{eq:converge}
\lim_{n \to \infty} \int_X g_{n,m} \d \mu=\lim_{n \to \infty} \int_X h_m \d \nu_n=\int_X h_m \d \nu\,,
\end{equation}
for all $m \in \N$. 

Since $g_{n,m}$ has bounded $L^1(\mu)$-norm for all $n,m \in \N$, by Koml\'{o}s' theorem, there exists a subsequence $\{g_{n_j^{(1)},1}\}\subseteq \{g_{n,1}\}$ such that the Ces\`{a}ro sum 
\[
\frac{1}{k} \sum_{j=1}^k g_{n_j^{(1)},1}
\]
converges as $k \to \infty$, $\mu$-almost surely. Then, for this subsequence, by Koml\'{o}s' theorem again, there exists a further subsequence $\{g_{n_j^{(2)}, 2}\}\subseteq \{g_{n_j^{(1)}, 2}\} \subseteq \{g_{n, 2}\}$ such that the Ces\`{a}ro sum converges. Recursively, for each $m$, we obtained a subsequence $\{g_{n_j^{(m)},m}\} \subseteq \{g_{n_j^{(m-1)},m}\} \subseteq \{g_{n,m}\}$ such that the Ces\`{a}ro sum converges $\mu$-almost surely. Now, using the Cantor diagonalization argument, and the fact that countable union of $\mu$-measure zero sets is still $\mu$-measure zero, there exists a subsequence $\{n_j\}$ and a measurable set $X_\infty \subseteq X$ such that $\mu(X_\infty)=1$ and that the Ces\`{a}ro sum 
\[
\ub{g}_{k,m}(x):=\frac{1}{k} \sum_{j=1}^k g_{n_j,m}(x)
\]
converges as $k \to \infty$ for all $m$ and all $x \in X_\infty$. Let $g_m$ be (in the sense of $\mu$-a.s. equivalence class) the pointwise limit of $\{\ub{g}_{k,m}\}$ for all $m \in \N$. As a $\mu$-a.s. pointwise limit, $g_m$ is measurable. Moreover, since $|g_{n,m}| \leq 1$ for all $n,m \in \N$, $|g_m| \leq 1$ for all $m \in \N$ as well. 

In the meantime, for each $k \in \N$, and for the same subsequence $\{n_j\}$, let 
\[
\ub{P}_k:=\frac{1}{k}\sum_{j=1}^k P_{n_j}\,.
\]
Since $\Pc$ is regular, and thus is convex, $\ub{P}_k \in \Pc$ for all $k \in \N$. In particular, $\{\ub{P}_k(\,\cdot\,\mid x)\}_{k=1}^\infty \subseteq \Pc_x$ for all $x \in X$. Then, by definition, 
\begin{equation}\label{eq:ubg}
\ub{g}_{k,m}(x)=\frac{1}{k} \sum_{j=1}^k g_{n_j,m}(x)=\frac{1}{k} \sum_{j=1}^k \int_X h_m(y) P_{n_j}(\d y\mid x)=\int_X h_m(y) \ub{P}_k(\d y\mid x)\,,
\end{equation}
for all $x \in X$, and for all $k,m \in \N$. Since $\Pc$ has a closed graph, $\Pc_x$ is closed and hence compact. Fix any $x \in X_\infty$. Since $\Pc_x$ is compact, there exists a subsequence $\{\ub{P}_{k_j}(\,\cdot\,\mid x)\}$ that converges to some $\eta_x$. We claim that every convergent subsequence of $\{\ub{P}_k(\,\cdot\,\mid x)\}$ converges to the same $\eta_x$, and thus---since $\Pc_x$ is a compact metric space---the entire sequence $\{\ub{P}_k(\,\cdot\,\mid x)\}$ converges to $\eta_x$. Indeed, for any $m \in \N$, for any two convergent subsequences $\{P_{k^{(1)}_j}(\,\cdot\,\mid x)\}$ and $\{P_{k^{(2)}_j}(\,\cdot\,\mid x)\}$  of $\{\ub{P}_k(\,\cdot\,\mid x)\}$, let $\eta^{(1)}_x$ and $\eta^{(2)}_x$ be their limits, respectively. By definition, we have
\[
\lim_{j \to \infty} \int_X h_m(y) \ub{P}_{k^{(i)}_j}(\d y\mid x)=\int_X h_m \d\eta^{(i)}_x\,,
\]
for $i \in \{1,2\}$ and for all $m \in \N$. In the meantime, since $x \in X_\infty$, $\{\ub{g}_{k,m}(x)\} \to g_m(x)$ as $k \to \infty$. Together with \eqref{eq:ubg}, we have
\begin{equation}\label{eq:eta}
\int_X h_m \d \eta^{(i)}_x=\lim_{j \to \infty} \int_X h_m(y) \ub{P}_{k^{(i)}_j}(\d y\mid x)=\lim_{j \to \infty} \ub{g}_{k^{(i)}_j,m}(x)=g_m(x)\,,
\end{equation}
for $i \in \{1,2\}$, and for all $m \in \N$. In particular, for all $m \in \N$, 
\[
\int_X h_m \d\eta^{(1)}_x=\int_X h_m \d\eta^{(2)}_x\,.
\]
Moreover, since $\{h_m\}$ is dense in $\{h\in C(X): \|h\| \leq 1\}$, for any $h \in C(X)$ with $\|h\| \leq 1$, there exists $\{h_{m_l}\} \subseteq \{h_m\}$ such that $\|h_{m_l}-h\| \to 0$ as $l \to \infty$. Therefore, for any continuous function $h \in C(X)$ with $\|h\| \leq 1$, by the dominated convergence theorem, 
\[
\int_X h \d \eta^{(1)}_x=\lim_{l \to \infty} \int_X h_{m_l} \d \eta_x^{(1)}=\lim_{l \to \infty} \int_X h_{m_l} \d \eta^{(2)}_x=\int_X h \d \eta^{(2)}_x\,.
\]
Thus, for any $h \in C(X)$ with $h \neq 0$,
\[
\int_X h \d \eta^{(1)}_x=\|h\| \cdot \int_X \frac{h}{\|h\|} \d \eta^{(1)}_x=\|h\| \cdot \int_X \frac{h}{\|h\|} \d \eta^{(2)}_x=\int_X h \d \eta^{(2)}_x\,.
\]
As a result, by Riesz's representation theorem, $\eta^{(1)}_x=\eta^{(2)}_x=\eta_x$. Furthermore, as $\{\ub{P}_{k}(\,\cdot\,\mid x)\} \to \eta_x$ and $\Pc_x$ is closed, we have that $\eta_x \in \Pc_x$ as well. Thus, for all $x \in X_\infty$, $\{\ub{P}_k(\,\cdot\,\mid x)\} \to \eta_x$ for some $\eta_x \in \Pc_x$. Since $\Pc$ is rectangular and  $x \mapsto \Pc_x$ has a closed graph and compact values, by the Kuratowski–Ryll-Nardzewski measurable selection theorem, there exists a measurable selection $s: X \to \Delta(X)$ so that $s(x) \in \Pc_x$ for all $x \in X$. Let
\[
P(\,\cdot\,\mid x):=\begin{cases}
\eta_x,&\mbox{if } x \in X_\infty\,;\\
s(x),&\mbox{otherwise.} 
\end{cases}
\]
Since $x \mapsto \eta_x$ is the pointwise limit of the transition kernels $\{\ub{P}_k\}$ on $X_\infty$ and the pointwise limit of measurable functions into a metrizable space is measurable, we have that $x \mapsto \eta_x$ is measurable on $X_\infty$.  Moreover, since $s$ is measurable, $P$ by construction is a valid Markov kernel. Furthermore, since $\Pc$ is rectangular,  $\eta_x \in \Pc_x$ for all $x \in X_\infty$, and $s(x) \in \Pc_x$ for all $x \in X$, we have that $P \in \Pc$. 

Now, for each $k \in \N$, let 
\[
\ub{\nu}_k:=\ub{P}_k *\mu=\frac{1}{k}\sum_{j=1}^k P_{n_j}*\mu=\frac{1}{k} \sum_{j=1}^k \nu_{n_j}\,.
\]
Since $\{\nu_n\} \to \nu$, $\{\ub{\nu}_k\} \to \nu$ as well. For any $m,k \in \N$, by \eqref{eq:ubg}, we have 
\[
\int_X h_m \d \ub{\nu}_k= \int_X\left( \int_X h_m(y)\ub{P}_{k}(\d y\mid x)\right)\mu(\d x)=\int_X \ub{g}_{k,m} \d \mu\,.  
\]
In the meantime, since $\{\ub{g}_{k,m}\} \to g_m$ $\mu$-almost surely for all $m$, and since $\|\ub{g}_{k,m}\| \leq 1$, by the dominated convergence theorem, 
\[
\lim_{k \to \infty} \int_X \ub{g}_{k,m} \d \mu =\int_X g_m \d \mu\,.
\]
Together, since $\{\ub{\nu}_k\} \to \nu$, we have that for all $m \in \N$,
\begin{align*}
\int_X h_m \d \nu=\lim_{k \to \infty} \int_X h_m \d \ub{\nu}_k=\lim_{k \to \infty} \int_X \ub{g}_{k,m} \d \mu=&\int_X g_m \d \mu\\
=&\int_X \left(\int_X h_m(y)P(\d y\mid x)\right)\mu(\d x)\\
=& \int_X h_m \d P*\mu\,,
\end{align*}
where the fourth equality follows from \eqref{eq:eta} with the conclusion that $\eta^{(1)}_x=\eta^{(2)}_x=\eta_x$ for all $x \in X_\infty$ and from the definition of $P$---in particular, from $P(\,\cdot\,\mid x)=\eta_x$ for all $x \in X_\infty$. Again, since $\{h_m\}$ is dense in $\{h \in C(X): \|h\| \leq 1\}$, for any $h \in C(X)$ with $\|h\| \leq 1$, there exists $\{h_{m_l}\}\subseteq \{h_m\}$ such that $\|h_{m_l}-h\| \to 0$ as $l \to \infty$, and hence, by the dominated convergence theorem, 
\[
\int_X h \d \nu=\lim_{l \to \infty} \int_X h_{m_l} \d \nu=\lim_{l \to \infty} \int_X h_{m_l} \d P*\mu=\int_X h \d P*\mu\,.
\]
As a result, for any $h \in C(X)$ with $h \neq 0$,  
\[
\int_X h \d \nu=\|h\| \cdot \int_X \frac{h}{\|h\|} \d \nu=\|h\|\cdot \int_X \frac{h}{\|h\|} \d P*\mu=\int_X h \d P*\mu\,. 
\]
Thus, by Riesz's representation theorem again, we have $\nu=P*\mu$. Since $P \in \Pc$, it follows that $\nu \in M(\mu)$, as desired. 
\end{proof}

\begin{lemma}\label{lem:piecewiseconvexapprox}
Let $X$ be a compact convex metrizable subset of a locally convex Hausdorff topological
vector space and let $Y$ be a metric space. Let
$f: X \times Y \to \R$ be a bounded lower semicontinuous function where $x\mapsto f(x,y)$ is convex for all $y\in Y$. Then, there exists a sequence of continuous functions $f_n: X \times Y \to \R$ such that:
\begin{enumerate}
\item[(i)] $f_n \leq f_{n+1} \leq f$ for all $n$,
\item[(ii)] $f_n \uparrow f$ pointwise on $X\times Y$,
\item[(iii)] for every $y \in Y$, the function $x \mapsto f_n(x,y)$ is convex.
\end{enumerate}
\end{lemma}

\begin{proof}
Let $\mathcal{A}(X)$ denote the space of continuous affine functions on $X$.
Since $X$ is compact and metrizable, $C(X)$ is separable in the sup-norm. Since
$\mathcal{A}(X)$ is a closed subspace of $C(X)$, $\mathcal{A}(X)$ is also separable.
Fix a countable sup-norm dense subset
\[
\{a_k\}\subset \mathcal{A}(X).
\]

For each $k\in \N$, define
\[
m_k(y):=\min_{x\in X}(f(x,y)-a_k(x))\,,
\]
for all $y \in Y$.
Since $(x,y)\mapsto f(x,y)-a_k(x)$ is lower semicontinuous and $X$ is compact, the minimum is attained, and the map $m_k:Y\to\mathbb{R}$ is lower semicontinuous by Berge's theorem of maximum. 

Since $Y$ is a metric space and $m_k$ is real-valued and lower semicontinuous, there exists an increasing sequence of continuous functions
$\phi_{k,l}:Y\to\mathbb{R}$ such that $\{\phi_{k,l}\}_{l}\uparrow m_k$ pointwise on $Y$.

Now, for each $k$ and each $l$, let
\[
h_{k,l}(x,y):=a_k(x)+\phi_{k,l}(y).
\]
Then, by construction, $h_{k,l}$ is continuous on $X \times Y$, and for each fixed $y$,
$x\mapsto h_{k,l}(x,y)$ is affine. Moreover, since
$\phi_{k,l}\leq m_k$ for all $l$ and for all $k$, we have
\[
h_{k,l}(x,y)=a_k(x)+\phi_{k,l}(y) \leq a_k(x)+m_k(y) \leq f(x,y)
\]
for all $(x,y)\in X\times Y$, where the last inequality follows from the definition of $m_k$. Thus, for all $k$ and all $l$, we have that $h_{k,l} \leq f$ is a continuous function where $x \mapsto h_{k, l}(x, y)$ is affine for all $y$. 

Enumerate the countable family $\{h_{k,l}\}$ as
$\{h_j\}_{j\in \N}$, and define
\[
f_n(x,y):=\max\Big\{h_1(x,y)\,,\,\ldots\,,\,h_n(x,y)\Big\}\,,
\]
for all $n \in \N$ and for all $(x,y) \in X \times Y$. By construction, for each $n \in \N$, $f_n \leq f$ and $f_n$ is continuous on $X\times Y$. Moreover, for every fixed $y$, the map $x\mapsto f_n(x,y)$ is convex, since it is the pointwise maximum of finitely many affine functions. Also, $f_n\leq  f_{n+1}$ for all $n \in \N$.

It remains to prove that $f_n\uparrow f$ pointwise. Since $f_n\uparrow
\sup_j h_j$ pointwise, it is enough to show that
\[
\sup_{j\in \N} h_j(x,y)=f(x,y)
\]
for all $(x,y) \in X \times Y$. To see this, fix $(x_0,y_0)\in X\times Y$ and $\varepsilon>0$. Let
\[
g(x):=f(x,y_0)\,,
\]
for all $x \in X$. Then $g$ is a lower semicontinuous, convex function on $X$. We claim that there exists $a\in \mathcal{A}(X)$ such that $a \leq g$ and $a(x_0) \geq g(x_0)-\varepsilon/4$. Indeed, let $E \supset X$ be the ambient locally convex topological vector space, and extend $g$ to a function $\tilde g:E\to(-\infty,+\infty]$ on $E$ by
\[
\tilde{g}(x):=
\begin{cases}
g(x), &\text{ if $x\in X$},\\
+\infty, &\text{ otherwise}.
\end{cases}
\]
Since $X$ is compact and hence closed, $\tilde{g}$ is a proper lower semicontinuous convex function.
By the Fenchel–Moreau theorem (see, e.g., \citealt{rockafellar1970convex}), a proper lower semicontinuous convex function is the supremum of its continuous affine minorants.
Thus, we have 
\[\tilde{g}(x_0) = \sup\Big\{b(x_0): b\text{ is a continuous affine function on $E$}\,,\ b \leq \tilde{g}\Big\}\,.\]
Then, since for any continuous affine $b$ on $E$ with $b \leq \tilde{g}$, $b \mid_X$ is a continuous affine function on $X$ with $b \mid_X \leq g$, and since $x_0 \in X$, it follows that 
\[
g(x_0)= \tilde{g}(x_0) \leq \sup\Big\{a(x_0): a\in \mathcal{A}(X)\,,\ a \leq g\Big\} \leq g(x_0)\,.
\]
Thus, we have 
\[g(x_0) = \sup\Big\{a(x_0): a\in \mathcal{A}(X)\,,\ a \leq g\Big\}\,.\]
Therefore, there exists $a\in \mathcal{A}(X)$ such that $a \leq g$ on $X$ and $a(x_0) \geq g(x_0)-\varepsilon/4$, as desired. 

Now, since $\{a_k\}$ is dense in $\mathcal{A}(X)$, there exists $k \in \N$ such that $\|a_k-a\|_\infty<\varepsilon/4$. Then
\[
a_k(x_0)\geq a(x_0)-\varepsilon/4 \geq g(x_0)-\varepsilon/2
= f(x_0,y_0)-\varepsilon/2\,.
\]
Meanwhile, for every $x\in X$, we have 
\[
a_k(x)\leq a(x)+\varepsilon/4 \leq g(x)+\varepsilon/4
= f(x,y_0)+\varepsilon/4\,,
\]
which implies that 
\[
m_k(y_0)
=\min_{x\in X}(f(x,y_0)-a_k(x))
\geq -\varepsilon/4\,.
\]
Since $\{\phi_{k,l}(y_0)\}_{l} \uparrow m_k(y_0)$, there exists some $l \in \N$ such that
\[
\phi_{k,l}(y_0)\geq-\varepsilon/2 \,.
\]
Then, we have 
\[
h_{k,l}(x_0,y_0)
= a_k(x_0)+\phi_{k,l}(y_0)
\geq f(x_0,y_0)-\varepsilon\,.
\]
Since $h_{k,l} \leq f$, it follows that 
\[
f(x_0,y_0)-\varepsilon
\leq \sup_{j\in \N}h_j(x_0,y_0)
\leq f(x_0,y_0) \,.
\]
Taking $\varepsilon \downarrow 0$ gives 
\[
\sup_{j\in \N}h_j(x_0,y_0)=f(x_0,y_0) \,.
\]
Since $(x_0,y_0)$ is chosen arbitrarily, it follows that $f_n\uparrow f$ pointwise, as desired. 
\end{proof}

\subsection{Proof of \texorpdfstring{\Cref{thm:main}}{Theorem 1}}
$(a) \implies (b)$: Suppose $\C$ is min-closed. First, observe that for any $\nu \preceq_\C \mu$, 
\[
\int_X f \d \nu \leq \int_X \overline{f} \d \nu \leq \int_X \overline{f} \d \mu\,,
\]
where the second inequality follows from $\overline{f} \in \C$ by \Cref{lem:envelope}. Therefore, 
\[
\max_{\nu \in \Delta(X): \nu \preceq_\C \mu}\int_X f \d \nu \leq \int_X\overline{f} \d \mu\,.
\]

We now show that 
\[
\int_X \overline{f} \d \mu \leq \max_{\nu \in  \Delta(X): \nu \preceq_\C \mu}\int_X f \d \nu\,. 
\]
To this end, let $C(X)$ be the normed linear space of continuous functions on $X$ under the sup-norm. By Riesz's representation theorem, the dual $[C(X)]^*$ of $C(X)$ is isomorphic to the set of regular Radon measures on $X$.  Define $\Theta:C(X) \to \R \cup \{\infty\}$ and $\Xi:C(X) \to \R \cup \{\infty\}$ as: 
\[
\Theta(g):=\begin{cases}
0,&\mbox{if } g \geq f\\
+\infty,& \mbox{otherwise}
\end{cases}\,,
\quad \mbox{ and } \quad 
\Xi(g):=\begin{cases}
\int_X g \d \mu,& \mbox{if } g \in \tilde{\C}\\
+\infty,&\mbox{otherwise}
\end{cases}\,,
\]
where 
\[\tilde{\C} = \C \cap C(X)\,.\]
Note that $\tilde{\C}$ is a convex cone. Therefore, $\Theta$ and $\Xi$ are convex functions with values in $\R \cup \{+ \infty\}$. 

Next, define the dual functions $\Theta^*:[C(X)]^* \to \R \cup \{\infty\}$ and $\Xi^*:[C(X)]^* \to \R \cup \{\infty\} $
\[
\Theta^*(\nu):=\sup_{g \in C(X)}\left[\int_X g \d \nu-\Theta(g)\right]=\sup_{g \geq f} \int_X g \d \nu
\]
and 
\[
\Xi^*(\nu):=\sup_{g \in C(X)}\left[\int_X g \d \nu-\Xi(g)\right]=\sup_{g \in \tilde{\C}} \left(\int_X g \d \nu-\int_X g \d \mu\right)\,.
\]
Note that for any $\nu \in [C(X)]^*$ that is not a positive measure, there exists $\varphi \geq 0$ such that 
\[
\int_X \varphi \d \nu<0
\]
and hence by considering the sequence of functions $g_t :=  \varphi \cdot t + M$, where $M$ is a continuous function above $f$, we have 
\[
\inf_{g \geq f} \int_X g \d \nu=-\infty\,.
\]
Meanwhile, note that for any $\nu \in [C(X)]^*$ such that $\nu \npreceq_\C \mu$, there exists $g \in \C$ such that 
\[
\int_X g \d \nu > \int_X g \d \mu\,,
\]
which implies that the same holds for some $\tilde{g} \in \tilde{\C}$ as well since any such $g$ can be approximated by a decreasing sequence $\{g_n\} \subseteq \tilde{\C}$ that converges pointwise to $g$ and hence $\int g_n \d(\nu - \mu) \rightarrow \int g \d(\nu - \mu) > 0$ by the dominated convergence theorem. Thus, since $\tilde{\C}$ is a cone, we have 
\[
\inf_{g \in \tilde{\C}} \left(\int_X g \d \mu -\int_X g \d \nu\right)=-\infty\,.
\]
Together, since $1,-1 \in \tilde{\C}$,  it then follows that 
\[
\inf_{g \geq f} \int_X g \d \nu+\inf_{g \in \tilde{\C}}\left( \int_X g \d \mu-\int_X g \d \nu\right)>-\infty
\]
if and only if $\nu \in \Delta(X)$ and $\nu \preceq_\C \mu$. Lastly, note that for any $\nu \in \Delta(X)$ such that $\nu \preceq_\C \mu$, 
\[
\inf_{g \in \tilde{\C}} \left(\int_X g \d \mu -\int_X g \d \nu\right)=0\,.
\]
Consequently, we have 
\begin{align*}
\sup_{\nu \in [C(X)]^*}[-\Theta^*(-\nu)-\Xi^*(\nu)]=& \sup_{\nu \in [C(X)]^*} \left[\inf_{g \geq f} \int_X g \d \nu+\inf_{g \in \tilde{\C}}\left(\int_X g \d \mu-\int_X g \d \nu\right)\right]\\
=& \sup_{\nu \in \Delta(X): \nu \preceq_\C \mu} \left[\inf_{g \geq f} \int_X g \d \nu+\inf_{g \in \tilde{\C}}\left(\int_X g \d \mu-\int_X g \d \nu\right)\right]\\
=& \sup_{\nu \in \Delta(X): \nu \preceq_\C \mu} \inf_{g \geq f} \int_X g \d \nu\\
=& \max_{\nu \in \Delta(X):\nu \preceq_\C \mu} \int_X f \d \nu\,,
\end{align*}
where the last equality also uses that any bounded upper semicontinuous function on a compact metric space admits a decreasing sequence of continuous functions approximating it pointwise from above.

By \Cref{lem:envelope}, $\overline{f} \in \C$ and hence $\overline{f} + 1 \in \C$. Moreover, by assumption, there exists a sequence $\{h_n\}$ where $h_n \in \tilde{\C}$ such that $h_n$ converges to $\overline{f} + 1$ pointwise from the above. Let $g_0 = h_1$. Then $g_0 \in \tilde{\C}$ and $g_0 \geq f+1$. Note that at this $g_0$, we have that (i) $\Theta$ and $\Xi$ are finite and (ii) $\Theta$ is continuous at $g_0$ in the sup norm. Therefore, by the Fenchel-Rockafellar duality theorem (\citealt{villani2003topics}, Theorem 1.9), we have 
\begin{align*}
&\max_{\nu \in \Delta(X):\nu \preceq_\C \mu} \int_X f \d \nu\\
=& \sup_{\nu \in [C(X)]^*} \left[\inf_{g \geq f} \int_X g \d \nu+\inf_{g \in \tilde{\C}} \left(\int_Xg \d \mu-\int_X g \d \nu\right)\right]\\
=& \sup_{\nu \in [C(X)]^*}[-\Theta^*(-\nu)-\Xi^*(\nu)]\\
=&\inf_{g \in C(X)}[\Theta(g)+\Xi(g)]\\
=& \inf_{g \in \tilde{\C}, \, \, g \geq f} \int_X g \d \mu\\
\geq & \int_X \overline{f} \d \mu\,,
\end{align*}
where the last inequality follows from that for any $g \in \tilde{\C}$ and $g \geq f$, by the definition of $\overline{f}$, we have that for all $x$
\[g(x) \geq \overline{f}(x)\,.\]
This completes the proof. \\
\mbox{}

\noindent \noindent $(b) \implies (c)$: The proof is by the Krein-Milman theorem. Consider the orbit set 
\[\mathcal{K}_\mu = \Big\{ \nu \in \Delta(X): \nu \preceq_{\C} \mu \Big\}\,,\]
which is clearly convex. Moreover, by \Cref{lem:compactK}, $\mathcal{K}_\mu$ is also compact. 

We first claim that for any exposed point $\nu \in \Exp(\mathcal{K}_\mu)$, there exists some $P$ such that $\nu = P * \mu$ and $P *\delta_x \preceq_\C \delta_x$ for all $x$. To see this, consider any exposed point $\nu \in \Exp(\mathcal{K}_\mu)$. 
By definition, there exists some continuous $f$ such that $\nu$ is the unique solution to \weqref{\textbf{Optimization}}{eq:stochastic-optimization}. Thus,
\[ \int f \d \nu = V^\star_f(\mu) = \int V^\star_f(\delta_x) \mu(\d x) = \int \Big(\max_{\eta: \eta \preceq_\C \delta_x} \int f \d\eta\Big) \mu(\d x) \,,\]
where the second equality follows from $(b)$ since $V^\star_f(\mu)$ is an upper semicontinuous, and hence Baire-one, affine function and thus satisfies the barycentric formula (see Corollary 4.22 of \citealt{lukevs2009integral}). By the Kuratowski–Ryll-Nardzewski measurable selection theorem, the correspondence 
\[
x \mapsto \argmax_{\eta \preceq_\C \delta_x} \int f \d \eta
\]
admits a measurable selection. Fix any such measurable selection $x \mapsto \eta_x$, and define $P$ by $P(\,\cdot\, \mid x) = \eta_x$ for all $x \in X$. Then, $P$ is a transition kernel with $P*\delta_x \preceq_\C \delta_x$ for all $x \in X$. Moreover, by construction, 
\[\int f \d P*\mu=\int \left(\int f(y) P(\d y \mid x)\right) \mu(\d x) = \int V^\star_f(\delta_x) \mu(\d x) = V^\star_f(\mu) = \max_{\nu: \nu \preceq_{\C} \mu} \int f \d\nu\,,\]
and hence $P * \mu$, which satisfies $P * \mu \preceq_\C \mu$, must also be an optimal solution. But since $\nu$ is the unique solution, it follows that $\nu = P * \mu$, proving the claim. 

Next, we show that the same property holds for any $\nu \in \mathrm{co}(\Exp(\mathcal{K}_\mu))$. Indeed, for any $\nu \in \mathrm{co}(\Exp(\mathcal{K}_\mu))$, $\nu=\sum_{i=1}^n \lambda_i \nu_i$ for some $\{\nu_i\}_{i=1}^n \subseteq \Exp(\mathcal{K}_\mu)$, and some $\{\lambda_i\} \subseteq \R_+$ with $\sum_{i=1}^n \lambda_i=1$. For each $i \in \{1,\ldots,n\}$, as shown above, since $\nu_i \in \Exp(\mathcal{K}_\mu)$, there exists a transition kernel $P_i$ such that $\nu_i=P_i*\mu$ and $P_i*\delta_x \preceq_\C \delta_x$ for all $x \in X$. Let 
\[
P:=\sum_{i=1}^n \lambda_iP_i\,.
\]
Then, 
\[
\nu=\sum_{i=1}^n \lambda_i \nu_i=\sum_{i=1}^n \lambda_i P_i*\mu=P*\mu\,.
\]
Moreover, for any $g \in \C$, and for any $x \in X$
\[
\int g(y)P(\d y\mid x )=\sum_{i=1}^n \lambda_i\int g(y)P_i(\d y\mid x) \leq \sum_{i=1}^n \lambda_i g(x)=g(x)\,.
\]
Thus, $P*\delta_x \preceq_\C \delta_x$ for all $x \in X$. 

Now, consider any $\nu \in \mathcal{K}_\mu$. Since $\mathcal{K}_\mu$ is a compact convex subset of a compact metric space, by a theorem of \citet{Kleejr1958} (in particular, see Theorem 2.51 of \citealt{lukevs2009integral}), there exists a sequence $\{\nu_n\} \subseteq \mathrm{co}(\Exp(\mathcal{K}_\mu))$ such that $\{\nu_n\} \to \nu$. For each $n \in \N$, as shown above, there exists a transition kernel $P_n$ such that $\nu_n=P_n*\mu$ and $P_n*\delta_x \preceq_\C \delta_x$ for all $x \in X$. For each $n \in \N$, define $\gamma_n \in \Delta(X \times X)$ as 
\[
\gamma_n(B\times A):=\int_B P_n(A\mid x)\mu(\d x)\,,
\]
for all measurable $A,B \subseteq X$. Since $X$ is a compact Polish space, $\Delta(X\times X)$ is also a compact Polish space. Therefore, by possibly taking a subsequence, $\{\gamma_n\}$ converges to some $\gamma \in \Delta(X \times X)$. By the disintegration theorem \citep[see, e.g.,][pp. 154, Theorem 2.18]{cinlarprobability2010}, there exists a transition kernel $P$ such that  
\[
\gamma(B\times A)=\int_B P(A\mid x)\mu(\d x)\,,
\]
for all measurable $A,B \subseteq X$. As a result, for any continuous function $h:X \to \R$. 
\begin{align*}
\int_X h \d \nu=\lim_{n \to \infty} \int_X h \d \nu_{n}=&\lim_{n \to \infty }\int_X \left(\int_X h(y) P_{n}(\d y\mid x)\right) \mu(\d x)\\
=&\lim_{n\to \infty}\int_{X \times X} h(y) \gamma_{n}(\d x, \d y)\\
=&\int_{X \times X} h(y) \gamma(\d x,\d y)\\
=&\int_X  \left(\int_X h(y) P(\d y\mid x)\right)\mu(\d x)\\
=&\int_X h \d (P* \mu)\,.
\end{align*}
Therefore, $\nu=P*\mu$. From this $P$, we now construct a transition kernel $\widetilde{P}$ such that $\widetilde{P}*\delta_x \preceq_\C \delta_x$ for all $x \in X$ and that $\nu=\widetilde{P}*\mu$. To begin with, we first show that for any continuous $g \in \C$, there exists a measurable set $X_g \subseteq X$ with $\mu(X_g)=1$, such that 
\[ \Phi^g(x) := \int g(y)P(\d y\mid x)  \leq g(x)\]
for all $x \in X_g$. 
Indeed, fix any $g \in \C$. For each $n \in \N$, let 
\[
\Phi^g_n(x):=\int g(y) P_n(\d y\mid x)\,,
\]
for all $x \in X$. By construction, since $P_n*\delta_x \preceq_\C \delta_x$ for all $x$ and since $g \in \C$, we have $\Phi^g_n(x) \leq g(x)$ for all $x$. Therefore, for any nonnegative continuous function $\psi:X \to \R$,
\begin{align*}
\int_{X} \psi(x)\Phi^g(x) \mu(\d x)=&\int_X \psi(x) \left(\int_X g(y) P(\d y\mid x)\right)\mu(\d x)\\
=&\int_{X\times X} \psi(x)g(y)\gamma(\d x,\d y)\\
=&\lim_{n\to \infty} \int_{X \times X} \psi(x) g(y) \gamma_{n}(\d x,\d y)\\
=& \lim_{n \to \infty} \int_{X} \psi(x) \left(\int g(y)P_{n}(\d y\mid x)\right)\mu(\d x)\\
=&\lim_{n \to \infty} \int_X \psi(x) \Phi^g_n(x) \mu(\d x)\\
\leq& \int_X \psi(x) g(x) \mu(\d x)\,,
\end{align*}
where the third equality follows from $\psi\cdot g$ being continuous on $X \times X$ and $\{\gamma_n\} \to \gamma$. Thus, $\psi \mapsto \int \psi [g - \Phi^g] \d \mu$ is a positive linear functional and hence by the Riesz-Markov-Kakutani representation theorem, $[g(x) - \Phi^g(x)] \mu(\d x)$ is identified by a unique positive Radon measure, which immediately implies that $g(x) - \Phi^g(x) \geq 0$ for $\mu$-a.e. $x$. Therefore, there exists a measurable set $X_g \subseteq X$ with $\mu(X_g)=1$ such that $\Phi^g(x) \leq g(x)$ for all $x \in X_g$. 

Next, we show that there exists a set $X_\infty$ where this property holds uniformly across all $g \in \C$. By \Cref{lem:conti-approx}, there exists a countable set $\{g_m\}_{m=1}^\infty \subseteq \C \cap C(X)$ that can approximate any $g \in \C$ from above. As shown above, for each $m \in \N$, since $g_m \in \C$ is continuous, there exists a measurable set $X_m$ such that $\mu(X_m)=1$ and $\Phi^{g_m}(x) \leq g_m(x)$ for all $x \in X_m$. Let $X_\infty:=\cap_{m=1}^\infty X_m$. Clearly, $\mu(X_\infty)=1$ as well. Now, fix any $g \in \C$. Let $\{g_{m_k}\} \subseteq \{g_m\}$ be a sequence such that $g_{m_k} \geq g$ for all $k$ and $g_{m_k} \rightarrow g$ as $k \rightarrow \infty$. Then, for all $x \in X_\infty$, $\Phi^{g_{m_k}}(x) \leq g_{m_k}(x)$ for all $k \in \N$. Therefore, for all $x \in X_\infty$ and all $k \in \N$, we have 
\[
\int g(y) P(\d y\mid x) \leq  \int g_{m_k}(y) P(\d y\mid x) =\Phi^{g_{m_k}}(x) \leq  g_{m_k}(x)\,,
\]
which, by taking limit over $k$, implies that 
\[\Phi^g(x) =\int g(y) P(\d y\mid x) \leq \lim_{k \rightarrow \infty} g_{m_k}(x) = g(x)\,. \]
Since $g$ is an arbitrary element of $\C$, we have for all $g \in \C$, $\Phi^g \leq g$ on $X_\infty$ where $\mu(X_\infty)=1$. 

Now, let 
\[
\widetilde{P}(\,\cdot\,\mid x):=\begin{cases}
P(\,\cdot\,\mid x),&\mbox{if } x \in X_\infty\,;\\
\delta_x,&\mbox{otherwise.}
\end{cases}
\]
By construction, since $P$ is a measurable Markov kernel, $\widetilde{P}$ is also a measurable Markov kernel. Moreover, for all $x \in X$ and all $g \in \C$,
\[
\int g(y) \widetilde{P}(\d y\mid x) =\begin{cases}
\int g(y) P(\d y\mid x) \leq g(x),&\mbox{if } x \in X_\infty\\
g(x),&\mbox{otherwise.}
\end{cases}
\]
Therefore, $\widetilde{P}*\delta_x \preceq_\C \delta_x$ for all $x \in X$. Meanwhile, since $P(\,\cdot\,\mid x)=\widetilde{P}(\,\cdot\,\mid x)$ for $\mu$-almost all $x \in X$, and since $\nu=P*\mu$, we have $\nu=\widetilde{P}*\mu$ as well. Together, we have that for any  $\nu \in \mathcal{K}_\mu$, there exists a transition kernel $\widetilde{P}$ such that $\nu=\widetilde{P}*\mu$ and $\widetilde{P}*\delta_x \preceq_\C \delta_x$ for all $x \in X$, as desired. \\
\mbox{}

\noindent $(c) \implies (d)$: Fix any closed and convex subset $M$ of $\Delta(X)$. Fix any upper semicontinuous $f:X \to \R$. Let 
\[
h(x):=\max_{P: P*\delta_x \preceq_\C \delta_x} \int f(y) P(\d y\mid x)\,,
\]
for any $x \in X$. From $(c)$, note that for any $\mu \in \Delta(X)$, we have 
\[V^\star_f(\mu) = \max_{\nu \preceq_{C} \mu} \int f \d\nu = \int \Big(\max_{P: P* \delta_x \preceq_\C \delta_x} \int f(y) P(\d y \mid x)\Big) \mu(\d x) =\int h \d\mu\,,\]
where the second equality also uses the Kuratowski–Ryll-Nardzewski measurable selection theorem as before. Therefore, $V^\star_f$ is affine and satisfies the barycentric formula. This immediately implies that $G^M_f$ has a convex graph. Indeed, for any $(\mu_1, \nu_1), (\mu_2, \nu_2) \in G^M_f$ and any $\alpha \in (0, 1)$, we have 
\[\nu := \alpha \nu_1 + (1 - \alpha) \nu_2 \preceq_{\C} \alpha \mu_1 + (1 - \alpha) \mu_2 =: \mu\]
and thus 
\[\int f(x) \nu (\d x) = \alpha V^\star_f(\mu_1) + (1-\alpha) V^\star_f(\mu_2) = V^\star_f(\mu)\,,\]
and therefore $\nu \in X^\star_f(\mu)$. To see that $G^M_f$ has decomposable extreme points, observe that the above also shows that $\nu \in X^\star_f(\mu)$ if and only if $\nu = P * \mu$ for some $P$ such that $P(\,\cdot\,\mid x)$ solves 
\[
\max_{\eta:\eta \preceq_\C \delta_x} \int f \d \eta\,,
\]
for $\mu$-almost all $x \in X$. It is easy to see that if $\mu \in \Ext(M)$ and $\nu \in \Ext(X^\star_f(\mu))$, then $(\mu, \nu) \in \Ext(G^M_f)$. Indeed, if $\mu \in \Ext(M)$, $\nu \in \Ext(X^\star_f(\mu))$, and $(\mu, \nu)=\frac{1}{2}(\mu_1, \nu_1) + \frac{1}{2}(\mu_2, \nu_2)$ where $(\mu_1, \nu_1), (\mu_2, \nu_2) \in G^M_f$, then we must have $\mu_1 = \mu_2 = \mu$ since $\mu \in \Ext(M)$, which implies that $\nu_1, \nu_2 \in X^\star_f(\mu)$ and hence $\nu_1 = \nu_2 = \nu$ since $\nu \in \Ext(X^\star_f(\mu))$. 

Now, we prove the other direction. Fix any $(\mu, \nu) \in \Ext(G^M_f)$. Clearly, we have that $\nu \in \Ext(X^\star_f(\mu))$. Now, suppose for contradiction that $\mu \not \in \Ext(M)$. Then $\mu = \frac{1}{2} \mu_1 + \frac{1}{2} \mu_2$ for some $\mu_1 \neq \mu_2 \in M$. Note that $\mu_1$ and $\mu_2$ are both absolutely continuous with respect to $\mu$. By the above observation, there exists some $P$ such that $P * \mu = \nu$ and $P(\,\cdot\,\mid x)$ solves the above pointwise maximization problem $\mu$-a.e. Let $\nu_1 = P * \mu_1$ and $\nu_2 = P * \mu_2$. Since $\mu_i$ are absolutely continuous with respect to $\mu$, it follows immediately that $\nu_i \preceq_\C \mu_i$, and 
\[\int f(x) \nu_i (\d x) = \int f(y) P(\d y \mid x) \mu_i (\d x) = \int h(x) \mu_i (\d x) = V^\star_f(\mu_i)\,, \]
where the second equality uses the $\mu$-a.e., and hence $\mu_i$-a.e., pointwise optimality of $P$, and the third equality uses the barycentric formula for $V^\star_f$. Therefore, $\nu_i \in X^\star_f(\mu_i)$. It follows that $(\mu_i, \nu_i) \in G^M_f$ for $i \in \{1, 2\}$. However, note that, by construction, $\frac{1}{2}(\mu_1, \nu_1) + \frac{1}{2}(\mu_2, \nu_2) = (\mu, \nu)$ and $\mu_1 \neq \mu_2$---contradicting $(\mu, \nu)$ is an extreme point of $G^M_f$. \\
\mbox{}

\noindent $(d) \implies (a)$: Note that since $G^M_f$ is convex for $M = \Delta(X)$, we must have that $V^\star_f$ is affine. Indeed, for any $\mu_1, \mu_2 \in \Delta(X)$. Fix any $\nu_1 \in X^\star_f(\mu_1), \nu_2 \in X^\star_f(\mu_2)$. By convexity of the graph, we have that for any $\alpha \in (0, 1)$, $\alpha \nu_1 + (1 -\alpha) \nu_2 \in X^\star_f\big(\alpha \mu_1 + (1-\alpha) \mu_2\big)$ and thus 
\[V^\star_f(\alpha \mu_1 + (1-\alpha) \mu_2) = \alpha \int f(x) \nu_1 (\d x) + (1 - \alpha) \int f(x) \nu_2(\d x) = \alpha V^\star_f(\mu_1) + (1 - \alpha) V^\star_f(\mu_2)\,.\]
Therefore, $V^\star_f$ is an upper semicontinuous affine function on $\Delta(X)$. Now, suppose for contradiction that $\C$ is not min-closed. Then there exist $g_1, g_2 \in \C$ and 
\[k := \min\{g_1, g_2\} \notin \C\,.\]
By assumption, there exists a decreasing sequence $\{g^n_i\} \subseteq \C \cap C(X)$ such that $g^n_i \downarrow g_i$ for each $i \in \{1, 2\}$. Let $k^n := \min\{g^n_1, g^n_2\}$. Note that each $k^n$ is continuous and the decreasing sequence $\{k^n\}\downarrow k$ by construction. Therefore, there must exist some $n$ such that $k^n \not\in \C$, because otherwise by assumption we must have $k \in \C$. Fix such $n$. Then, $k^n \in C(X)$ but $k^n \not\in \C \cap C(X)$.

Let $\tilde{\C} := \C \cap C(X)$. Since $\tilde{\C}$ is a closed convex set in the space of $C(X)$ under the sup-norm, which is a locally convex topological vector space, by the strict Hahn–Banach separation theorem (see, e.g., \citealt{lukevs2009integral}, Theorem A.1), there exists a continuous linear functional that separates $\{k^n\}$ and $\tilde{\C}$. In particular, there exists a signed finite Borel measure $\sigma$ and some $\alpha \in \R$ such that 
\[\int k^n(x) \sigma(\d x) > \alpha \geq \sup_{g \in\tilde{\C}} \int g(x) \sigma(\d x)\,. \]
Since $\tilde{\C}$ is a cone, it must be that 
\[\sup_{g \in \tilde{\C}} \int g(x) \sigma(\d x) \leq 0\,.\]
Since $\tilde{\C}$ contains $-1$ and $1$, we must also have $\sigma(X) = 0$, and $\int k^n(x) \sigma( \d x) > \alpha \geq 0$. By the Jordan decomposition theorem, there exist Borel measures $\sigma^+$ and $\sigma^-$ on $X$ such that $\sigma=\sigma^+-\sigma^-$. Since $\sigma(X)=0$ and $\int k^n(x) \sigma(\d x) > 0$, $\sigma^+(X)=\sigma^-(X) > 0$.  Let $m:=\sigma^+(X)$, and let $\nu:=\sigma^+/m$, $\mu:=\sigma^-/m$. Then, $\mu,\nu \in \Delta(X)$, and for any $g \in \tilde{\C}$, we have 
\begin{align*}
\int_X g \d \nu -\int_X g \d \mu=\frac{1}{m}\left(\int_X g \d \sigma^+-\int_X g \d \sigma^-\right) \leq 0\,,
\end{align*}
and hence for any $g \in \C$, we have 
\[\int_X g \d \nu \leq \int_X g \d \mu\]
by the assumption that any $g \in \C$ admits a sequence of continuous approximations from above and the monotone convergence theorem. Therefore, $\nu \preceq_\C \mu$. Moreover, by construction, we also have 
\[\int k^n(x) \nu(\d x) > \int k^n(x) \mu(\d x) \,.\]
Now, note that for any $x \in X$ and for any $\eta \preceq_\C \delta_x$, since $g^n_1,g^n_2 \in \C$, we have 
\[
\int k^n \d \eta=\int \min\{g^n_1,g^n_2\}\d \eta \leq \min \left\{\int g^n_1 \d \eta ,\int g^n_2 \d \eta\right\} \leq \min\{g^n_1(x),g^n_2(x)\}=k^n(x)\,.
\]
It then follows that, for all $x \in X$,
\[
V_{k^n}^\star(\delta_x)=k^n(x)\,.
\]
Together, since $\nu \preceq_\C \mu$ and since $V^\star_{k^n}$ is affine, we have
\[\int k^n \d\nu \leq V^\star_{k^n}(\mu)= \int V^\star_{k^n}(\delta_x) \mu (\d x) = \int k^n \d \mu < \int k^n \d \nu\,,\]
where the first equality uses that $V^\star_{k^n}$ is an upper semicontinuous, and hence Baire-one, affine function and thus satisfies the barycentric formula (see Corollary 4.22 of \citealt{lukevs2009integral})---a contradiction.  \hfill \qedsymbol

\subsection{Proof of \texorpdfstring{\Cref{prop:exposed}}{Proposition 1}}
$(i)$: Consider any $\nu \in \mathcal{K}_\mu$ such that $\nu$ is uniquely rationalized by some order-preserving $P$ with $P(\,\cdot\,\mid x)$ being an extreme point of $\mathcal{K}_{\delta_x}$ $\mu$-a.e.. Suppose for contradiction that there exist $\nu_1,\nu_2 \in \mathcal{K}_\mu$ and $\lambda \in (0,1)$, with $\nu_1 \neq \nu_2$, such that $\nu=\lambda\nu_1+(1-\lambda)\nu_2$. Since $\nu_i \in \mathcal{K}_\mu$, by \Cref{thm:main}, there exists an order-preserving kernel $P_i$ such that $P_i*\delta_x \preceq_\C \delta_x$ for all $x$ and $i \in \{1,2\}$. Thus, 
\[
P*\mu=\nu=\lambda \nu_1+(1-\lambda) \nu_2=(\lambda P_1+(1-\lambda)P_2)*\mu\,,
\]
which is also order-preserving. Since $\nu$ is uniquely rationalized, we have that 
\[
P(\,\cdot\,\mid x)=\lambda P_1(\,\cdot\,\mid x)+(1-\lambda)P_2(\,\cdot\,\mid x)\,,
\]
for $\mu$-almost all $x \in X$. However, since $\nu_1\neq \nu_2$, there exists a subset $A \subseteq X$ with $\mu(A)>0$ such that $P_1(\,\cdot\,\mid x)\neq P_2(\,\cdot\,\mid x)$ for all $x \in A$. Then, for this $\mu$-positive measure of $x \in A$, we have that $P(\,\cdot\,\mid x)$ is not an extreme point of $\mathcal{K}_{\delta_x}$, a contradiction. 

For the converse, consider any extreme point $\nu$ of $\mathcal{K}_\mu$. By \Cref{thm:main}, $\nu=P*\mu$ for some order-preserving $P$ such that $P*\delta_x \preceq_\C \delta_x$ for all $x \in X$. We first show that $\nu$ must be uniquely rationalized. Indeed, suppose for contradiction that there exist order-preserving kernels $P_1, P_2$ where $P_1, P_2$ differ on a $\mu$-positive measure set such that $P_i*\delta_x \preceq_\C \delta_x$ for $\mu$-a.e. $x$ and  
\[
\nu=P_i*\mu,
\]
for all $i \in \{1, 2\}$. Since $X$ is Polish, $\mathcal{B}(X)$ is generated by a countable algebra $\mathcal{A}$. Therefore, for any $x$, if $P_1(\,\cdot\,\mid x) \neq P_2(\,\cdot\, \mid x)$, then there exists $B \in \mathcal{A}$ such that $P_1(B\mid x) \neq P_2(B \mid x)$. By countability of $\mathcal{A}$, this implies that there exists a single measurable set $B \subseteq X$ such that $\mu\big(\{x: P_1(B\mid x) \neq P_2(B\mid x)\}\big) > 0$. Let 
\[
A:=\big\{x \in X: P_1(B\mid x)>P_2(B\mid x)\big\}\,.
\]
By possibly relabeling, it is without loss of generality to assume that $\mu(A)>0$. Then, let 
\[
Q_1(\,\cdot\,\mid x):=
\begin{cases}
P_1(\,\cdot\,\mid x),&\mbox{if } x \in A\\
P_2(\,\cdot\,\mid x),&\mbox{otherwise } 
\end{cases}\,, \quad \quad \mbox{ and } \quad \quad 
Q_2(\,\cdot\,\mid x):=\begin{cases}
P_2(\,\cdot\,\mid x),&\mbox{if } x \in A\\
P_1(\,\cdot\,\mid x),&\mbox{otherwise }
\end{cases}\,.
\]
Clearly, $Q_1,Q_2$ are transition kernels since $A$ is a measurable set. Moreover, since $P_i*\delta_x \preceq_\C \delta_x$ for $\mu$-a.e. $x$ and for $i \in \{1,2\}$, we have that $Q_i *\delta_x \preceq_\C \delta_x$ for $\mu$-a.e. $x \in X$ and $i \in \{1,2\}$. 

Let $\nu_i:=Q_i*\mu$ for all $i \in \{1,2\}$. It then follows that $\nu_1,\nu_2 \in \mathcal{K}_\mu$. Furthermore, by construction, we have  
\[
\frac{1}{2}Q_1+\frac{1}{2}Q_2=\frac{1}{2}P_1+\frac{1}{2}P_2\,,
\]
and thus, 
\[
\nu=\left(\frac{1}{2}P_1+\frac{1}{2}P_2\right)*\mu=\frac{1}{2}Q_1*\mu+\frac{1}{2}Q_2*\mu=\frac{1}{2}\nu_1+\frac{1}{2}\nu_2\,.
\]
In the meantime, by construction, we have 
\begin{align*}
\nu_1(B)-\nu_2(B)=&Q_1*\mu(B)-Q_2*\mu(B)\\
=&\int_A \left(P_1(B\mid x)-P_2(B\mid x)\right)\mu(\d x)+\int_{A^c} \left(P_2(B\mid x)-P_1(B\mid x)\right) \mu(\d x)\\
\geq &\int_A \left(P_1(B\mid x)-P_2(B\mid x)\right)\mu(\d x)\\
>&0\,,
\end{align*}
where the first inequality follows from $P_1(B\mid x) \leq P_2(B\mid x)$ for all $x \in A^c$  by construction; and the second inequality follows from $P_1(B\mid x)>P_2(B\mid x)$ for all $x \in A$ and $\mu(A)>0$. Therefore, $\nu_1\neq \nu_2$, and hence $\nu$ is not an extreme point, a contradiction. 

Next, we show that for the (essentially) unique $P$ such that $\nu=P*\mu$, it must be that $P(\,\cdot\,\mid x) \in \Ext(\mathcal{K}_{\delta_x})$ for $\mu$-almost all $x \in X$. To prove this, we identify any $\mu$-measurable Markov kernel $P$ as a joint distribution $\gamma \in \Delta(X \times X)$ whose first marginal is $\mu$. We claim that the set of order-preserving $\mu$-measurable Markov kernels can be identified by 
\[\Gamma_\mu := \Bigg\{\gamma: (\pi_1)_*\gamma = \mu\,,\int h(x)g(y) \gamma(\d x, \d y) \leq \int h(x)g(x) \mu(\d x) \text{ for all } h \in C_+(X), g \in \mathcal{C} \cap C(X)\Bigg\}\,, \]
where $C_+(X)$ is the set of nonnegative continuous functions on $X$ and $(\pi_1)_*\gamma$ denotes the first marginal of $\gamma$.  By inspection, any order-preserving $P$ can be identified by an element in $\Gamma_\mu$. To see the converse, fix any $\gamma \in \Gamma_\mu$. Note that by the monotone convergence theorem and the assumption that any element $g \in \C$ can be approximated pointwise from above by continuous functions in $\C$, we have 
\[\int h(x)g(y) \gamma(\d x, \d y) \leq  \int h(x)g(x) \mu(\d x) \text{ for all } h \in C_+(X), g \in \mathcal{C}\,.\]
By the disintegration theorem, there exists a Markov kernel $Q$ such that $\gamma(\d x, \d y) = Q(\d y \mid x) \mu(\d x)$, and hence 
\[\int h(x)  \Big[\int g(y) Q(\d y \mid x)\Big]  \mu(\d x) \leq  \int h(x)g(x) \mu(\d x) \text{ for all } h \in C_+(X), g \in \mathcal{C}\,.\]
Thus, for any $g \in \C$, we have 
\[\int h(x)  \Big[g(x) - \int g(y) Q(\d y \mid x)\Big]  \mu(\d x) \geq  0 \text{ for all } h \in C_+(X)\,,\]
which implies that 
\[q_g(x) :=  g(x) - \int g(y) Q(\d y \mid x) \geq 0 \text{ for $\mu$-a.e. $x$}\,.\]
Indeed, to see the above, note that $h \mapsto \int h q_g \d \mu$ is a positive linear functional and hence by the Riesz-Markov-Kakutani representation theorem, $q_g(x) \mu(\d x)$ is identified by a unique positive Radon measure, which immediately implies that $q_g(x) \geq 0$ for $\mu$-a.e. $x$. Let $X_g := \{x: q_g(x) \geq 0\}$. Then, $\mu(X_g) = 1$ for all $g \in \C$. Let $\{g_m\}_{m\in\N} \subseteq \C \cap C(X)$ be the countable collection given by \Cref{lem:conti-approx}. Let 
\[X_\infty := \bigcap_{m=1}^\infty X_{g_m}\,.\]
Note that $X_\infty$ is measurable and satisfies $\mu(X_\infty) = 1$. Now, fix any $x \in X_\infty$ and any $g \in \C$. By \Cref{lem:conti-approx}, there exists a sequence $\{g_{m_k}\} \subseteq \{g_m\}_{m\in\N}$ such that $g_{m_k} \geq g$ for all $k$, and $g_{m_k}(x) \rightarrow g(x)$ as $k \rightarrow \infty$. Then, by construction, for any $k \in \N$, we have 
\[\int g(y) Q(\d y \mid x) \leq \int g_{m_k}(y) Q(\d y \mid x) \leq  g_{m_k}(x)\,,\]
where the last inequality uses that $x \in X_\infty$. Taking $k \rightarrow \infty$ then gives that 
\[\int g(y) Q(\d y \mid x) \leq g(x)\,. \]
Since this holds for all $g \in \C$, it follows immediately that 
\[Q * \delta_x \preceq_\C \delta_x \]
for all $x \in X_\infty$. Since $\mu(X_\infty) = 1$, modifying $Q$ on a $\mu$-null set if necessary, we conclude that $Q$ is an order-preserving kernel. 

We equip $\Delta(X \times X)$ with the weak-* topology. Note that $\Delta(X \times X)$ is compact. Moreover, $\Gamma_\mu$ is a closed subset of $\Delta(X \times X)$ and hence compact. Now consider the map $T: \Delta(X \times X) \rightarrow \Delta(X)$ defined by: for all $\gamma \in \Delta(X \times X)$,
\[T(\gamma) := (\pi_2)_* \gamma\,,\]
i.e., the projection to the second marginal of $\gamma$. Clearly, $T$ is a continuous linear map given the weak-* topology. Moreover, note that $T(\Gamma_\mu) = \mathcal{K}_\mu$ by construction and our previous argument. Since $T$ is continuous and $\Gamma_\mu$ is compact, $\mathcal{K}_\mu$ is compact. By Proposition 8.26 of \citet{simon2011convexity}, we conclude that 
\[\Ext(\mathcal{K}_\mu) \subseteq T(\Ext(\Gamma_\mu))\,.\]
For any $\nu \in \Ext(\mathcal{K}_\mu)$, by our previous argument, there exists a unique $\gamma \in \Gamma_\mu$ such that $(\pi_2)_* \gamma = \nu$. Thus, the unique $\gamma \in \Ext(\Gamma_\mu)$. 

It remains to argue that for the unique $\gamma \in \Ext(\Gamma_\mu)$,  its $\mu$-a.e. unique disintegration satisfies $Q(\,\cdot\,\mid x) \in \Ext(\K_{\delta_x})$ for $\mu$-a.e. $x$. Note that, by the disintegration theorem, without loss of generality, we can take $Q$ as a Borel-measurable kernel. Now, suppose for contradiction that the set 
\[S:= \Big\{x\in X: Q(\,\cdot\,\mid x) \not \in \Ext(\mathcal{K}_{\delta_x})\Big\}\]
satisfies $\mu(S)> 0$. Define 
\[\Sigma:= \Big\{(x, \eta_1, \eta_2)\in X \times \Delta(X) \times \Delta(X): \eta_1 \in \K_{\delta_x},\, \eta_2 \in \K_{\delta_x},\, Q(\,\cdot\,\mid x) = \frac{1}{2}\eta_1 + \frac{1}{2}\eta_2,\, \eta_1 \neq \eta_2 \Big\}\,.\]
Since the set 
\[\big\{(x, \eta) \in X \times \Delta(X): \eta \in \K_{\delta_x}\big\} = \bigcap_{g \in \C \cap C(X)} \Big\{(x, \eta): \int g \d \eta \leq g(x)\Big\}\]
is an intersection of closed sets and hence closed; the set 
\[ \Big\{(x, \eta_1, \eta_2): Q(\,\cdot\,\mid x) = \frac{1}{2}\eta_1 + \frac{1}{2}\eta_2\Big\}\]
is Borel since the Markov kernel $Q$ is Borel-measurable; and the set $\{(\eta_1, \eta_2): \eta_1 \neq \eta_2\}$ is open. Therefore, $\Sigma$ is a Borel set. Moreover, Note that 
\[S = \pi_{X}(\Sigma)\,,\]
i.e., the set $S$ is exactly the projection of $\Sigma$, and thus $S$ is an analytic set. By the Jankov-von Neumann uniformization theorem (\citealt{kechris2012classical}, Theorem 18.1), it follows that there exists a universally measurable selection $x \mapsto (\eta_1(x), \eta_2(x))$ mapping from $S$ to $\Delta(X) \times \Delta(X)$. Now, for $i = 1, 2$, let  
\[Q_i(\,\cdot\, \mid x) := \begin{cases}
    \eta_i(x)  &\text{ if } x \in S\,, \\
    Q(\,\cdot\, \mid x)  &\text{ otherwise } \,.\\
\end{cases}\]
By construction, $x \mapsto Q_i(\,\cdot\,\mid x)$ is $\mu$-measurable, and hence $Q_i$'s are $\mu$-measurable Markov kernels. Moreover, $Q_i(\cdot\mid x) \in \mathcal{K}_{\delta_x}$ for $\mu$-a.e. $x$ by construction. Thus, each $Q_i$ is identified by some $\gamma_i \in \Gamma_\mu$. By construction, note that $\frac{1}{2}\gamma_1 + \frac{1}{2}\gamma_2 = \gamma$. Finally, to verify $\gamma_1 \neq \gamma_2$, fix a countable collection of determining functions $\{\psi_n\}$ where $\psi_n \in C(X)$. For each $n$, define 
\[S^+_n:= \Big\{ x \in S: \int \psi_n \d \eta_1(x) > \int \psi_n \d \eta_2(x) \Big\} \,\text{ and }\,
S^-_n:= \Big\{ x \in S: \int \psi_n \d \eta_1(x) < \int \psi_n \d \eta_2(x) \Big\}\,.\]
Since $\{\psi_n\}$ is a determining class, every $x \in S$ belongs to $S^+_n \cup S^-_n$ for some $n$. Hence, $S \subseteq \bigcup_n S^+_n \cup S^-_n$. Since $\mu(S) > 0$, there exists some $n$ such that $\mu(S^+_n) > 0$ or $\mu(S^-_n) > 0$. Without loss of generality, say $\mu(S^+_n) > 0$. Then 
\begin{align*}
&\int \int \ind_{S^+_n}(x) \psi_n(y) \gamma_1(\d x, \d y) - \int \int \ind_{S^+_n}(x) \psi_n(y) \gamma_2(\d x, \d y)\\
=& \int_{S^+_n} \Big[\int \psi_n \d \eta_1(x) - \int \psi_n \d \eta_2(x)\Big] \mu(\d x)\\
>& 0\,,
\end{align*}
and thus $\gamma_1 \neq \gamma_2$. Therefore, $\gamma \not\in \Ext(\Gamma_\mu)$---a contradiction.  \\
\mbox{}

\noindent $(ii)$: Suppose that there exist some continuous $f$ and some order-preserving kernel $P$ such that $\nu = P * \mu$, and $\mu$-a.e. we have $\int f(y) \eta(\d y)= \overline{f}(x) \text{ and } \eta \preceq_\C \delta_x \iff \eta = P_x$. Let $A$ be the set under which the second claim holds. So $\mu(A) = 1$. We claim that $\nu = P* \mu$ must be an exposed point. Fix the objective given by $f$. Suppose for contradiction that $\nu' \neq \nu$ is also optimal under $f$. Then by \Cref{thm:main}, there exists an order-preserving kernel $P'$ such that $\mu$-a.e. $P'_x$ is optimal for the pointwise optimization problem, i.e., there exists a set $B$ with $\mu(B) = 1$ such that for all $x \in B$, we have $P'_x$ is a solution to: 
\[\max_{\eta \preceq_\C \delta_x} \int f(y) \eta(\d y)\,.\]
Clearly, $\mu(A \cap B)=1$ and for all $x \in A\cap B$, we have that 
\[\max_{\eta \preceq_\C \delta_x} \int f(y) \eta(\d y)= \overline{f}(x)\,,\]
and $P'_x = P_x$. But then immediately we have $\nu' = P' * \mu = P * \mu = \nu$, a contradiction.

For the converse, suppose $\nu$ is an exposed point. Let $f$ be any continuous function such that $\int f(x) \nu(\d x)$ exposes $\nu$. Let 
\[S(x):= \argmax_{\eta \preceq_\C \delta_x }\int f(y) \eta(\d y)\,.\]
Note that by \Cref{thm:main}, $\nu$ is optimal for $f$ if and only if 
\[\nu \in \mathcal{S} := \Big\{\nu = P * \mu \text{ and $\mu\Big(P_x \in S(x)\Big) = 1$}   \Big\}\,.\]
Suppose for contradiction that 
\[\mu\Big(|S(x)| > 1\Big) > 0\,.\]
Let $E$ be the above set with $\mu(E) > 0$. Fix a countable set $\{\varphi_n\}_n \subset C(X)$ that separates measures on $\Delta(X)$. Let 
\[R_n(x):= \Big\{\int \varphi_n(y) \eta(\d y): \eta \in S(x)\Big\} \subset \mathbb{R}\,.\]
Since $S(x)$ is compact and $\eta \mapsto \int \varphi_n(y) \eta(\d y)$ is continuous, $R_n(x)$ is compact. Moreover, $R_n(x)$ is also convex, since $S(x)$ is convex. Thus, $R_n(x)$ is a compact interval. 

For any $S(x)$ with $|S(x)| > 1$, there exist two distinct $\eta_1, \eta_2 \in S(x)$. Since $\{\varphi_n\}$ separates measures, there exists some $n$ such that 
\[\int \varphi_n(y) \eta_1(\d y) \neq \int \varphi_n(y) \eta_2(\d y)\,,\]
and hence $|R_n(x)| > 1$. Let 
\[E_n:= \{x: |R_n(x)| > 1\}\,.\]
Then, we have 
\[E \subseteq \bigcup_{n \geq 1} E_n\,,\]
and hence for some $n$, $\mu(E_n) > 0$. Fix such $n$, and let $\varphi := \varphi_n$. 

Let 
\[S^+(x) := \argmax_{\eta \in S(x)} \int \varphi(y) \eta(\d y)\qquad \qquad S^-(x) := \argmin_{\eta \in S(x)} \int \varphi(y) \eta(\d y)\,.\]
Note that each is nonempty compact-valued. By the Kuratowski–Ryll-Nardzewski theorem, there exist measurable selections 
\[\eta^+_x \in S^+(x) \qquad \qquad \eta^-_x \in S^-(x)\,.\]
Note that by construction, for any $x \in E_n$,  we have 
\[\int \varphi(y) \eta^+_x(\d y) - \int \varphi(y) \eta^-_x(\d y) = \int \varphi_n(y) \eta^+_x(\d y) - \int \varphi_n(y) \eta^-_x(\d y) > 0\,.\]
Fix any $P$ such that $\nu = P * \mu$ and $P_x \in S(x)$ $\mu$-a.e. Construct $P^+$ and $P^-$ by  \[P^+_x=\begin{cases}
  \eta^+_x &\text{ if $x \in E_n$} \\
  P_x &\text{ otherwise}\,.
\end{cases} \qquad P^-_x=\begin{cases}
  \eta^-_x &\text{ if $x \in E_n$} \\
  P_x &\text{ otherwise}\,.
\end{cases} \]
It follows immediately that $\nu^+ := P^+ * \mu \preceq_\C \mu$ and $\nu^- := P^- * \mu \preceq_\C \mu$. Moreover, by construction, we also have that $P^+_x \in S(x)$ $\mu$-a.e. and similarly $P^-_x \in S(x)$ $\mu$-a.e. Therefore, $\nu^+$ and $\nu^-$ must be optimal for $f$. Now, note that by construction 
\[\int \varphi(x) \nu^+(\d x) - \int \varphi(x) \nu^-(\d x) = \int_{E_n} \Big(\int \varphi(y) \eta^+_x(\d y) - \int \varphi(y) \eta^-_x(\d y) \Big)\mu(\d x)> 0\,,\]
where the strict inequality follows since $\mu(E_n) > 0$ and on $E_n$, the difference is strictly positive pointwise. Therefore, $\nu^+ \neq \nu^-$---a contradiction. \hfill \qedsymbol

\subsection{Proof of \texorpdfstring{\Cref{prop:mps}}{Proposition 2}}

\noindent $(a) \implies (b)$: Suppose $\nu$ is an exposed point and let $f$ be the continuous function that exposes $\nu$. Let $\mathcal{A}$ be the set of differentiability points of $\overline{f}$, which is a Lebesgue full-measure set in $X$. By  \Cref{prop:exposed}, since $\mu$ is mutually absolutely continuous with the Lebesgue measure, for Lebesgue-almost all $x$, we have that the pointwise problem 
\[\max_{\eta \preceq_{\text{concave}} \delta_x} \int f(y) \eta(\d y) = \overline{f}(x) \,\]
admits a unique solution. We claim that $f$ must have a strict contact---indeed, if there exists a positive measure of $x \in S \cap \mathcal{A}$ such that $x \in \text{conv}(T_{h^{x}}\backslash \{x\})$, then for such a positive measure of $x$, which by definition satisfies $f(x) = \overline{f}(x)$, the pointwise maximization problem admits multiple solutions, contradicting the above. 

Now, we claim that for a full-measure set of $x \in S^c \cap \mathcal{A}$, the affine envelope region $S^c_{h^x}$ must be simplicial-touching. We first note that, by \Cref{lem:ri-contact-ae}, there exists a full-measure set $\mathcal{B}$ in $X$ such that $x$ is in the relative interior of $R_{h^x}$ for all $x \in \mathcal{B}$. 

Now, consider the set 
\[\mathcal{Z}:=\Big\{x \in S^c \cap \mathcal{A} \cap \mathcal{B}: \text{ pointwise problem admits a unique solution at $x$}\Big\}\,.\]
By \Cref{prop:exposed} and that $\mathcal{B}$ has full measure, $\mathcal{Z}$ is a full-measure subset of $S^c \cap \mathcal{A}$. We claim that for all $x \in \mathcal{Z}$, the affine envelope region $S^c_{h^x}$ must be simplicial-touching. To see this, fix any such $x$. Let $h:= h^x$ to simplify the notation. 

Note that by construction, $h \geq \overline{f} \geq f$, and hence 
\[T_h \subseteq R_h\,.\]
Note that $R_h$ is a convex set. Thus, $\text{conv}(T_h) \subseteq R_h$. Now, we show that $R_h \subseteq \text{conv}(T_h)$. Fix any $z \in R_h$. By definition, we have $\overline{f}(z) = h(z)$. Moreover, note that 
\[V^\star_f(\delta_z) = \overline{f}(z) = h(z)\,.\]
By compactness, as argued before, the pointwise problem always admits a solution. In particular, there exists a solution $\eta_z$ to the pointwise problem at $z$. Then, $\eta \preceq_{\text{concave}} \delta_z$ and $\int f \d \eta_z  = \overline{f}(z) = h(z) = \int h \d \eta_z$. Hence, $\eta_z$ must be supported on $T_h$ and have barycenter $z$. Thus, $z \in \text{conv}(T_h)$. Therefore, $R_h \subseteq \text{conv}(T_h)$, proving that $\text{conv}(T_h) = R_h$. 

Note that by construction $x$ is in the relative interior of $R_{h} = \text{conv}(T_h)$. This implies that $T_h$ must be affinely independent, because otherwise at the fixed $x \in \text{conv}(T_h)$, the optimization problem would admit multiple pointwise solutions that are supported on $T_h$. Indeed,if $T_h$ is not affinely independent, by Radon’s Theorem, there exists some point $w \in \text{conv}(T_h)$ that has two distinct barycentric representations. Since $x \in \text{ri}(\text{conv}(T_h))$, there exists some $\lambda \in [0, 1)$ and some $y \in \text{conv}(T_h)$ such that 
\[x = \lambda y + (1- \lambda) w\,.\]
Therefore, the point $x$ must have two distinct barycentric representations supported on $T_h$. Note that each of which would be optimal for the pointwise problem, contradicting the uniqueness property of the pointwise problem at $x$. Therefore, $T_h$ is affinely independent, and hence we have 
\[\Ext(R_h) = \Ext(\text{conv}(T_h))= T_h\,,\]
and hence $R_h$ is a simplex. 

By \Cref{prop:exposed}, there exists an order-preserving coupling $P^f$ such that $\nu = P^f * \mu$ and that $P^f * \delta_x$ is exposed by $f$ for $\mu$-almost all, and hence Lebesgue-almost all, $x \in X$. This implies that $P^f = \delta_x$ for a full-measure set of $x \in S$, and $P^f$ is the unique barycentric splitting for Lebesgue-almost all $x \in \mathcal{Z}$, which is a full-measure set in $S^c$, since, as we have shown, the unique pointwise solution for all $x \in \mathcal{Z}$ is given by the barycentric splitting. Thus, if necessary, modifying the kernel $P^f(\,\cdot\, \mid x)$ to be $\delta_x$ on a $\mu$-null set gives the desired construction, concluding the proof.\\

\noindent $(b) \implies (a)$: Fix any such $f$, $P^f$, and $\nu = P^f * \mu$. We verify the conditions in \Cref{prop:exposed}. Consider the set of differentiability points $\mathcal{A}$ of $\overline{f}$. We claim that for Lebesgue-almost all $x \in \mathcal{A}$, which is a Lebesgue full-measure set in $X$ by concavity of $\overline{f}$, we have 
\[\int f(y) \eta(\d y) = \overline{f}(x) \text{ and } \eta \preceq_{\text{concave}} \delta_x \iff \eta = P^f_x \,.\]
Indeed, consider the full-measure set of $x \in S \cap \mathcal{A}$ for which $x \not \in \text{conv}(T_{h^x} \backslash \{x\})$. Fix any such $x$. Note that $\overline{f}(x) = f(x)$ and hence $\eta = P^f_x = \delta_x$ clearly satisfies the left side; moreover, if there exists some $\eta$ satisfying the left side, then $\eta$ must be supported on $T_{h^x}$ and have barycenter $x$, but then it must be that $\eta = \delta_x$. 

Now, consider the full-measure set of $x \in S^c \cap \mathcal{A}$ such that $S^c_{h^x}$ is simplicial-touching. Fix any such $x$. Let $h := h^x$. Since $S^c_h$ is simplicial-touching, if $\eta = P^f_x$, then by construction we have the left side; conversely, if the left side is satisfied by $\eta$, then $\eta$ must be supported on $T_{h}$, but since $T_h = \Ext(R_h)$, where $R_h$ is a simplex, is affinely independent, there exists a unique mean-preserving splitting of $x$ supported on $T_h$ given by $P^f_x$, and hence $\eta = P^f_x$.

Since the equivalence relation holds almost everywhere on $\mathcal{A}$ and $\mathcal{A}$ is a full-measure set in $X$, by \Cref{prop:exposed}, $\nu$ must be an exposed point. \hfill \qedsymbol

\subsection{Proof of \texorpdfstring{\Cref{prop:fosd}}{Proposition 3}}

\noindent $(a) \implies (b)$: Suppose $\nu$ is an exposed point of $\text{LSD}(\mu)$. Let $f$ be a continuous function such that $\int f \d \nu$ exposes $\nu$. By \Cref{prop:exposed}, there exists an order-preserving kernel $P^f$ such that $\nu = P^f * \mu$, and for $\mu$-a.e. $x \in X$, and hence Lebesgue a.e. $x\in X$ (since $\mu$ is mutually absolutely continuous with Lebesgue measure), we have   
\begin{equation}
    \int f(y) \eta(\d y) = \overline{f}(x) \text{ and } \eta \preceq_{\text{nondecreasing}} \delta_x \iff \eta = P^f_x \,. \label{eq:uniquefosd}
\end{equation}
We claim that $f$ must have a strict monotone contact---indeed, if there exists a positive-Lebesgue-measure set of $x \in S$ such that there exists $y \leq x$ and $y\neq x$ with $f(y) = f(x) = \overline{f}(x)$, then for such a positive measure of $x$, the pointwise maximization problem admits multiple solutions, contradicting the above. 

Now, fix $z \in \text{Ran}(\overline{f})$. Let  
\[A_z := \Big\{x: \text{there exist two distinct $t, t' \in T_z$ with $t\leq x$, $t' \leq x$} \Big\} \cup \Big\{x: \overline{f}(x) > z\Big\}\,.\]
Note that $A_z$ is upward-closed. We claim that 
\[R_z \backslash A_z = \biguplus_{\underline{x} \in T_z}\Big(\big\{\,x: x \geq \underline{x} \,\big\}\, \backslash \,A_z\Big) \,.\]
Indeed, if $x \in R_z \backslash A_z$, then we have $\overline{f}(x) = z = \max_{y \leq x} f(y)$, and hence there exists $t \leq x$ such that $f(t) = z$. Since $t \leq x$, we have \[\overline{f}(t) \leq \overline{f}(x) = z\,,\] 
but with $\overline{f}(t) \geq f(t) = z$, we have $t \in T_z$. Moreover, since $x\not\in A_z$, there exists a unique $t \leq x$ such that $t \in T_z$. Thus, $x$ belongs to exactly one piece on the right-hand side $\big\{\,x: x \geq t \,\big\}\, \backslash \,A_z$. Conversely, if $x \geq t$ for some $t \in T_z$ and $x \not\in A_z$, then 
\[\overline{f}(x) \geq \overline{f}(t) = f(t) = z\,,\]
and moreover, 
\[\overline{f}(x) \leq z\]
since $x \not\in A_z$, and hence 
\[\overline{f}(x) = z\,,\]
and so $x \in R_z \backslash A_z$. This then proves that 
\[R_z \backslash A_z = \biguplus_{\underline{x} \in T_z}\Big(\big\{\,x: x \geq \underline{x} \,\big\}\, \backslash \,A_z\Big) \,,\]
with the disjointness follows from that, as we have shown, any element in the LHS belongs to a unique piece of the RHS. 

Now, fix any $x \in S^c$ such that \eqref{eq:uniquefosd} holds, which is a full-measure set in $S^c$ by \Cref{prop:exposed}. Let $z^x := \overline{f}(x)$. By definition, $x \in R_{z^x}$. We show that $x \not\in A_{z^x}$. Indeed, suppose for contradiction that $x \in A_{z^x}$. Then, since $\overline{f}(x) = z^x$, we must have that there exist two distinct $t, t' \in T_{z^x}$ with $t, t' \leq x$---however, then there exist two distinct solutions to the pointwise problem at $x$---a contradiction. Therefore, $x \in R_{z^x}\backslash A_{z^x}$. Then, by our previous argument, it follows that there exists a unique $\underline{x} \in T_{z^x}$ with $\underline{x} \leq x$. Since $P^f_x$ is the unique pointwise solution at $x$, it must be that $P^f_x = \delta_{\underline{x}}$, which is exactly the unique downward transport of $x$. 

Thus, for a full-measure set of $x \in S^c$, the region $S^c_{z^x}$ is downward staircase-like, $x \in R_{z^x} \backslash A_{z^x}$, and $P^f(\,\cdot\,\mid x)$ is the unique downward transport. Moreover, since $f$ has a strict monotone contact, $P^f(\,\cdot\,\mid x) = \delta_x$ for a full-measure set of $x \in S$. Thus, without loss of generality, we may take $P^f(\,\cdot\,\mid x) = \delta_x$ for all $x \in S$, yielding the desired construction.  \\

\noindent $(b) \implies (a)$: Now suppose $(b)$ holds. Fix any such $f$, $P^f$, and $\nu = P^f * \mu$. We verify the pointwise uniqueness condition in \Cref{prop:exposed}. 

First, fix any $x$ in the full-measure subset of $S$ where strict monotone contact holds. Since $\overline{f}(x) = f(x)$, $\delta_x$ is feasible and optimal. Moreover, any optimal $\eta$ solving the pointwise problem at $x$ must be supported on points $y \leq x$ with $f(y) = f(x)$, which, by strict monotone contact, implies that $\eta = \delta_x$. 

Now, fix any $x$ in the full-measure subset of $S^c$ from statement $(b)$. Let $z^x := \overline{f}(x)$. By assumption, 
\[x \in R_{z^x} \backslash A_{z^x} = \biguplus_{\underline{x} \in T_{z^x}}\Big(\big\{\,x': x' \geq \underline{x} \,\big\}\, \backslash \,A_{z^x}\Big) \,,\]
and hence there exists a unique $\underline{x} \in T_{z^x}$ with $\underline{x} \leq x$. By construction, $P^f(\,\cdot\,\mid x) = \delta_{\underline{x}}$, which is a feasible solution and in fact optimal since $f(\underline{x}) = z^x = \overline{f}(x)$. Moreover, if 
\[\eta \preceq_{\text{nondecreasing}} \delta_x  \text{ and } \int f(y) \eta(\d y) = \overline{f}(x) = z^x\,,\]
then $\eta$ must be supported on points $y \leq x$ with $f(y) = z^x$. For any such $y$, in fact, we have 
\[\overline{f}(y) \leq \overline{f}(x) = z^x \,\text{ and }\, \overline{f}(y) \geq f(y) = z^x\,,\]
we must have 
\[\overline{f}(y) = f(y) = z^x\,,\]
and hence $y \in T_{z^x}$. Hence, $\eta$ must be supported on points of $T_{z^x}$ that are weakly below $x$. However, as we have argued, $\underline{x}$ is the unique such point, and hence $\eta = \delta_{\underline{x}} = P^f(\,\cdot\,\mid x)$. 

Therefore, we conclude that for a full-measure set of $x \in X$, $P^f(\,\cdot\,\mid x)$ is the unique solution to the pointwise problem at $x$ with the objective given by $f$. By \Cref{prop:exposed}, $\nu$ is an exposed point of $\text{LSD}(\mu)$, completing the proof. \hfill \qedsymbol

\subsection{Proof \texorpdfstring{\Cref{thm:main-compare}}{Theorem 3}}
We prove this in the order of 
\[
(a) \implies (b) \implies (c) \implies (d) \implies (e) \implies (a)\,.
\]

\noindent $(a) \implies (b)$: We first show that $\Pc=\Psi(\C)$. First, note that $\Pc \subseteq \Psi(\C)$. Indeed, for any $P \in \Pc$ and for any $x \in X$, since $P(\,\cdot\,\mid x)=P*\delta_x$ and since $(\C,\Pc)$ is Blackwell-consistent, $P(\,\cdot\,\mid x) \succeq_{\C} \delta_x$. Thus, for all $g \in \C$
\[
\int g(y)P(\d y\mid x) \geq g(x)
\]
for all $x \in X$, which implies that $P \in \Psi(\C)$. As a result, 
\[
\nu=P*\mu\,\text{ for some $P \in \Pc$} \implies \nu=P*\mu\,\text{ for some $P \in \Psi(\C)$}\,.
\]

Next, consider any $\mu,\nu \in \Delta(X)$ such that $\nu=P*\mu$ for some $P \in \Psi(\C)$. Since $P \in \Psi(\C)$, we have that $g*P \geq g$ for all $g \in \C$. Therefore, for any $g \in \C$, 
\[
\int g \d \nu=\int g \d (P*\mu)=\int g*P \d \mu \geq \int g \d \mu\,.
\]
That is, $\mu \preceq_{\C} \nu$. As a result, since $(\C,\Pc)$ is Blackwell-consistent, we have 
\[
\nu=P*\mu\,\text{ for some $P \in \Psi(\C)$} \implies \mu \preceq_{\C} \nu \iff \nu=P*\mu\,\text{ for some $P \in \Pc$}\,.
\]
Together, these imply that 
\[
\nu=P*\mu\,\text{ for some $P \in \Pc$} \iff \nu=P*\mu\,\text{ for some $P \in \Psi(\C)$}\,.
\]
Furthermore, by \Cref{lem:basic}, $\Psi(\C)$ is rectangular. Also, $\Pc$ is rectangular as it is regular. Therefore, by \Cref{lem:rectangle}, we conclude that $\Pc=\Psi(\C)$.  

Next, we show that $\C=\Phi(\Pc)$. Since $\Pc=\Psi(\C)$, \Cref{lem:basic} implies that $\C \subseteq \Phi(\Psi(\C))=\Phi(\Pc)$. We show that $\Phi(\Pc) \subseteq \C$. By the proof of the ``only if'' part of \Cref{thm:main2}, since the order $\preceq_\C$ is Blackwell-consistent, $\C$ must be max-closed. Now, fix any $g \in \Phi(\Pc)$. Let 
\[\widehat{g} = - \overline{(-g)}^{[-\C]}\,,\]
i.e., $-\widehat{g}$ is the $[-\C]$-envelope of $-g$. Note that 
\[-\widehat{g}(x) = \inf\big\{h(x): h \in [-\C]\,,\, h \geq - g\big\} = -\sup\big\{k(x): k \in \C\,,\, k \leq g\big\}\,.\]
Note that under our regularity on $\C$, the convex cone $[-\C]$ satisfies all the assumptions in \Cref{sec:abstract}. Moreover, since $\C$ is max-closed, $[-\C]$ is min-closed. By \Cref{lem:envelope}, $\overline{(-g)}^{[-\C]} \in -\C$ and hence $\widehat{g} \in \C$. We now argue that $\widehat{g} = g$. Indeed, consider the maximization problem 
\[V^\star_{-g}(\delta_x) := \max_{\eta \preceq_{[-\C]} \delta_x} \int [-g](x) \eta(\d x)\,.\]
Since $[-\C]$ is min-closed, by $(a) \implies (b)$ direction of \Cref{thm:main}, we have the following strong duality: for any $x \in X$, we have 
\[\max_{\eta \preceq_{[-\C]} \delta_x} \int [-g](y) \eta(\d y) = V^\star_{-g}(\delta_x) = -\widehat{g} (x) = \min_{h \in [-\C],\, h \geq -g} h(x)\,.\]
For any $x \in X$, note that 
\[\eta \in \Psi(\C)_{x} \iff \int k(y) \eta(\d y) \geq k(x) \,\,\,\forall k \in \C \iff \delta_x \preceq_{\C} \eta \iff \eta \preceq_{[-\C]} \delta_x\,.\]
Therefore, since $g\in \Phi(\Pc) = \Phi(\Psi(\C))$, for any $x$, and any $P \in \Psi(\C)$, we have 
\[\int g(y) P(\d y \mid x) \geq g(x)\,,\]
which is equivalent to that for any $x$, we have 
\[\inf_{P \in \Psi(\C)} \int g(y) P(\d y \mid x) \geq g(x)\,,\]
which, in turn, is equivalent to 
\[\inf_{\eta \in \Psi(\C)_x} \int g(y) \eta(\d y) \geq g(x)\,.\]
For any $x$, since $\delta_x \in \Psi(\C)_x$ by definition, the above holds if and only if 
\[\min_{\eta \in \Psi(\C)_x} \int g(y) \eta(\d y) = g(x)\,,\]
which holds, by our previous observation, if and only if 
\[\min_{\eta \preceq_{[-\C]} \delta_x} \int g(y) \eta(\d y) = g(x)\,.\]
Thus, for any $x$, we obtain
\[-\widehat{g}(x) = \max_{\eta \preceq_{[-\C]} \delta_x} \int [-g](y) \eta(\d y) = - \min_{\eta \preceq_{[-\C]} \delta_x} \int g(y) \eta(\d y) = -g(x)\,,\]
and hence $\widehat{g}(x) = g(x)$. Thus, $g = \widehat{g} \in \C$. Hence, $\Phi(\Pc) \subseteq \C$, as desired. \\
\mbox{}

\noindent $(b) \implies (c)$: Since $(\C,\Pc)$ is Blackwell-invariant, we have $\C=\Phi(\Pc)$. It thus suffices to show that $\Pc$ is one-shot sufficient. Indeed, by Blackwell invariance, $\Pc =\Psi(\C)$, which by \Cref{lem:basic}, is composition-closed and hence one-shot sufficient. \\
\mbox{}

\noindent $(c) \implies (d)$: Since $\Pc$ is regular, $\mathrm{Id} \in \Pc$, and hence $\Pc \subseteq \Pc \otimes \Pc$. To see that $\Pc \otimes \Pc \subseteq \Pc$, consider $Q \in \Pc \otimes \Pc$ and suppose that, by way of contradiction, $Q \notin \Pc$. Since $Q \notin \Pc$ and $\Pc$ is rectangular, there exists $x \in X$ such that $Q(\,\cdot\,\mid x) \notin \Pc_x$. Fix such $x$. Since $\Pc$ has a closed graph, $\Pc_x$ is closed. Since $\Pc_x$ is a closed convex set in $M(X)$, the set of finite signed Borel measures under the weak-* topology, which is a locally convex topological vector space, by the strict Hahn–Banach separation theorem (see, e.g., \citealt{lukevs2009integral}, Theorem A.1), there exists a continuous function $h \in C(X)$ (since $C(X) = [M(X)]^*$) such that 
\[
\int h(y) Q(\d y\mid x) > \sup_{\mu \in \Pc_x}\int h(y) \mu(\d y)\,.
\]
By definition, note that we have 
\[
\sup_{\mu \in \Pc_x} \int h(y) \mu (\d y)= \sup_{\mu \in \{P(\,\cdot\, \mid x):\, P \in \Pc\}} \int h(y) \mu (\d y) =\sup_{P \in \Pc} \int h(y) P(\d y \mid x)\,.
\]
Meanwhile, since $Q \in \Pc \otimes \Pc$, we have 
\[
\int h(y) Q(\d y\mid x) \leq \sup_{P \in \Pc \otimes \Pc} \int h(y) P(\d y\mid x)\,.
\]
Together, these imply that 
\[
\sup_{P \in \Pc} \int h(y) P(\d y\mid x)<\int h(y) Q(\d y\mid x) \leq \sup_{P \in \Pc \otimes \Pc} \int h(y) P(\d y\mid x)\,,
\]
contradicting $\Pc$ being one-shot sufficient. Therefore, we have $\Pc=\Pc \otimes \Pc$. \\
\mbox{}

\noindent $(d) \implies (e)$: By \Cref{lem:basic}, $\C =\Phi(\Pc)$ is max-closed. We show that $\Psi(\C)=\Pc$. By the rectangularity property of $\Psi(\C)$ and $\Pc$, it suffices to show $\preceq^{\Psi(\C)} = \preceq^\Pc$ by \Cref{lem:rectangle}. To show this, we will show that 
\[ \preceq^{\Psi(\C)} \implies \preceq_{\C} \implies \preceq^\Pc \implies \preceq^{\Psi(\C)}\,.\]

To see $\preceq^{\Psi(\C)} \implies \preceq_{\C}$, note that if $\mu \preceq^{\Psi(\C)} \nu$, then there exists some $P \in \Psi(\C)$ such that $\nu = P* \mu$. By the definition of $\Psi(\C)$, for all $g \in \C$ and all $x \in X$, we have 
\[\int g(y) P(\d y \mid x) \geq g(x)\,.\]
Thus, for all $g \in \C$, we have 
\[\int g \d \nu = \int g \d (P * \mu) = \int \Big(\int g(y) P(\d y \mid x) \Big) \mu(\d x)\geq \int g \d \mu \,,\]
and hence $\mu \preceq_\C \nu$. 

To see $\preceq^\Pc \implies \preceq^{\Psi(\C)}$, note that $\Pc \subseteq \Psi(\Phi(\Pc))) = \Psi(\C)$ by \Cref{lem:basic}. Thus, for any $\mu, \nu \in \Delta(X)$ where  $\nu = P * \mu$ for some $P \in \Pc$, we immediate have that  $\nu = P * \mu$ for some $P \in \Psi(\C)$. Thus, $\preceq^\Pc \implies \preceq^{\Psi(\C)}$. 

Thus, to complete the argument, it suffices to show $\preceq_{\C}  \implies \preceq^{\Pc}$. To this end, fix any $\mu,\nu \in \Delta(X)$ such that $\mu \preceq_\C \nu$ and let 
\[
M(\mu):=\{P*\mu: P \in \Pc\}\,.
\]
Suppose for contradiction that $\nu \notin M(\mu)$. By \Cref{lem:compactm}, $M(\mu) \subseteq \Delta(X)$ is compact and convex. Since $M(\mu)$ is compact, $M(\mu)$ is closed. Since $M(\mu)$ is a closed convex set in the space of finite signed Borel measures equipped with the weak-* topology, by the strict Hahn–Banach separation theorem (see, e.g., \citealt{lukevs2009integral}, Theorem A.1), there exists a continuous linear functional that strictly separates $\{\nu\}$ from $M(\mu)$. In particular, there exists a continuous function $f \in C(X)$ such that 
\[
\int f \d \nu<\inf_{\eta \in M(\mu)} \int f \d \eta = \inf_{P \in \Pc} \int f \d(P*\mu)= \int \Big(\inf_{\eta \in \Pc_x} \int f \d \eta\Big) \mu(\d x) = \int \Big(\min_{\eta \in \Pc_x} \int f \d \eta\Big) \mu(\d x)\,,
\]
where we have also used the rectangularity of $\Pc$, the Kuratowski–Ryll-Nardzewski measurable selection theorem, and the fact that $f$ is continuous and $\Pc_x$ is compact for all $x \in X$ (which we have shown previously).

Now, for any $x \in X$, let 
\[
\underline{f}(x) := \min_{\eta \in \Pc_x} \int f(y) \eta(\d y)\,. 
\]
We claim that $\underline{f}$ is lower semicontinuous. Indeed, the correspondence $x \mapsto \Pc_x$ has a closed graph and hence is upper hemicontinuous. The objective is given by $- \max_{\eta \in \Pc_x} \int -f(y) \eta(\d y)$ where $-f$ is continuous. Thus, by Berge's maximum theorem, $-\underline{f}$ is upper semicontinuous, and hence $\underline{f}$ is lower semicontinuous. So $\underline{f}\in \C_0$.

We claim that $\underline{f} \in \Phi(\Pc)$. Indeed, for any $P \in \Pc$ and any $x \in X$, we have 
\[
\int \underline{f}(y) P(\d y\mid x) = \int \Big(\min_{\eta \in \Pc_y} \int f(z) \eta(\d z) \Big) P(\d y\mid x)  =  \min_{Q \in \Pc} \int  f(y)  (Q \circ P)(\d y\mid x) 
\]
where the second equality again uses that $\Pc$ is rectangular and the Kuratowski–Ryll-Nardzewski measurable selection theorem. Moreover, for any $P \in \Pc$ and any $x \in X$, note that 
\[\min_{Q \in \Pc} \int  f(y)  (Q \circ P)(\d y\mid x)  \geq \inf_{Q \in \Pc \otimes  \Pc} \int  f(y)  Q(\d y\mid x) = \inf_{Q \in \Pc } \int  f(y)  Q(\d y\mid x) = \underline{f}(x)\,,\]
where the first equality uses that $\Pc$ is composition-closed and the second equality uses that the $\inf$ can be attained by the continuity of $f$ and compactness of $\Pc_x$ with the value given by $\underline{f}(x)$. Together, these imply that for any $P \in \Pc$ and any $x \in X$, 
\[\int \underline{f}(y) P(\d y\mid x) \geq \underline{f}(x)\,,\]
and hence $\underline{f} \in \Phi(\Pc)$. As a result, since $\mu \preceq_{\Phi(\Pc)} \nu$, we must have 
\[
\int \underline{f} \d \mu \leq \int \underline{f} \d \nu\,.
\]
Moreover, since $\mathrm{Id} \in \Pc$, we have $\underline{f} \leq f$ by construction. Together, these imply that 
\begin{align*}
\int \underline{f} \d \nu \leq \int f \d \nu
 < \int \underline{f} \d \mu \leq \int \underline{f} \d \nu\,,
\end{align*}
where the strict inequality was shown earlier using the fact that the separating linear functional is given by $f$---a contradiction. Therefore, $\preceq_{\C} \implies \preceq^\Pc$. 

Together, these imply that 
\[
\preceq^{\Psi(\C)}= \preceq_{\C} = \preceq^{\Pc}\,.
\]
Then, since $\Psi(\C)$ is rectangular by \Cref{lem:basic} and $\Pc$ is rectangular as it is regular, by \Cref{lem:rectangle}, it follows that $\Pc=\Psi(\C)$, as desired. \\
\mbox{}

\noindent $(e)\implies (a)$: Suppose $\C$ is max-closed. We claim that $\preceq_{\C}=\preceq^{\Psi(\C)}$. By the same argument as in the proof of $(d)\implies (e)$, we have that $\preceq^{\Psi(\C)} \implies \preceq_{\C}$. Now, to see the other direction, fix any $\mu \preceq_\C \nu$. Then, $\nu  \preceq_{[-\C]} \mu$ where the convex cone $[-\C]$ is min-closed. By \Cref{thm:main},\footnote{Note that the technical assumptions required by \Cref{thm:main} hold for $[-\C]$ given that $\C$ is regular.} there exists some $P$ such that $\nu = P * \mu$, and $P * \delta_x \preceq_{[-\C]} \delta_x$ for all $x \in X$. Thus, for all $x \in X$ and all $g \in \C$, we have 
\[ \int [-g(y)]P(\d y \mid x) \leq -g(x)\,,\]
and thus 
\[ \int g(y) P(\d y \mid x) \geq g(x)\,.\]
Therefore, $P \in \Psi(\C)$ by definition, which in turn implies that $\mu \preceq^{\Psi(\C)} \nu$, by definition. Thus, $\preceq_{\C} \implies \preceq^{\Psi(\C)}$. Hence, we have 
\[\preceq_{\C} = \preceq^{\Psi(\C)}\,.\]
Since $\Pc=\Psi(\C)$, it follows $\preceq_{\C}  = \preceq^{\Pc}$. Hence, $(\C,\Pc)$ is Blackwell-consistent.\\
\mbox{}

\noindent \textbf{Uniqueness:} We first show that every Blackwell-consistent pair $(\C, \Pc)$ is maximal in the following sense: $\preceq_{\C'}=\preceq_{\C} \implies \C' \subseteq \C$ for any $\C' \subseteq \C_0$, and $\preceq^{\Pc'}=\preceq^{\Pc} \implies \Pc' \subseteq \Pc$ for any $\Pc' \subseteq \Pc_0$. 

Fix any Blackwell-consistent pair $(\C, \Pc)$. By $(a) \implies (b)$ direction, we know that $(\C, \Pc)$ is Blackwell-invariant. Thus, we have 
\[\Pc = \Psi(\C)\,\,\text{ and }\,\, \C = \Phi(\Pc)\,.\]

First, fix any $\C'\subseteq \C_0$ such that $\preceq_{\C'} = \preceq_\C$. Fix any $g \in \C'$. We claim that $g \in \C$. Indeed, suppose for contradiction $g \not\in \C$. Since $g \in \C_0$ and $\C = \Phi(\Pc)$, for $g \not \in \C$, there must exist some $P \in \Pc$ and some $x \in X$ such that 
\[ \int g(y) P(\d y \mid x) < g(x)\,.\]
By construction, $\delta_x \preceq^\Pc P * \delta_x$. Moreover, by construction, $\delta_x \not \preceq_{\C'} P * \delta_x$, which implies that $\delta_x \not \preceq_{\C} P * \delta_x$---contradicting to $\preceq^\Pc = \preceq_\C$. This shows that $\C$ is maximal. 

Now, we show that $\Pc$ is maximal. Fix any $\Pc'\subseteq \Pc_0$ such that $\preceq^{\Pc'} = \preceq^\Pc$. Fix any $P \in \Pc'$. We claim that $P \in \Pc$. Indeed, suppose for contradiction $P \not\in \Pc$. Since $P \in \Pc_0$ and $\Pc = \Psi(\C)$, for $P \not \in \Pc$, there must exist some $g \in \C$ and some $x \in X$ such that 
\[ \int g(y) P(\d y \mid x) < g(x)\,.\]
By construction, $\delta_x \preceq^{\Pc'} P * \delta_x$, which implies that  $\delta_x \preceq^{\Pc} P * \delta_x$. Moreover, by construction, $\delta_x \not \preceq_{\C} P * \delta_x$---contradicting to $\preceq^\Pc = \preceq_\C$. This shows that $\Pc$ is maximal. 

Now, we prove the unique representation result. Fix any Blackwell-consistent order $\preceq$. Let $(\C, \Pc), (\C', \Pc')$ be two (regular) Blackwell-consistent pairs supporting the order $\preceq$, i.e.,
\[\preceq_{\C'} = \preceq_\C = \preceq = \preceq^\Pc = \preceq^{\Pc'}\,.\] 
By the previous claim about the maximality of $(\C, \Pc)$, we have $\C' \subseteq \C$ and $\Pc' \subseteq \Pc$. Now, applying the same claim to the pair $(\C', \Pc')$ gives $\C \subseteq \C'$ and $\Pc \subseteq \Pc'$. Therefore, we conclude $(\C', \Pc') = (\C, \Pc)$. Thus, there exists a unique pair $(\C, \Pc)$ that represents the Blackwell-consistent order $\preceq$.\hfill \qedsymbol

\subsection{Proof of \texorpdfstring{\Cref{prop:comparison}}{Proposition 4}}

Suppose that the constrained problem is not value-eliminating. Fix any prior $x$, under which the unconstrained optimal signal is unique. Clearly, if $\nu_{\text{KG}} = \delta_x$, then we immediately have that $\nu^\star = \delta_x$ is the unique solution for the constrained problem since $\delta_x$ is feasible. Thus, suppose otherwise.  Then, we must have $\widehat{f}^{\text{cav}}(x) > f(x)$. Let 
\[S = \text{conv}\big(\supp(\nu_{\text{KG}})\big) \backslash \supp(\nu_{\text{KG}})
\,.\]
Fix any constrained optimal $\nu^\star$. We claim that $\nu^\star(S) = 0$. Suppose for contradiction that $\nu^\star(S) > 0$. Then by \Cref{cor:envelope}, there exists some $y \in S$ such that 
\[\widehat{f}^{\C}(y) = f(y)\]
which by the assumption that the constrained problem is not value-eliminating implies  
\[\widehat{f}^{\text{cav}}(y) = \widehat{f}^{\C}(y) =f(y)\,.\]
Now, note that since $\nu_{\text{KG}}$ is finitely supported and assigns strictly positive mass to every point in $\supp(\nu_{\text{KG}})$, its barycenter $x$ is in the relative interior of $\text{conv}(\supp(\nu_{\text{KG}}))$. Then, given the above, $\nu_{\text{KG}}$ cannot be uniquely optimal since we can perturb $\nu_{\text{KG}}$ by constructing 
\[\nu' := \nu_{\text{KG}} + \varepsilon \delta_{y} - \varepsilon \sum_{i} \lambda_i \delta_{x^{\text{KG}}_i}\,,\]
for small enough $\varepsilon > 0$, where $\{x^{\text{KG}}_i\}_i := \supp(\nu_{\text{KG}})$ and $\lambda$ is the barycentric representation of $y$ supported on $\supp(\nu_{\text{KG}})$. Indeed, for small enough $\varepsilon > 0$, $\nu'$ is a feasible signal distinct from $\nu_{\text{KG}}$, and, since $\widehat{f}^{\text{cav}}(y) = f(y)$ and $y \in S$, the perturbed signal weakly improves the objective value and hence is optimal---a contradiction. Thus, $\nu^\star(S) = 0$. 

However, for $\nu^\star$ to be Blackwell-dominated by $\nu_{\text{KG}}$, it must be that 
\[\nu^\star\Big(\text{conv}\big(\supp(\nu_{\text{KG}})\big)\Big) = 1\,,\]
which then combined with $\nu^\star(S) = 0$ implies that 
\[\nu^\star\Big(\supp(\nu_{\text{KG}})\Big) = 1\,,\]
but since $\nu^\star$'s barycenter is also $x$, we must have $\nu^\star = \nu_{\text{KG}}$ given that $\nu_{\text{KG}}$ is the unique solution to the unconstrained problem. Thus, $\nu^\star$ cannot be strictly Blackwell-dominated by $\nu_{\text{KG}}$. Now, consider the case of $|\Omega| = 2$. Since $\nu^\star(S) = 0$, we have that 
\[\supp(\nu^\star) \subseteq [0, 1] \backslash (x^{\text{KG}}_1, x^{\text{KG}}_2)\,,\]
where $\{x^{\text{KG}}_1, x^{\text{KG}}_2\}$ is the support of $\nu^{\text{KG}}$. We claim that $\nu^\star$ weakly Blackwell-dominates $\nu^{\text{KG}}$---indeed, this follows by observing that we can fuse the mass of   $\nu^\star$ supported on $[0, x^{\text{KG}}_1]$ into a single point $x^\star_1$ and the mass of $\nu^\star$ supported on $[x^{\text{KG}}_2, 1]$ into a single point $x^\star_2$. The fused distribution weakly Blackwell-dominates $\nu_{\text{KG}}$ and is weakly Blackwell-dominated by $\nu^\star$. Thus, the result holds. \hfill \qedsymbol

\subsection{Proof of \texorpdfstring{\Cref{prop:divisible}}{Proposition 5}}
Note that $\cQ^d$ is rectangular by definition. Moreover, since $\hat{d}(x_B,x_D,x_B)=x_D$ and $\Pc_M$ contains the identity kernel, $\cQ^d$ contains the identity kernel as well. Lastly, since $d$ is continuous and $\Pc_M$ has a closed graph and is convex, $\cQ^d$ has a closed graph and is convex as well. Therefore, $\cQ^d$ is regular for any $d \in \mathcal{D}$. \\

\noindent $ (\impliedby) $: Suppose that $d$ is divisible. By \Cref{thm:main-compare}, it suffices to prove that $\cQ^d$ is composition-closed and that $\C^d=\Phi(\cQ^d)$ is regular, since these then imply that $(\C^d,\cQ^d)$ is a regular pair, and moreover by \Cref{thm:main-compare}, a Blackwell-consistent pair---and thus, $\preceq^{\cQ^d}$ is a Blackwell-consistent order. 

To see that $\cQ^d$ is composition-closed, consider any $Q_1,Q_2 \in \cQ^d$. Fix any $(x_B,x_D)$. By definition, $Q_1(\,\cdot\, \mid x_B,x_D)$ is the distribution of 
\[
(Y,\hat{d}(x_B,x_D,Y))\,,
\]
where $Y \sim P_1^{(x_B,x_D)}(\,\cdot\,\mid x_B)$ for some $P_1^{(x_B,x_D)} \in \Pc_M$. Moreover, $Q_2 \circ Q_1(\,\cdot\,\mid x_B,x_D)$ is the distribution of 
\[
\left(Z,\hat{d}\left(Y,\hat{d}(x_B,x_D,Y),Z\right)\right)\,,
\]
where $Z \sim P_2^{(x_B,x_D)}(\,\cdot\,\mid y)$ conditional on $Y=y$ for some $P^{(x_B,x_D)}_2 \in \Pc_M$. Note that, for any $y,z \in X^0$, since $d$ is divisible, by letting 
\[
\tilde{y}:=\frac{x_D \oslash x_B \odot  y}{\langle x_D \oslash x_B, y\rangle}\,\, \text{ and }\,\, \tilde{z}:=\frac{x_D \oslash x_B \odot  z}{\langle x_D \oslash x_B, z\rangle}\,,
\]
we have
\begin{align*}
\hat{d}(y,\hat{d}(x_B,x_D,y),z)=&d\left(d\left(x_D,\, \frac{x_D\oslash x_B \odot y}{\langle x_D\oslash x_B,y\rangle}\right),\frac{d\left(x_D,\,\frac{x_D\oslash x_B \odot y}{\langle x_D\oslash x_B,y\rangle}\right)\oslash y \odot z}{\langle d\left(x_D,\,\frac{x_D\oslash x_B \odot y}{\langle x_D\oslash x_B,y\rangle}\right)\oslash y,z\rangle}\right)\\
=&d\left(d(x_D,\tilde{y}),\frac{d(x_D,\tilde{y})\oslash \tilde{y} \odot \tilde{z}}{\langle d(x_D,\tilde{y})\oslash \tilde{y}\,,\,\tilde{z} \rangle}\right)\\
=&d_{\mathrm{II}}(x_D,\tilde{z}\,;\,\tilde{y})\\
=&d(x_D,\,\tilde{z})\\
=&d\left(x_D, \, \frac{x_D \oslash x_B \odot  z}{\langle x_D \oslash x_B, z\rangle}\right)\\
=&\hat{d}(x_B,x_D,z)\,.
\end{align*}
This also extends to all $y,z \in X$ by continuity of $d$. Moreover, since $Z \sim P_2^{(x_B,x_D)}(\,\cdot\, \mid y)$ conditional on $Y=y$, since $Y \sim P_1^{(x_B,x_D)}(\,\cdot\,\mid x_B)$, and $P_1^{(x_B,x_D)},P_2^{(x_B,x_D)} \in \Pc_M$,
\[
\E[Z]=\int_X \left(\int_X z\, P_2^{(x_B,x_D)}(\d z\mid y)\right)P_1^{(x_B,x_D)}(\d y\mid x_B)=\int_X y P_1^{(x_B,x_D)}(\d y\mid x_B)=x_B\,.
\]
As a result, $Q_2 \circ Q_1$ is the distribution of 
\[
(Z,\hat{d}(x_B,x_D,Z))\,,
\]
for some random variable $Z$ with $\E[Z]=x_B$, and hence $Q_2\circ Q_1(\,\cdot\,\mid x_B,x_D) \in \cQ^d_{(x_B,x_D)}$. Since $(x_B,x_D)$ is arbitrary, we have that $Q_2 \circ Q_1  \in \cQ^d$, and thus $\cQ^d$ is composition-closed.

To see that $\C^d=\Phi(\cQ^d)$ is regular. Note that by \Cref{lem:basic}, $\C^d$ is a convex cone that contains $1,-1$ and is closed under bounded increasing pointwise limits, with $\C^d \cap C(X)$ being closed under the sup-norm. Thus, it remains to prove that $\C^d$ can be approximated by continuous functions in $\C^d$ from below. Note that by Theorem 1 of \citet{cripps2018divisible}, there exists a homeomorphism $F:X \to X$ such that for all $x,y \in X^0$,
\[
d(x,y)=F^{-1}\left(\frac{F(x) \oslash x \odot y}{\langle F(x) \oslash x,y \rangle}\right)\,.
\]
Recall that $d^\mathrm{B}(x,y)=y$ denotes the Bayesian updating rule. It follows that 
\[
F(\hat{d}(x_B,x_D,y))=\hat{d}^\mathrm{B}(x_B,F(x_D),y)\,,
\]
for all $x_B,x_D,y \in X^0$. 

Let $\xi:X\times X \to X \times X$ be defined as 
\[
\xi(x_B,x_D):=(x_B,F(x_D))\,,
\]
for all $(x_B,x_D) \in X \times X$. Since $F$ is a homeomorphism, $\xi$ is a homeomorphism as well. We claim that for all $g \in \C_0$, 
\[
g \in \C^d \iff g \circ \xi^{-1} \in \C^{d^\mathrm{B}}\,,
\]
where $\C^{d^\mathrm{B}} = \Phi(\cQ^{d^\mathrm{B}})$. Indeed, note that, for any $g \in \C^d$, for any $Q \in \cQ^d$, and for any $(x_B,x_D) \in X \times X$, letting $\eta \in (\Pc_M)_{x_B}$ be the distribution with barycenter $x_B$ that defines $Q(\,\cdot\,\mid x_B,x_D)$, we can write 
\begin{align*}
\int g(y_B,y_D) Q(\d y_B, \d y_D\mid x_B,x_D)= \int g(y,\hat{d}(x_B,x_D,y))\eta(\d y)= \int g \circ \xi^{-1} (y,\hat{d}^{\mathrm{B}}(x_B,F(x_D),y))\eta ( \d y)\,.
\end{align*}
Thus, $g\in \C^d$ if and only if for all $(x_B, x_D) \in X \times X$, we have 
\[\min_{\delta_{x_B} \preceq_{\text{convex}} \eta}\int g \circ \xi^{-1} (y,\hat{d}^{\mathrm{B}}(x_B,F(x_D),y))\eta ( \d y) = g(x_B, x_D) = g \circ \xi^{-1} (x_B, F(x_D))\,,\]
which, since $F$ is a homeomorphism, is equivalent to that for all $(x_B, x'_D) \in X \times X$, we have 
\[\min_{\delta_{x_B} \preceq_{\text{convex}} \eta}\int g \circ \xi^{-1} (y,\hat{d}^{\mathrm{B}}(x_B,x'_D,y))\eta ( \d y) = g \circ \xi^{-1} (x_B, x'_D)\,,\]
which is equivalent to $g \circ \xi^{-1} \in \Phi(\cQ^{d^\mathrm{B}}) = \C^{d^\mathrm{B}}$, as desired. 

Therefore, since $\xi$ is a homeomorphism, it suffices to show that every $g \in \C^{d^{\mathrm{B}}}$ can be approximated by a sequence of continuous functions from below. To prove this, let
\[
\ell_{x_B,x_D}(y):=\hat{d}^{\mathrm{B}}(x_B,x_D,y)\,,
\]
for all $(x_B,x_D,y) \in X \times X \times X$. 
Since $\hat{d}^{\mathrm{B}}:X \times X\times X\to X$ is continuous and since $X$ is compact, 
\[
\ell_{x_B,x_D}:X\to X
\]
is continuous for all $(x_B,x_D) \in X \times X$. Let $C(X,X)$ be the set of continuous functions from $X$ to $X$ endowed with the sup metric
\[
d_\infty(\ell_1,\ell_2):=\sup_{y\in X}\|\ell_1(y)-\ell_2(y)\|\,.
\]
Since $X$ is compact metric space, $(C(X,X),d_\infty)$ is also a metric space. 

Now, let $\Gamma: X \times X \to C(X,X)$ be defined as
\[
\Gamma(x_B,x_D):=\ell_{x_B,x_D}
\]
for all $(x_B,x_D) \in X \times X$. We claim that $\Gamma$ is continuous. Indeed, suppose that there is a sequence $\{(x_B^n,x_D^n)\}$ that converges to some $(x_B,x_D)$ in $X\times X$. Since $\hat{d}^{\mathrm{B}}$ is continuous on the compact set $X \times X \times X$, it is uniformly continuous. Therefore, we have  
\[
\lim_{n \to \infty }\sup_{y\in X}
\left\|
\hat{d}^{\mathrm{B}}(x_B^n,x_D^n,y)-\hat{d}^{\mathrm{B}}(x_B,x_D,y)
\right\|=0\,,
\]
and hence 
\[
\lim_{n \to \infty} d_\infty(\ell_{x_B^n,x_D^n},\ell_{x_B,x_D})=0\,.
\]
Thus, $\Gamma$ is continuous. Let
\[
\Lambda:= \Big\{\ell_{x_B,x_D}:(x_B,x_D)\in X\times X \Big\}
\subset C(X,X).
\]
Since $X\times X$ is compact and $\Gamma$ is continuous, its image $\Lambda$ is compact in $C(X,X)$ with the sup metric. In particular, $\Lambda$ is compact and metrizable. 

Now note that since $d^{\mathrm{B}}$ is a divisible updating rule, as shown before, we have 
\[
\hat{d}^{\mathrm{B}}(y,\hat{d}^{\mathrm{B}}(x_B,x_D,y),z)=\hat{d}^{\mathrm{B}}(x_B,x_D,z)\,,
\]
for all $x_B,x_D,y, z \in X$. It then follows that 
\[
[\ell_{y,\ell_{x_B,x_D}(y)}](z)=[\ell_{x_B,x_D}](z)\,.
\]
Thus, for any $\ell \in \Lambda$, we have $\ell=\ell_{x,\ell(x)}$ for all $x \in X$. As a result, for $\Theta: X \times X \to X \times \Lambda$ defined as
\[
\Theta(x_B,x_D):=(x_B,\ell_{x_B,x_D})\,,
\]
we have that $\Theta$ is a homeomorphism. Indeed, $\Theta$ is surjective since for any $(x,\ell) \in X \times \Lambda$, $\Theta(x,\ell(x)) =(x,\ell_{x, \ell(x)}) = (x, \ell)$, where $(x,\ell(x)) \in X \times X$. Meanwhile, $\Theta$ is injective because
\[
\ell_{x_B,x_D}(x_B)=x_D\,,
\]
and hence $\Theta(x_B, x_D) = \Theta(x'_B, x'_D) \implies (x_B, x_D) = (x'_B, x'_D)$. Together, since $\Theta$ is continuous, $\Theta: X \times X \to X \times \Lambda$ is a homeomorphism. 

Note that for all $x_B,x_D,y \in X$
\[
\Theta(y,\hat{d}^{\mathrm{B}}(x_B,x_D,y))= (y, \ell_{y,\hat{d}^{\mathrm{B}}(x_B,x_D,y)}) = (y, \ell_{y,\ell_{x_B, x_D}(y)}) = (y,\ell_{x_B,x_D})\,,
\]
where the last equality uses the divisibility of $d^{\mathrm{B}}$. 

Now, let $\mathcal{U}$ be the set of bounded lower semicontinuous functions $u:X \times \Lambda \to \R$ such that $x \mapsto u(x,\ell)$ is convex for all $\ell \in \Lambda$. We claim that 
\[
g \in \C^{d^{\mathrm{B}}} \iff g \circ \Theta^{-1} \in \mathcal{U}\,,
\]
or, equivalently 
\[
u \in \mathcal{U} \iff u \circ \Theta \in \C^{d^{\mathrm{B}}}\,.
\]
Indeed, for any $g \in \C^{d^{\mathrm{B}}}$, since $g$ is bounded and lower semicontinuous and $\Theta$ is continuous, we have that $g \circ \Theta^{-1}$ is bounded and lower semicontinuous. Moreover, note that $g\in \C^{d^{\mathrm{B}}}$ if and only if for all $(x_B, x_D) \in X \times X$ we have 
\[\min_{\delta_{x_B} \preceq_{\text{convex}} \eta}\int g(y, l_{x_B, x_D}(y)) \eta ( \d y) = \min_{\delta_{x_B} \preceq_{\text{convex}} \eta}\int g(y,\hat{d}^{\mathrm{B}}(x_B,x_D,y)) \eta ( \d y) = g(x_B, x_D)\,.\]
Note that for all $(x_B, x_D) \in X \times X$, we can write  
\[\int g(y, l_{x_B, x_D}(y)) \eta ( \d y) = \int g\circ \Theta^{-1} \circ \Theta(y, l_{x_B, x_D}(y)) \eta ( \d y) = \int g\circ \Theta^{-1}(y,\ell_{x_B,x_D}) \eta ( \d y)\,,\]
and write 
\[g(x_B, x_D) = g \circ \Theta^{-1} \circ \Theta(x_B, x_D) = g \circ \Theta^{-1}(x_B,\ell_{x_B,x_D})\,.\]
Thus, $g\in \C^{d^{\mathrm{B}}}$ if and only if for all $(x_B, x_D) \in X \times X$
\[\min_{\delta_{x_B} \preceq_{\text{convex}} \eta} \int g\circ \Theta^{-1}(y,\ell_{x_B,x_D}) \eta ( \d y) =  g \circ \Theta^{-1}(x_B,\ell_{x_B,x_D})\,.\]
It is immediate that if for any $\ell \in \Lambda$, $x \mapsto g\circ\Theta^{-1}(x,\ell)$ is convex, then the above holds for all $(x_B, x_D) \in X \times X$. Conversely, suppose the above holds for all $(x_B, x_D) \in X \times X$. Suppose for contradiction that there exists some $\ell \in \Lambda$ such that $x \mapsto g\circ\Theta^{-1}(x,\ell)$ is not convex. Then, there exists some $x_B \in X$ such that 
\[\min_{\delta_{x_B} \preceq_{\text{convex}} \eta} \int g\circ \Theta^{-1}(y, \ell) \eta ( \d y) < g \circ \Theta^{-1}(x_B, \ell)\,.\]
By definition of $\Lambda$, there exists some $(x'_B, x'_D) \in X \times X$ such that $\ell = \ell_{x'_B, x'_D}$. Now, consider 
\[x_D := \hat{d}^{\mathrm{B}}(x'_B, x'_D, x_B)\,.\]
Note that, by divisibility of $\hat{d}^{\mathrm{B}}$, for all $y \in X$, we have 
\[ \ell_{x_B, x_D}(y) = \hat{d}^{\mathrm{B}}(x_B, \hat{d}^{\mathrm{B}}(x'_B, x'_D, x_B), y) = \hat{d}^{\mathrm{B}}(x'_B, x'_D, y) =    \ell_{x'_B, x'_D}(y) = \ell(y)\,.\]
Thus, $\ell_{x_B, x_D} = \ell$. Hence, we have 
\[\min_{\delta_{x_B} \preceq_{\text{convex}} \eta} \int g\circ \Theta^{-1}(y, \ell_{x_B, x_D}) \eta ( \d y) < g \circ \Theta^{-1}(x_B, \ell_{x_B, x_D})\,,\]
a contradiction. This proves that 
\[
g \in \C^{d^{\mathrm{B}}} \iff g \circ \Theta^{-1} \in \mathcal{U}\,.
\]

Since $\Theta$ is a homeomorphism, to show that any $g \in \C^{d^{\mathrm{B}}}$ can be approximated by continuous functions in $\C^{d^{\mathrm{B}}}$ from below, it suffices to show that this holds for all $u \in \mathcal{U}$. Recall that $u \in \mathcal{U}$ if and only if 
\[u: X \times \Lambda \rightarrow \R\]
is bounded lower semicontinuous and satisfies that $x \mapsto u(x, \ell)$ is convex for all $\ell$. Since $X \subseteq \R^n$ is a compact, convex, metrizable subset of a locally convex Hausdorff topological
vector space, and since $\Lambda$ is a metric space, by \Cref{lem:piecewiseconvexapprox}, any $u \in \mathcal{U}$ can be approximated from below by a sequence of continuous functions $\{u_n\} \subseteq \mathcal{U}$, as desired. \\
\mbox{}

\noindent $ (\implies) $: Suppose that $\preceq^{\cQ^d}$ is a Blackwell-consistent order. Then, by \Cref{thm:main-compare}, there exists a Blackwell-consistent regular pair $(\widetilde{\C}^d, \widetilde{\cQ}^d)$ such that $\preceq^{\cQ^d}=\preceq_{\widetilde{\C}^d}=\preceq^{\widetilde{\cQ}^d}$ where $\widetilde{Q}^d$ is composition-closed. Moreover, by the uniqueness part of \Cref{thm:main-compare}, it follows that $\cQ^d=\widetilde{\cQ}^d$ and thus $\cQ^d$ is composition-closed. 

Now, suppose for contradiction that $d$ is not divisible. Then, there exist $x,y,z \in X^0$, such that $d_{\mathrm{II}}(x,z;y)\neq d(x,z)$. Since $d(x',x')=x'$ for all $x' \in X$, it must be that $x\neq y$ and $y \neq z$. Moreover, since $x,y,z \in X^0$, there exist $\varepsilon, \xi, \zeta>0$ such that 
\[
y^{x,\varepsilon}:=\frac{x-\varepsilon y}{1-\varepsilon}\,, \quad \quad z^{x,\xi}:=\frac{x-\xi z}{1-\xi}\,,\quad \quad \mbox{ and } z^{y,\zeta}:=\frac{y-\zeta z}{1-\zeta} 
\]
are in $X^0$, and that $\varepsilon y +(1-\varepsilon)y^{x,\varepsilon}=x$, $\xi z+(1-\xi)z^{x,\xi}=x$, and $\zeta z+(1-\zeta)z^{y,\zeta}=y$. 

Now, let  
\[
P_1(\,\cdot\,\mid \tilde{x}):=\begin{cases}
\frac{1}{2}\left(\varepsilon \delta_{y}+(1-\varepsilon)\delta_{y^{x,\varepsilon}}\right)+\frac{1}{2} \left(\xi \delta_{z}+(1-\xi)\delta_{z^{x,\xi}} \right),&\mbox{if } \tilde{x}=x\\
\delta_{\tilde{x}},&\mbox{otherwise }
\end{cases}\,,
\]
and 
\[
P_2(\,\cdot\,\mid \tilde{x}):=\begin{cases}
\zeta \delta_{z}+(1-\zeta)\delta_{z^{y,\zeta}},&\mbox{if } \tilde{x}=y\\
\delta_{\tilde{x}},&\mbox{otherwise }
\end{cases}\,.
\]
By construction, $P_1,P_2 \in \Pc_M$. Then, let $Q_1:=Q^{P_1,d}$ and $Q_2:=Q^{P_2,d}$ where $Q^{P, d}$ is defined as the kernel that sends 
\[(x_B, x_D) \mapsto (Y, \hat{d}(x_B, x_D, Y))\]
where $Y \sim P(\,\cdot\,\mid x_B)$. By construction, $Q_1, Q_2 \in \cQ^d$.

Moreover, by construction, $Q_2 \circ Q_1(\,\cdot\,\mid x,x)$ must assign at least probability $\nicefrac{\varepsilon}{2}\cdot\zeta>0$ on $\{(z,d_{\mathrm{II}}(x,z;y))\}$ through the belief path of $x \mapsto y \mapsto z$. Since $d_{\mathrm{II}}(x,z;y)\neq d(x,z)$, we then have 
\[
Q_2 \circ Q_1\Big(\Big\{(\tilde{x}_B,\tilde{x}_D):\tilde{x}_D=d(x,\tilde{x}_B)\Big\}\mid x,x\Big)\leq 1-\frac{\varepsilon}{2}\cdot\zeta<1\,.
\]
However, since $\cQ^d$ is composition-closed, $Q_2 \circ Q_1 \in \cQ^d$. Hence, $Q_2 \circ Q_1(\,\cdot\,\mid x,x)$ must be the distribution of $(Z,d(x,Z))$, where $Z \sim P(\,\cdot\,\mid x)$ for some $P \in \Pc_M$. In particular, it must be that 
\[
Q_2 \circ Q_1\Big(\Big\{(\tilde{x}_B,\tilde{x}_D):\tilde{x}_D=d(x,\tilde{x}_B)\Big\}\mid x,x\Big)=1\,,
\]
a contradiction. Thus, $d$ must be divisible. \hfill \qedsymbol

\subsection{Proof of \texorpdfstring{\Cref{prop:ambiguity}}{Proposition 6}}
$(i)$: Consider any ambiguity-averse preference $\preceq $ defined by representation $\mathcal{V}$. Let $ \widetilde{\preceq} $ and $\widetilde{V}$ be the implied preference relation over lotteries, and its representation, respectively. That is
\[
\mu_1  \widetilde{\preceq}  \mu_2 \iff A_\C(\mu_1) \preceq  A_\C(\mu_2)\,,
\]
for all $\mu_1, \mu_2 \in \Delta(X)$; 
and 
\[
\widetilde{V}(\mu):=\mathcal{V}(A_\C(\mu))\,,
\]
for all $\mu \in \Delta(X)$. 

We first prove the ``if'' part. Suppose that $\C$ is min-closed. By \Cref{thm:main}, for all $\mu \in \Delta(X)$, we have 
\begin{align*}
\widetilde{V}(\mu) &= \alpha \cdot \max_{\nu: \nu \preceq_\C \mu} \int u(x) \nu(\d x) - (1 - \alpha) \cdot \max_{\nu: \nu \preceq_\C \mu} \int -u(x) \nu(\d x)    \\
 &= \alpha \int \overline{u}(x) \mu(\d x) - (1 - \alpha) \cdot \int \overline{-u}(x) \mu(\d x) \\
  &= \int \underbrace{\Big[\alpha\cdot  \overline{u}(x) - (1 - \alpha)\cdot \overline{-u}(x)\Big]}_{\widehat{u}(x)}  \mu(\d x)\,,
\end{align*}
which is an affine function of $\mu$, and hence an expected utility representation.

For the ``only if'' part, suppose that $\mathcal{C}$ is not min-closed. By \Cref{thm:main}, there exists $f:X \to \R$ such that $V_f^\star$ is not affine. Note that $V^\star_f$ corresponds to a representation $\widetilde{V}$ when $u=f$ and $\alpha=1$. Since $V^\star_f$ is not affine, it is not an expected utility representation.  \\
\mbox{}

\noindent $(ii)$: Suppose that $\C$ is strongly non-min-closed. Let $f \in C(X)$ and $\mu_0 \in \Delta(X)$ that is absolutely continuous with a density bounded away from zero, be such that 
\[
\min_{g \in \C,\, g \geq f} \int g \d \mu_0
\]
has two distinct solutions $g_1,g_2$ that are not affinely related. We first show that for this $f$, $V_f^\star(\mu)$ cannot be written as an increasing transformation of an expected utility. That is, there does not exist a strictly increasing function $\phi:\R \to \R$ and a bounded upper semicontinuous $u:X \to \R$ such that 
\[
V_f^\star(\mu)=\phi\left(\int u \d \mu\right)
\]
for all $\mu \in \Delta(X)$. To this end, suppose the contrary and let $\phi$ and $u$ be such functions. We first claim that, for any supergradient $g$ of $V_f^\star$ at $\mu_0$, we have $g \in \mathrm{span}\{u,1\}$ $\mu_0$-a.e. Indeed, consider any finite signed measure $\eta \ll \mu_0$ whose Radon–Nikodym derivative $\d \eta / \d \mu_0$ is essentially bounded with $\eta(X)=\int u \d \eta=0$. Then, there exists $\varepsilon>0$ such that $\mu^\varepsilon_1:=\mu_0+\varepsilon \eta \in \Delta(X)$ and $\mu^\varepsilon_2:=\mu_0-\varepsilon \eta \in \Delta(X)$. Since $\int u \d \eta=0$, we have 
\[
\int u \d \mu^\varepsilon_1=\int u \d \mu_0= \int u \d \mu^\varepsilon_2\,,
\]
and hence, for each $i \in \{1,2\}$,
\[
V_f^\star(\mu^\varepsilon_i)=\phi\left(\int u \d \mu^\varepsilon_i\right)=\phi\left(\int u \d \mu_0\right)=V_f^\star(\mu_0)\,.
\]
Thus, since $g$ is a supergradient of $V_f^\star$ at $\mu_0$, 
\[
V_f^\star(\mu_0)=V_f^\star(\mu_i^\varepsilon)\leq V_f^\star(\mu_0)+\int g \d (\mu_i^\varepsilon-\mu_0)
\]
for $i \in \{1,2\}$. Therefore, 
\[
\int g \d (\mu_1^\varepsilon-\mu_0)=\varepsilon \int g \d \eta \geq 0 
\]
and 
\[
\int g \d (\mu_2^\varepsilon-\mu_0)=-\varepsilon \int g \d \eta \geq 0\,.
\]
Together, we have $\int g \d \eta=0$ for all supergradient $g$ of $V_f^\star$ at $\mu_0$ and for all such $\eta$ satisfying $\int 1 \d \eta=\int u\d\eta=0 $. Now, consider the linear map $T$ that maps from 
\[M^\infty(\mu_0):= \Bigg\{\eta \in M(X): \eta \ll \mu_0,\,\, \frac{\d \eta}{\d \mu_0} \in L^\infty(\mu_0) \Bigg\}\]
to $\R^2$, defined as 
\[
T(\eta):=\left(\int 1 \d \eta, \int u \d \eta\right)\,,
\]
and also consider the linear functional 
\[
\psi_g(\eta):= \int g \d \eta\,.
\] 
We then have that $\mathrm{ker}(T) \subseteq \mathrm{ker}(\psi_g)$. It follows that $\psi_g = l \circ T$ for some linear $l$. Therefore, there exist $a,b \in \R$ such that for all $\eta \in M^\infty(\mu_0)$, 
\[
\int g \d \eta=\langle (a,b), T(\eta) \rangle=\int (bu+a) \d \eta\,. 
\]
Thus, we have $g=bu+a$, and hence $g \in \mathrm{span}\{u,1\}$, $\mu_0$-a.e. 

In the meantime, by the same argument as the proof of \Cref{thm:main} $(a) \implies (b)$, the duality gap is zero, and thus 
\[
V_f^\star(\mu)=\min_{g \in \C,\, g \geq f} \int g \d \mu
\]
for all $\mu \in \Delta(X)$. In particular, 
\[
\int g_1 \d\mu_0=\int g_2 \d \mu_0=V_f^\star(\mu_0)\,.
\]
Thus, for any $\mu \in \Delta(X)$, and for each $i \in \{1,2\}$, since 
\[
\int f \d \nu \leq \int g_i \d \nu \leq \int g_i \d \mu
\]
for all $\nu \preceq_\C \mu$, it must be that for all $\mu \in \Delta(X)$ and for $i \in \{1,2\}$,
\[
V_f^\star(\mu) \leq \int g_i \d\mu=V_f^\star(\mu_0)+\int g_i \d (\mu-\mu_0)\,. 
\]
That is, both $g_1$ and $g_2$ are supergradients of $V_f^\star$ at $\mu_0$. Therefore, by the previous argument, we have that $g_1, g_2 \in \mathrm{span}\{u,1\}$ $\mu_0$-a.e., which implies that $g_1, g_2$ are affinely related since $g_1$ and $g_2$ are continuous functions and $\mu_0$ has full support---contradicting the assumption that $g_1$ and $g_2$ are not affinely related. 

Now, consider the preference $\widetilde{\preceq}$ defined by 
\[
\mu_1 \widetilde{\preceq} \mu_2 \iff V_f^\star(\mu_1) \leq V_f^\star(\mu_2)\,,
\]
for all $\mu_1,\mu_2 \in \Delta(X)$. Then, since $V^\star_f$ corresponds to a representation $\widetilde{V}$ when $u=f$ and $\alpha=1$. Since $V^\star_f$ cannot be written as a monotone transformation of an affine function, it does not represent an expected utility preference \hfill \qedsymbol

\setlength\bibsep{2pt}
\bibliographystyle{ecta} 
\bibliography{references}

\end{document}